\documentclass[12pt]{article}
\usepackage{amsmath}
\usepackage{graphicx,psfrag,epsf}
\graphicspath{{figures/}}
\usepackage{enumerate}
\usepackage{url} 
\usepackage{float}
\newcommand{\blind}{0}

\usepackage{algorithm,algorithmic}
\usepackage[algo2e]{algorithm2e}
\usepackage{multirow}
\usepackage{lipsum}
\usepackage{amsfonts}
\usepackage{cite}      
                  
\usepackage{graphicx}   

\usepackage{rotating}
\usepackage{amsmath}
\usepackage{amsthm}
\usepackage{amsfonts}
\usepackage{amssymb}
\usepackage{graphicx}
\usepackage{subcaption}
\usepackage[skip=0pt]{caption}
\usepackage{transparent}
\usepackage{multirow}
\usepackage{color}
\usepackage{blkarray}
\usepackage[round]{natbib}
\usepackage{booktabs}
\usepackage[dvipsnames]{xcolor}
\usepackage{xr}
\usepackage[]{hyperref}
\hypersetup{colorlinks=true, 
	linkcolor=red,       
    citecolor=blue,        
    filecolor=magenta,      
    urlcolor=cyan           
}  

\newcommand{\bm}[1]{\boldsymbol{#1}}
\newcommand{\E}{\mathbb{E}}

\renewcommand{\P}{\mathbb{P}}

\newcommand{\U}{\bm{U}}

\newcommand{\wt}[1]{\widetilde{#1}}

\newtheorem{prop}{Proposition}

\newcommand{\Real}{\mathbb R}

\newcommand{\Var}{\mathsf{Var\,}}

\newcommand{\wh}[1]{\smash{\widehat{#1}}}

\usepackage{array}
\newcolumntype{C}{@{\extracolsep{0.5in}}c@{\extracolsep{0pt}}}

\def\C {\,|\:}

\def\C {\,|\:}

\newtheorem{definition}{Definition}[section]
\newtheorem{lemma}{Lemma}[section]

\newtheorem{theorem}{Theorem}[section]

\newtheorem{remark}{Remark}[section]
\newtheorem{corollary}{Corollary}[section]

 \theoremstyle{assumption}

\usepackage{url} 
\usepackage{amsthm}
\usepackage{amssymb}
\theoremstyle{definition}
\graphicspath{{figures/}}
\usepackage{mathtools}
\usepackage{xr}
\def\C {\,|\:}
\def\db {\bar d_{B_1(t)}}
\def\bm {\boldsymbol}
\def\dbb {\bar d_{B_2(t)}}

\def\wh {\widehat}
\def\hsn {\widehat s_n}
\def\tsn {\widetilde{s_n}}
\def\P{\mathbb{P}}
\def\O{\mathcal{O}}
\def\Var {\text{Var}}
\def\wt {\widetilde}

\def\firstmu {(\bm \mu^*)_{\wh \tau-}}
\def\secondmu {(\bm \mu^*)_{\wh\tau+}}
\def\firstempmu {(\bm \mu)_{\wh \tau-}}
\def\secondempmu {(\bm \mu)_{\wh\tau+}}
\def\firstempep {(\bm \epsilon)_{\wh\tau-}}
\def\secondempep {(\bm \epsilon)_{\wh\tau+}}
\def\firstep {(\bm \epsilon)_{\tau^*-}}
\def\diff {\text{diff}}
\def\secondep {(\bm \epsilon)_{\tau^*+}}
\theoremstyle{definition}

\begin{document}

\def\spacingset#1{\renewcommand{\baselinestretch}%
{#1}\small\normalsize} \spacingset{1}


\if0\blind
{
	\title{\sf Weighted-Graph-Based Change Point Detection}
	\author{
		Lizhen Nie\,\, and    Dan L. Nicolae \footnote{
			{\sl  \small Department of  Statistics, University of Chicago}}
	} 
	\maketitle
} \fi

\if1\blind
{
	\bigskip
	\bigskip
	\bigskip
	\begin{center}
		{\LARGE\sf Weighted-Graph-Based Change Point Detection}
	\end{center}
	\medskip
} \fi

\bigskip
\begin{abstract}
We consider the detection and localization of change points in the distribution of an offline sequence of observations. Based on a nonparametric framework that uses a similarity graph among observations, we propose new test statistics when at most one change point occurs and generalize them to multiple change points settings. The proposed statistics leverage edge weight information in the graphs, exhibiting substantial improvements in testing power and localization accuracy in simulations. We derive the null limiting distribution, provide accurate analytic approximations to control type I error, and establish theoretical guarantees on the power consistency under contiguous alternatives for the one change point setting, as well as the minimax localization rate. In the multiple change points setting, the asymptotic correctness of the number and location of change points are also guaranteed. The methods are illustrated on the MIT proximity network data.
\end{abstract}

\noindent%
\vfill

\newpage

\spacingset{1.5} 

\section{Introduction}

The task of change point detection (CPD) is to identify possible changes in the distribution of a time-ordered sequence. Classical change point detection methods assume parametric models or are focused on univariate settings. Novel approaches are needed due to the increasing richness of high-dimensional and non-Euclidean data in various scientific applications. For instance, identifying changes in genetic networks involves graphical encoding of data \citep{li2011large, lu2011high}. In financial modeling, segmentation of historical data involves multidimensional correlated assets \citep{talih2005structural}. 

The change point problem is usually decomposed into two stages: detection (testing for whether the distribution changes); and localization (estimation of the location of change point(s) when detected). Thus, an ideal change point method should have (1) high power for detection, and (2) high accuracy in localization. 

Recently, several nonparametric change point methods have been proposed. They can be classified into three main categories: kernel-based \citep{desobry2005online, harchaoui2009kernel, harchaoui2009regularized, huang2014high, li2015m, garreau2018consistent, celisse2018new, chang2019kernel, arlot2019kernel}, Euclidean-distance-based \citep{matteson2014nonparametric}, and graph-based \citep{chen2015graph, chen2017new, chen2018weighted, chu2018sequential, chu2019asymptotic, chen2019change, chen2019sequential,liu2020fast, song2020asymptotic}. 
Kernel-based methods are applicable to any type of data, but existing ones either do not offer false positive controls \citep{desobry2005online,garreau2018consistent,celisse2018new,arlot2019kernel}, or do not provide guarantees on localization accuracy \citep{desobry2005online,harchaoui2009kernel,harchaoui2009regularized,huang2014high,li2015m,chang2019kernel}. 
Euclidean-distance-based methods \citep{matteson2014nonparametric} provide localization consistency, but they are not applicable to non-Euclidean data.
Graph-based methods \cite{chen2015graph} are built on a binary similarity graph among observations. They are applicable to any type of data and yield analytic formulas for controlling type I error. However, they do not provide theoretical guarantees on testing power or localization consistency, and using binary graphs leads to information loss. 

Our work is based on the graph-based change point detection framework but utilizes the weight information in the graph.
We consider both settings where at most one change point (AMOC) exists or multiple change points are possible. Starting with the AMOC setting, we propose new statistics and give explicit formulas for type I error control. Further, we show that the proposed tests are consistent under local alternatives and the estimated change point has the minimax localization rate. A generalized algorithm for multiple change points setting is also proposed, and is guaranteed to identify both the correct number and locations of change points. Our framework unifies the kernel-based, Euclidean-distance-based and graph-based CPD methods, leading to a general CUSUM-type \cite{page1954continuous} decomposition of the proposed statistics which holds for any distance and any data under mild assumptions, and providing a better understanding of properties of the proposed statistics.

This paper is structured as follows: Section \ref{sec:setup} introduces problem setting and background, Section \ref{sec:meth} proposes new statistics and discusses its connection to previous methods, Section \ref{sec:asymptotics} presents the asymptotic theoretical results, Section \ref{sec:simulation} shows simulation results, Section \ref{sec:real_data} shows results on a real data example, Section \ref{sec:discuss} gives discussion and conclusions.

\section{Preliminary Setups}\label{sec:setup}

\subsection{Problem Setting}
Suppose we observe an independent, time-ordered sequence $\{y_i\}_{i=1}^n$. Depending on the number of change points, 
we consider two settings in increasing complexity. 

\paragraph{At most one change point (AMOC)}
In the simplest setting, there exists at most one change point. Suppose $F_0\neq F_1$ are probability distributions on the space $y_i$'s take values. We are concerned with the following problems:

\,\,1. (Detection) Testing the null hypothesis
	$$H_0: y_i\sim F_0,\,i=1,2,\cdots,n$$
	against the single change-point alternative
	$$H_A: \exists \rho^*\in(0, 1)\,\,\text{s.t.}
	\begin{cases}
	y_1,\cdots,y_{\tau^*}\sim F_0\\
	y_{\tau^*+1},\cdots,y_n\sim F_1,
	\end{cases}
	$$
	where $\tau^*=\lceil n\rho^*\rceil$. Here $\lceil x \rceil$ denotes the least integer no less than $x$.
	
\,\,2. (Localization) When rejecting $H_0$, obtain an estimator $\hat \tau$ of the true change point location.

\paragraph{Multiple change points}
In this setting, there is a fixed but unknown number $K$  of change points that partition the whole sequence into $K+1$ phases. The change points are $\mathcal{D}=\{\tau_1^*,\tau_2^*,\cdots,\tau_K^*\}$ where $0<\rho_1^*<\rho_2^*<\cdots<\rho_K^*<1$, $\tau_k^*=\lceil n\rho_k^*\rceil$. Suppose 
$$y_{\tau_k^*+1},y_{\tau_k^*+2},\cdots,y_{\tau_{k+1}^*}\sim F_k,$$ 
where $F_k\neq F_{k+1}$ for all $k=0,1,\cdots,K$, $\tau_0^*:=0$ and $\tau_{K+1}^*:=n$. 
When there are no change points, $K=0$, $\mathcal{D}=\varnothing$.
Our task is to estimate $K$ as well as $\mathcal{D}$. 

\subsection{Graph-based Methods}\label{sec:recall_graph}
The cornerstone of this paper is the graph-based CPD framework \citep{chen2015graph,chen2017new,chen2018weighted,chu2019asymptotic}. This section introduces important ideas and quantities behind them.

Graph-based CPD methods focus on the AMOC setting. They are based on binary similarity graphs where nodes represent observations and edges similarity. A binary similarity graph is usually constructed from a weighted one via a minimum spanning tree (MST) \citep{friedman1979multivariate}, minimum distance pairing \citep{rosenbaum2005exact} or nearest neighbor \citep{henze1988multivariate}. For each $t$,  denote the count of edges among $\{y_i\}_{i=1}^t$ (phase I) by $C_{B_1(t)}$ and that among $\{y_i\}_{i=t+1}^n$ (phase II) by $C_{B_2(t)}$. Utilizing $C_{B_1(t)}$ and $C_{B_2(t)}$, various scan statistics have been proposed, which are found to be combinations of two statistics:
\begin{align}
&Z_w(t)=\text{Strd}\left(\frac{n-t-1}{n-2}C_{B_1(t)}+\frac{t-1}{n-2}C_{B_2(t)}\right),   \label{eq:Zw}
\\&Z_{\diff}(t)=\text{Strd}\left(C_{B_1(t)}-C_{B_2(t)}\right),   \label{eq:Zdiff}
\end{align}
where $\text{Strd}$ denotes a standardized statistic such that it has the same variance and mean across $t$. For example, the generalized edge-count two-sample test statistic $G$ in \citet{chen2017new} can be written as
$$
G:=\max_t G(t)\quad\text{where}\quad G(t):=\left(Z_w(t)\right)^2+\left(Z_{\diff}(t)\right)^2,
$$
and the max-type edge-count test statistic $M$ from \citet{chu2019asymptotic} is defined as 
$$M:=\max_t M(t)\quad\text{where}\quad 
M(t):=\max \left(Z_w(t),|Z_{\diff}(t)|\right).$$
A change point is detected when $G$ (or $M$) exceeds a given threshold, and its estimated location is defined as $\hat\tau=\arg\max_t G(t)$ (or $\hat\tau=\arg\max_t M(t)$).
 
Further, \citet{chu2019asymptotic} found that $Z_w(t)$ works well for detecting mean changes, and $Z_{\diff}(t)$ for detecting scale changes. Some intuition: for mean change, observations from the same distribution are similar to each other, and thus edges are more
likely to form among them. In this case the true change point $\tau^*$ is the time $t$ which maximizes the number of edges within $\{y_i\}_{i=1}^t$, $C_{B_1(t)}$, and within $\{y_i\}_{i=t+1}^n$, $C_{B_2(t)}$.  The weights in $Z_w(t)$ balance the influence from the unequal sample size of $\{y_i\}_{i=1}^t$ and $\{y_i\}_{i=t+1}^n$.   For scale changes, edges are more likely to form among observations from the distribution with smaller dispersion. There the true change point $\tau^*$ is the time $t$ maximizing the difference between number of edges within $\{y_i\}_{i=1}^t$ and that within $\{y_i\}_{i=t+1}^n$, i.e., $\left|C_{B_1(t)}-C_{B_2(t)}\right|$, which is $\left|Z_{\diff}(t)\right|$.

\subsection{Notations}
We denote $[n]=\{1,2,\cdots,n\}$, $A(t)=\{(i,j)\in [n]^2:i\leq t, j>t\}$, $B_1(t)=\{(i,j)\in [n]^2:i\leq t, j\leq t, i\neq j\}$ and $B_2(t)=\{(i,j)\in [n]^2:i> t, j> t,i\neq j\}$. Similarly, we denote $A^{\,l,r}(t)=\{(i,j)\in [n]^2:l\leq i\leq t, t<j\leq r\}$, $B_1^{\,l,r}(t)=\{(i,j)\in [n]^2:l\leq i\leq t, l\leq j\leq t, i\neq j\}$ and $B_2^{\,l,r}(t)=\{(i,j)\in [n]^2:t<i\leq r, t<j\leq r,i\neq j\}$.
For any set $D\subseteq [n]^2$, we denote $d_{D}=\sum_{(i,j)\in D}d(y_i,y_j)$ and $\bar d_{D}=\sum_{(i,j)\in D}d(y_i,y_j)/|D|$ where $|D|$ is the cardinality of $D$. We write $\bar y= \frac{1}{n}\sum_{i=1}^ny_i$. For any function $f$, we write $\bar f(y) _{t-}= \frac{1}{t}\sum_{i=1}^tf(y_i)$, $\bar f(y) _{t+}= \frac{1}{n-t}\sum_{i=t+1}^nf(y_i)$. Denote ${V(f(y),\|\cdot\|)}_{t-}=\frac{1}{t-1}\sum_{i=1}^t \|f(y_i)-\bar f(y)_{t-}\|^2$ and ${V(f(y),\|\cdot\|)}_{t+}=\frac{1}{n-t-1}\sum_{i=t+1}^n \|f(y_i)-\bar f(y)_{t+}\|^2$, which can be viewed as the estimated dispersion of $f(y)$ in phase I or II measured in $\|\cdot\|$. We denote $\xrightarrow{w}$ as weak convergence,  $\O_p$ stochastic as boundedness, and $o_p$  as convergence in probability. Denote $W^0$ as Brownian bridge.

\section{Proposed statistics}
\label{sec:meth}

 Using unweighted graphs ultimately leads to information loss, motivating the development of weighted-graph-based change point detection methods. Suppose we have an weighted undirected graph on $\{y_i\}_{i=1}^n$ where an edge between $y_i$ and $y_j$ comes with a weight $d(y_i,y_j)$ and  $d$ is a given distance. For each $t$, the $n(n-1)/2$ pairs of weights can be split into three parts: ${d_{B_1(t)}}$ which corresponds to sum of distances within $\{y_i\}_{i=1}^t$, ${d_{B_2(t)}}$ which corresponds to sum of distances within $\{y_i\}_{i=t+1}^n$, and ${d_{A(t)}}$ which corresponds to sum of distances between $\{y_i\}_{i=1}^t$ and $\{y_i\}_{i=t+1}^n$. A weighted-graph-based test statistic is of the form
\begin{equation}\label{eq:stat_general}
	\max_t \left(\text{Strd}\left(w_0(t)d_{A(t)}+w_1(t)d_{B_1(t)}+w_2(t)d_{B_2(t)}\right)\right)
\end{equation}
\normalsize
where different weights $w_0(t),\,w_1(t),\,w_2(t)$ lead to different statistics. Optimal weights depend on the distribution of data under $H_A$, and we propose new statistics motivated by $Z_w(t)$ and $Z_{\diff}(t)$ and extend them to the multiple change points setting. Furthermore, in Section \ref{sec:connections}, we draw a connection between the proposed statistics and the existing nonparametric change point detection methdos, allowing us to decompose the proposed statistics into a CUSUM form and to better understand their properties.

\subsection{At Most One Change Point}\label{sec:stat_def}
We will introduce new statistics based on \eqref{eq:stat_general} where selection of $w_0(t)$, $w_1(t)$, $w_2(t)$ are motivated by $Z_w$ and $Z_{\diff}$. 

Note that maximizing $Z_w(t)$ is equivalent to maximizing within-phase similarity, which is also equivalent to minimizing within-phase distance, i.e., $d_{B_1(t)}$ and $d_{B_2(t)}$.  In addition,  the between-phase distance, i.e., $d_{A(t)}$, should be maximized. This suggests choosing $w_0(t)>0,\,w_1(t)<0,\,w_2(t)<0$ in \eqref{eq:stat_general}. Under $H_0$, we expect $\bar d_{A(t)}\approx \bar d_{B_1(t)}\approx \bar d_{B_2(t)}$, leading to the following statistic:
$$T_1(t):=\bar d_{A(t)}-\frac{1}{2}\bar d_{B_1(t)}-\frac{1}{2}\bar d_{B_2(t)}.$$ 
$T_1$ has mean zero under $H_0$;  stabilizing the variance across $t$ leads to the following test statistic
\begin{equation}\label{dfn:S1}
S_1:=\max_{n_0\leq t\leq n_1} \frac{t(n-t)}{n}T_1(t),
\end{equation}
where $n_0,n_1$ are pre-specified constraints for $\tau^*$ s.t. $n_0=\lceil n\rho_0\rceil, n_1=\lceil n\rho_1\rceil$ and $0<\rho_0<\rho^*<\rho_1<1$.  

The intuition behind $Z_{\diff}(t)$ is to maximize difference in within-phase similarities, which is equivalent to maximizing ${d_{B_1(t)}}-{d_{B_2(t)}}$. It suggests choosing $w_1(t)>0, w_2(t)<0,w_0(t)=0$ in \eqref{eq:stat_general}. Under $H_0$, we expect that $\bar d_{B_1(t)}\approx\bar d_{B_2(t)}$ suggesting a test based on 
$$T_2(t)=\left|\bar d_{B_1(t)}-\bar d_{B_2(t)}\right|.$$ 
Standardization leads to 
\begin{equation}\label{dfn:S2}
	S_2:=\max_{n_0\leq t\leq n_1} \frac{1}{2\hsn}\sqrt{\frac{t(n-t)}{n}}T_2(t),
\end{equation}
where 
the scaling factor $\widehat s_n^2=\frac{1}{n}\sum_{i=1}^n (\bar d_i)^2-(\bar d)^2,$
$\bar d_i = \frac{1}{n}\sum_jd(y_i,y_j)$, $\bar d=\frac{1}{n^2}\sum_{i,j=1}^{n}d(y_i,y_j)$.  
Roughly speaking, $\hsn^2$ measures the variance in $\bar d_i$, and we refer the reader to the example below and Section \ref{sec:connections} for more details.

$S_1$ and $S_2$ are scan statistics taking the maximum of a standardized score function across all possible change points.  
Rejection thresholds for $S_1$ and $S_2$ depend on the distance measure and the unknown null distribution of the data.  The thresholds can be estimated using empirical samples as discussed in Section \ref{sec:pval_computation}, where we also give analytic formulas for controlling type I error. When the test statistic is significant, the estimated change point $\hat\tau$ is defined as the time $t$ where the maximum is taken. 

The example below gives intuition on the mathematical decomposition underlying $S_1$ and $S_2$. 

\paragraph{An Illustrating Example}
Let $\{y_i\}_{i=1}^n$ be univariate normal, $F_0=N(0,\sigma^2)$, and $d(y_i,y_j)=(y_i-y_j)^2$. Then,  
\begin{equation}\label{eq:T_1_decomp}
	T_1(t)=(\bar y_{t-}-\bar y_{t+})^2+\O_p\left(n^{-1}\right),
\end{equation}
where the first term corresponds to the z-statistic in likelihood ratio test for testing mean differences. It is also the CUSUM statistic for testing changes in $\E y_i$. Under the null,  
$$nT_1(t)\xrightarrow{w} \sigma^2\left[\rho(1-\rho)\right]^{-1}\left(\chi_1^2-1\right),$$
when $t(n-t)/n^2\rightarrow \rho(1-\rho)$. 
Similarly,
\begin{equation}\label{eq:T_2_decomp}
T_2(t)= 2\left| {V(y,|\cdot|)}_{t-}-{V(y,|\cdot|)}_{t+}\right|
\end{equation}
measures the difference in estimated variance before and after $t$. It is the CUSUM of squares statistic for testing variance changes \citep{lee2003cusum}. Under the null, 
\begin{align}
\nonumber n^{1/2}\left(\bar d_{B_1}(t)-\bar d_{B_2}(t)\right)
\xrightarrow{w}&\,\,2\sigma^2[\rho(1-\rho)]^{-1/2}N(0,2).    
\end{align}
Thus, the scaling factor $\sqrt{t(n-t)/n}\approx \sqrt{n\rho(1-\rho)}$ in Equation \eqref{dfn:S2} cancels the effect of varying variances. And in this case we have $\widehat s_n^2=\frac{1}{n}\sum_{i=1}^{n}(y_i-\bar y)^4-[\frac{1}{n}\sum_{i=1}^{n}(y_i-\bar y)^2]^2$. Thus $\hsn$ converges in probability to $\sqrt{2}\sigma^2$. 

In this example, Equation \eqref{eq:T_1_decomp} and Equation \eqref{eq:T_2_decomp} match the empirical conclusion from \citet{chu2019asymptotic} that $Z_w$ (corresponding to $T_1$) is useful for detecting location changes while $Z_{\diff}$ (corresponding to $T_2$) is useful for detecting scale changes. In Section \ref{sec:connections} we will see a similar CUSUM decomposition for \emph{any} type of data and distance. 

\subsection{Multiple Change Points}
In this section we show a bisection based procedure which generalizes the proposed statistics $S_1$ and $S_2$ to the multiple change point setting. Similar to the AMOC setting, we assume that 
$
0<\rho_0\le\frac{\rho_{k+1}^*-\rho_k^*}{\rho_{k+2}^*-\rho_k^*}\le\rho_1<1
$
for any $k=1,2,\cdots,K-2$.
Based on this assumption, we define $S_1^{\,l,r}$, $T_1^{\,l,r}(t)$, $S_2^{\,l,r}$, $T_2^{\,l,r}(t),\hsn^{\,l,r}$ as counterparts to $S_1,T_1,S_2,T_2,\hsn$ on the subsequence $\{y_l,\cdots,y_r\}$: denote $l'=l+\lceil (r-l)\rho_0\rceil$, $r'=l+\lceil (r-l)\rho_1\rceil$, then
\begin{align}
S_1^{\,l,r}:=\max_{l'\leq t\leq r'} \frac{(t-l)(r-t)}{r-l}T_1^{\,l,r}(t),
\quad
\text{where}\quad
T_1^{\,l,r}(t):=\bar d_{A^{\,l,r}(t)}-\frac{1}{2}\bar  d_{B_1^{\,l,r}(t)}-\frac{1}{2}\bar d_{B_2^{\,l,r}(t)}\,\,,\label{dfn:S1_l_r}
\end{align}
and
\begin{align}
S_2^{\,l,r}:=\max_{l'\leq t\leq r'} \frac{1}{2\hsn^{l,r}}\sqrt{\frac{(t-l)(r-t)}{r-l}}T_2^{\,l,r}(t),
\label{dfn:S2_l_r}
\end{align}
where
\begin{align*}
&T_2^{\,l,r}(t):=\left|\bar d_{B_1^{\,l,r}(t)}-\bar d_{B_2^{\,l,r}(t)}\right|,
\\
&[\hsn^{l,r}]^2=\frac{1}{r-l}\sum_{i=l}^r \left(\frac{1}{r-l}\sum_{j=l}^rd(y_i,y_j)\right)^2-\left(\frac{1}{(r-l)^2}\sum_{i,j=l}^{r}d(y_i,y_j)\right)^2.
\end{align*}
Then at each bisection iteration, for the current subsequence $\{y_l,\cdots,y_r\}$, if we detect a change in distribution, we segment it at $k$ where
\begin{align}
k&=\arg\max_{l'\leq t\leq r'}\frac{(t-l)(r-t)}{r-l}T_1^{\,l,r}(t) \quad \text{for }S_1,\text{or}\label{eq:multiple_S1}
\\
k&=\arg\max_{l'\leq t\leq r'} \frac{1}{2\hsn^{\,l,r}}\sqrt{\frac{t(n-t)}{n}}T_2^{\,l,r}(t) \quad \text{for }S_2,\label{eq:multiple_S2}
\end{align}
where $l'=l+\lceil\rho_0(r-l)\rceil,r'=l+\lceil\rho_1(r-l)\rceil$.
The whole bisection procedure is summarized in Algorithm \ref{alg:multiple_CP}. Notice we do not allow segments with less than $n_{\min}$ observations, which is set a priori in order to stabilize the performance.  

\begin{algorithm}[tb]
   \caption{Weighted-Graph-Based CPD for multiple change points}
   \label{alg:multiple_CP}
\begin{algorithmic}
\STATE{\bfseries Input:} Significance level $\alpha$, minimum length  $n_{\min}$.      
\STATE{\bfseries Output:}    Set of detected change points $\hat{\mathcal{D}}=\texttt{BS(}1,n\texttt{)}$.        
\FUNCTION{\texttt{BS(}{$l$, $r$}\texttt{)}}
    \STATE Calculate the realization $s$ of random variable $S:=S_1^{\,l,r}$ or $S:=S_2^{\,l,r}$.
    \STATE Estimate the new change point $k$ by 
    Equation \eqref{eq:multiple_S1} or \eqref{eq:multiple_S2}.
        \IF{$\P\left(S\geq s\right)\leq \alpha$ (using formulas in Section \ref{sec:asymptotic_null}) and $k-l,r-k\geq n_{\min}$,}
	        \STATE Update $\hat{\mathcal{D}}\leftarrow\hat{\mathcal{D}}\cup \{k\}.$
	        \STATE Call $\texttt{BS(}l,k\texttt{)}.$
	        \STATE Call $\texttt{BS(}k+1,r\texttt{)}.$
        \ENDIF  
    \STATE{\bfseries return} $\hat{\mathcal{D}}$
\ENDFUNCTION
\end{algorithmic}
\end{algorithm}

\subsection{Connections to Other Change Point Methods}\label{sec:connections}

We demonstrate here that our statistics are also related to various existing nonparametric CPD methods. Further, we will derive a similar CUSUM representation as in Equation \eqref{eq:T_1_decomp} and \eqref{eq:T_2_decomp} for any distance and any type of data.

First, the Euclidean-distance-based method \citep{matteson2014nonparametric} is a special case of $S_1$ where $d$ is the square distance. 
And \citet{sejdinovic2013equivalence,celisse2018new} show that kernel-based and distance-based methods are essentially equivalent. In one direction, we can always define a kernel from a distance as long as it satisfies some mild conditions:
\begin{lemma}[Lemma 12 in \cite{sejdinovic2013equivalence}]\label{lemma:distance_induced_kernel} Define
\begin{equation}\label{eq:distance_induced_kernel}
k^{(y_0)}(y_i,y_j)=\frac{1}{2}\left[d(y_i,y_0)+d(y_j,y_0)-d(y_i,y_j)\right]
\end{equation}
as the distance-induced kernel induced by $d(\cdot,\cdot)$ and centered at $y_0$. Then $k^{(y_0)}$ is a valid kernel if and only if  $d$ is a semi-metric of negative type, i.e., if and only if $d$ satisfies:

(1) $d(y_i,y_j)=0$ if and only if $y_i=y_j$.

(2) $\forall y_i,y_j\in \mathcal{X}$, $d(y_i,y_j)=d(y_j,y_i)$.

(3) $\sum_{i,j=1}^n c_ic_jd(y_i,y_j)\leq 0,\forall n\geq 2, y_1,\cdots,y_n\in\mathcal{X}, c_1,\cdots,c_n\in\Real, \sum_ic_i=0$.
\end{lemma}
\begin{remark}
The notion of semi-metric of negative type encompasses a large collection of metric spaces, including $L_p$ spaces for $0<p\leq 2$. We refer the reader to Theorem 3.6 of \citet{meckes2013positive} for a list of examples of metrics spaces of negative type.
\end{remark}

In the other direction, Proposition 14 in \citet{sejdinovic2013equivalence} shows we can always define a distance from a kernel. Together with Lemma \ref{lemma:distance_induced_kernel}, it implies that kernel-based and distance-based methods are essentially equivalent. 

Using Lemma \ref{lemma:distance_induced_kernel}, we find 
$T_1(t)$
equals the empirical maximum mean discrepancy \cite{gretton2012kernel} which was proposed for two sample testing. In this sense, $S_1$ shares similar intuition as \citet{sinn2012detecting,celisse2018new}. Comparing with existing kernel methods, the proposed statistics are simple to compute, and have theoretical guarantees for both detection and localization. For example, \citet{sinn2012detecting,celisse2018new} do not offer false positive controls. Statistic in \citet{harchaoui2009kernel} is more complicated to compute and do not provide guarantees on localization.  \citet{li2015m} develop M-statistics which are computationally cheaper but require the availability of reference data.

Establishing the equivalence with kernel methods also allows us to use concepts in kernels to derive familiar CUSUM representations for a general $d$. 
For any $y_0$, we define a centered kernel:
\begin{equation}\label{eq:k_tilde}
\begin{split}
&\widetilde k(y_i,y_j)
=k^{(y_0)}(y_i,y_j)-\E_{y_l\sim F_0}k^{(y_0)}(y_i,y_l)
-\E_{y_l\sim F_0} k^{(y_0)}(y_l,y_j)+\E_{y_l,y_m\sim F_0} k^{(y_0)}(y_l,y_m).
\end{split}
\end{equation}
It is easy to show $\widetilde k$ does not depend on $y_0$ (see Appendix Proposition \ref{prop_appendix:widetilte}). 
For $\widetilde k$, we may write it in terms of eigenfunctions $\psi_l$ with respect to the probability measure $F_0$:
\begin{align}
&\widetilde k(y,y')=\sum_{l=1}^\infty \lambda_l\psi_l(y)\psi_l(y'),\quad\text{where}\label{eq:eigen_decomp}
\\
&\int \widetilde k(y,y')\psi_l(y)dF_0(y)=\lambda_l\psi_l(y'),\quad
\int \psi_l(y)\psi_{l'}(y)dF_0(y)=\delta_{l,l'}.\nonumber
\end{align}
Denote the feature map $\phi$ as
\begin{equation}\label{def:phi}
\phi(y)=(\lambda_1^{1/2}\psi_1(y), \lambda_2^{1/2}\psi_2(y),\cdots)^\top\in\mathcal{H},
\end{equation}
and $\langle\phi(y),\phi(y')\rangle_{\mathcal{H}}:=\sum_{l=1}^{\infty}\phi_l(y)\phi_l(y')=\widetilde k(y,y')$.


\normalsize
We investigate next the proposed statistics $S_1, S_2$. Utilizing previous results, we have
\begin{equation}\label{eq:T1_CUSUM}
T_1(t)= \left\|\bar \phi(y)_{t-}-\bar \phi(y)_{t+}\right\|_{\mathcal{H}}^2+\O_p\left(n^{-1}\right),
\end{equation}
where $\|\cdot\|_{\mathcal{H}}:=\langle\cdot,\cdot\rangle_{\mathcal{H}}^{1/2}$ is the norm in $\mathcal{H}$.
Intuitively $T_1(t)$ measures the difference between the average feature map between $\{y_i\}_{i=1}^t$ and $\{y_i\}_{i=t+1}^n$. For univariate data and Euclidean distance, this reduces to the example in Section \ref{sec:stat_def}. Notice that 
expression \eqref{eq:T1_CUSUM} has a familiar form as the CUSUM statistic \citep{page1954continuous}. Indeed, if we directly observe $\{\phi(y_i)\}_{i=1}^n$, $S_1$ equals the CUSUM statistic for Hilbert space valued data \cite{tewes2017change}.
Similarly, we find
\begin{equation}\label{eq:T2_CUSUM}
T_2(t)= 2\left|{V(\phi(y),\|\cdot\|_{\mathcal{H}})}_{t-}-{V(\phi(y),\|\cdot\|_{\mathcal{H}})}_{t+}\right|,
\end{equation}
Intuitively, if $\E \phi(y_i)$ does not change (here $\E\phi(y_i)$ is the mean feature map defined such that $\langle \E\phi(y), \phi(y_j)\rangle_{\mathcal{H}}=\E _{y}k(y,y_j)$), $T_2(t)$ measures the difference between the magnitude of noise $\epsilon_i=\phi(y_i)-\E\phi(y_i)$ in terms of $\|\cdot\|_{\mathcal{H}}$. Again, expression \eqref{eq:T2_CUSUM} exhibits the form of CUSUM statistics and can be seen as a generalization of the statistic in \citet{lee2003cusum}, which is proposed for detecting changes in variance in time series models. Notice
$\hsn^2
=1/n\sum_{i=1}^n[\|\widehat \epsilon_i\|_{\mathcal{H}}^2-\sum_{i=1}^n{\|\widehat \epsilon_i\|_{\mathcal{H}}^2}/n]^2$
is the empirical estimator for the variance of $\|\epsilon\|_{\mathcal{H}}^2$, where $\widehat \epsilon_i=\phi(y_i)-\bar\phi(y)$ for $i=1,\cdots,n$.

The proposed method is also closely related to that of \citet{dubey2019fr} that developed a test statistic for detecting change point in Frechet mean and/or Frechet variance in the AMOC setting. Their statistics, before taking $\max$ with respect to $t$, is equivalent to (properly normalized) $4\widetilde T_1^2(t)+T_2^2(t)$, where $\widetilde T_1(t)$ is defined in Equation \eqref{eq:widetilde_T1} and can be seen as a non-centered version of $T_1(t)$. In this sense $S_1, S_2$ and that of \citet{dubey2019fr} are highly related. One difference is that we analyze the two components, $T_1(t)$ and $T_2(t)$, separately, and the combination of them follows automatically (see Section \ref{sec:combine_S1_S2}). An advantage of using $S_1$ (or $S_2$) individually is that it achieves higher power and localization accuracy under local alternatives when a specific type of change occurs, as suggested by our theory and simulations. Moreover, by expressing $S_1$ and $S_2$ through $d_{A}(t), d_{B_1}(t), d_{B_2}(t)$, we are free of solving the optimization problem in \citet{dubey2019fr}, granting the proposed statistics greater practicality. Finally, we also generalize our statistics to the multiple change points setting and prove next its asymptotic correctness.

\section{Asymptotics}\label{sec:asymptotics}
This section presents asymptotic guarantees for both AMOC and multiple change point setting.  
For $l=0,1,\cdots,K$, we define $\mu_l=\E_{Y\sim F_l}\phi(Y)$ and $\sigma^2_l=\E_{Y\sim F_l}\|\phi(Y)-\mu_l\|_{\mathcal{H}}^2$ with $\phi$ defined in \eqref{def:phi}.  Throughout this section we assume: 
\begin{align*}
    d \text{ is a semi-metric of negative type.}\quad \quad \text{(Assump 1)} 
\end{align*}
Other than subsection \ref{sec:asymptotic_null}, we also assume: 
\begin{align*}
    \exists M>0,\,\,\forall i\in\{1,2,\cdots,n\},\,\,\|\widetilde k(y_i,y_i)\|^2\leq M^2,\,\,\text{a.s.}
    \quad \quad \text{(Assump 2)} 
\end{align*}
Proofs of results in this section are included in the Appendix. 

\subsection{At Most One Change Point}
In AMOC setting, we provide guarantees for both detection (type I error, power) and localization (accuracy).

\subsubsection{Approximations to Significance Levels}\label{sec:asymptotic_null}
To control type I error, we focus on approximating tail distributions of $S_1,S_2$ under $H_0$. We investigate utilizing their asymptotic null distribution, and also propose an improvement based on higher order corrections.

\paragraph{Asymptotic null distribution of $S_1$ and $S_2$} 

Notice that 
$
\widetilde k(y,y')=[\E_yd(y,y')+\E_{y'}d(y,y')-d(y,y')-\E_{y,y'}d(y,y')]/2
$. Then we have:
\begin{theorem}[Asymptotic null]\label{thm:S1_null}\label{thm:S2_null}
	Under $H_0$, as $n\rightarrow\infty$, 
	
	(a) \textbf{For $\bm{S_1}$}: 
	there exists positive constant $\delta$ such that $\E_y|\widetilde k(y,y)|^{2+\delta}+\E_{y,y'}|\widetilde k(y,y')|^{2} <+\infty$, 
	\begin{equation}\label{eq:S1_asymp_null}
	S_1\xrightarrow{w} \max_{\rho_0\leq \rho\leq \rho_1}\frac{\sum_{l=1}^\infty\lambda_l\left(W_l^0(\rho)^2-\rho(1-\rho)\right)}{\rho(1-\rho)},    
	\end{equation}
	where $\lambda_l$'s are defined in Equation \eqref{eq:eigen_decomp}.

    (b) \textbf{For $\bm{S_2}$}: 
    there exists positive constant $\delta$ such that $\E_y|\widetilde k(y,y)-\E \widetilde k(y,y)|^{2+\delta}<+\infty$, 
	\begin{equation}\label{eq:S2_asymp_null}
	S_2\xrightarrow{w} 
	\max_{\rho_0\leq \rho\leq \rho_1}\left|\frac{1}{\sqrt{\rho(1-\rho)}}W^0(\rho)\right|.
	\end{equation}
\end{theorem}
\begin{remark}
	In (a), the boundedness of $\E_y|\widetilde k(y,y)|^{2+\delta}$ says $\E\left\|\phi(y)\right\|_{\mathcal{H}}^{4+2\delta}$ is finite so that the functional central limit theorem \cite{tewes2017change} holds. And the boundedness of $\E_{y,y'}|\widetilde k(y,y')|^{2}$ says $\int\widetilde k(y,y')^2dF_0(y)dF_0(y')<+\infty$, which ensures that eigen-decomposition \eqref{eq:eigen_decomp} holds and the right hand side of \eqref{eq:S1_asymp_null} is well-defined. 
	In (b), finite $\E_y|\widetilde k(y,y)-\E \widetilde k(y,y)|^{2+\delta}$ implies finite $\E|\left\|\phi(y)\right\|_{\mathcal{H}}^2-\E\left\|\phi(y)\right\|_{\mathcal{H}}^2|^{2+\delta}$, which guarantees that eigen-decomposition \eqref{eq:eigen_decomp} and the weak convergence to Brownian motion \citep{doukhan2012mixing} holds. 
\end{remark}


\paragraph{P-value approximation using asymptotic null distribution}\label{sec:pval_computation}
Type I error of the proposed tests can be controlled utilizing their asymptotic null distribution, where we can obtain approximate p-value via simulations. 
For $S_1$, we can estimate the eigenvalues $\lambda_1,\lambda_2,\cdots,\lambda_n$ by Equation (5) of \citet{gretton2009fast}, and then use simulations on the first $m\le n$ sums to approximate the infinite sum.

\begin{table}[t]
\vskip 0.15in
\begin{center}
\begin{small}
\scalebox{0.8}{
\begin{sc}
\begin{tabular}{lccccr}
\toprule
Setting & $Z_w$    & $M$     & $S_1$ ($\widetilde S_1$+correction) & $S_2$ ($\widetilde S_2$+correction) \\
\midrule
$N(0,1)$    & 0.06 & 0.1 & 0.07 (0.07) & 0.06 (0.06)  \\ 
$N(0, I_{10})$  & 0.05 &0.06 & 0.06 (0.06) & 0.02 (0.06)   \\ 
$N(0, I_{50})$   & 0.12 & 0.13 & 0.06 (0.05) & 0.15 (0.02)  \\
$N(0, I_{100})$   & 0.08 & 0.15 & 0.09 (0.07) & 0.45 (0.04)  \\
$t(4)$   & 0.06 & 0.09 & 0.06 (0.06) & 0.08 (0.08)  \\
$\text{Pois}(2)$   & 0.15 & 1 & 0.09 (0.09) & 0.04 (0.04)  \\
$\chi_1^2$   & 0.07 & 0.12 & 0.10 (0.11) & 0.04 (0.05)  \\
\bottomrule
\end{tabular}
\end{sc}
}
\end{small}
\end{center}
\vskip -0.1in
\caption{\small Comparison of coverage probability under $\alpha=0.05$. Set $n=200$. $Z_w, M$ \citep{chu2019asymptotic}  are using 1-MST and skewness correction in R package ``gSeg" \citep{chen2015gseg}. All distances are set to Euclidean distance. Null distribution for $S_1,\tilde S_1$ uses empirical eigenvalues. Results are based on 200 simulations. } 
\label{table:converage_probability}
\end{table}

\paragraph{P-value approximation using higher order corrections}
To investigate the accuracy of asymptotic approximations, we show coverage probability of $S_1, S_2$ under significance level $\alpha=0.05$ (Table \ref{table:converage_probability}). 
For comparison, we also show estimates for $Z_w, M$ \citep{chu2019asymptotic}. We observe that the approximation for $S_1$ works well,  while $S_2$ becomes anti-conservative when dimensionality becomes high. $Z_w$ and $M$ are anti-conservative for Poisson and high dimensional normal. 
A closer look into $S_2$ reveals the reason: under $H_0$, 
$$
\frac{1}{\widehat s_n}\sqrt{\frac{t(n-t)}{n}}T_2(t)
= N(0,1) + \O_p\left(n^{-1/2}\right),
$$
where the last term depends on (i) magnitude of $\E\left\|\epsilon\right\|^{2}_{\mathcal{H}}$, and (ii) skewness of $\|\epsilon\|^2_{\mathcal{H}}$. Here $\epsilon=y-\phi(y)$. For $F_0=N(0,I_{p})$ where $I_p\in\mathbb{R}^{p\times p}$ is the identity matrix, as $p$ becomes larger, $\E \|\epsilon\|^2$ increases and thus, the $\O_p\left(n^{-1/2}\right)$ term becomes increasingly in-negligible. To reduce the bias from $\E\|\epsilon\|^2$, we suggest using $\widetilde S_2$ instead of $S_2$, where 
\begin{equation}\label{eq:S2_tilde}
\widetilde S_2=\max_{n_0\leq t\leq n_1} \frac{1}{2\hat s_n}\sqrt{\frac{t(n-t)}{n}}\widetilde T_2(t),
\end{equation}
with $\widetilde T_2$ defined as
\begin{align}\label{eq_appendix:tilde_T_2}
\widetilde T_2(t)&=\left|\bar d_{B_1(t)}-\bar d_{B_2(t)}-\frac{2\widehat{\E\|\epsilon\|^2}}{\sqrt{n\rho(1-\rho)}} \left(\frac{2t}{n}-1\right)\right|,\quad  \text{where}
\quad
\widehat{\E\|\epsilon\|^2}&=\frac{1}{2n^2}\sum_{i\neq j, i,j=1}^nd(y_i,y_j).
\end{align}
To further alleviate the bias from the skewness of $\|\epsilon\|_{\mathcal{H}}^2$, following \citet{chen2015graph}, we propose a skewness correction: 
\begin{equation}\label{eq:S2_higher_order_correctness}
\begin{split}
\P\left(\widetilde S_2\geq x\right)
\approx\,\, &x\varphi(x)\int_{\rho_0}^{\rho_1}\left[1+\frac{Vx(x^2-3)}{6\sqrt{n}}\right]
\frac{1}{u(1-u)}\nu\left(\sqrt{\frac{x}{u(1-u)n}}\right)du,
\end{split}
\end{equation}
where $\varphi$ is the density function of standard normal, $\nu(\cdot)$ is defined as
\begin{align}\label{eq_appendix:nu}
 \nu(x)=\frac{(2/x)(\Phi(x/2)-0.5)}{(x/2)\Phi(x/2)+\varphi(x/2)},   
\end{align} 
where $\Phi$ are the cumulative distribution function for standard normal.
And $V$ is defined as
\begin{align}\label{eq_appendix:V}
V=\frac{1-2u}{\sqrt{u(1-u)}}\frac{m_6-3m_2m_4+2m_2^3}{\widehat s_n^3},    
\end{align} 
where $m_i$ is the sample $i$-th moment of $\|\phi(y)\|_{\mathcal{H}}$ under the null. 
The expressions for $m_2,m_4,m_6$ are
\begin{equation}\label{dfn:m2_m4_m6}
\begin{split}
& m_2=\frac{1}{2n^2}\sum_{i\neq j} d(y_i,y_j),
\\& m_4=\frac{1}{4n}\sum_{i=1}^n\left[\frac{2}{n}\sum_{j=1}^nd(y_i,y_j)-\frac{1}{n^2}\sum_{l,j=1}^nd(y_l,y_j)\right]^2, 
\\&
m_6=\frac{1}{8n}\sum_{i=1}^n\left[\frac{2}{n}\sum_{j=1}^nd(y_i,y_j)-\frac{1}{n^2}\sum_{l,j=1}^nd(y_l,y_j)\right]^3 .
\end{split}
\end{equation}
Coverage probabilities using $\widetilde S_2$ and \eqref{eq:S2_higher_order_correctness} are shown in brackets in Table \ref{table:converage_probability}. The gain of using higher order correction becomes significant as $\E\|\epsilon\|^2$ or the  skewness of $\|\epsilon\|^2$ becomes larger, and it is not only in terms of p-value calibration, but also in terms of localization. So we suggest using $\widetilde S_2$ instead of $S_2$, especially when data dimension is high. Notice that theoretical properties of $S_2$ hold also for $\widetilde S_2$ (Theorem \ref{thm:S2_null},  \ref{thm:S2_power}, \ref{thm:consistency_S2}).

For $S_1$, notice that  under $H_0$,
$$
\frac{t(n-t)}{n}T_1(t)=\sum\lambda_l(\chi_1^2-1)+\O_p\left(n^{-1/2}\right),
$$
where last term depends on (i) $\E d(y_i,y_j)$ and (ii) the skewness of each dimension of $\phi(y_i)$. Correcting for (ii) is complicated; and the terms resulting in (i) come from $d_{B_1},d_{B_2}$. We can change the weight before $d_{B_1},d_{B_2}$ to eliminate (i):
\begin{equation}\label{eq:widetilde_T1}
\widetilde S_1=\max_{n_0\leq t\leq n_1} \frac{t(n-t)}{n}\widetilde T_1(t),
\end{equation}
where $\widetilde{T_1}(t)=\frac{ d_{A(t)}}{t(n-t)}-\frac{d_{B_1(t)}}{2t^2}-\frac{ d_{B_2(t)}}{2(n-t)^2}.$
Recall 
${T_1}(t)=\frac{d_{A(t)}}{t(n-t)} -\frac{d_{B_1(t)}}{2t(t-1)} -\frac{d_{B_2(t)}}{2(n-t)(n-t-1)}$. Notice
$$\widetilde T_1(t)=\left\|\bar \phi(y)_{t-}-\bar \phi(y)_{t+}\right\|^2_{\mathcal{H}}\ge 0.$$
Power and localization consistency (Theorem \ref{thm:consistency_S1}, \ref{thm:S1_power}) also hold for $\widetilde S_1$, while its asymptotic null distribution becomes
\begin{equation*}
\widetilde S_1\xrightarrow{w} \max_{\rho_0\leq \rho\leq \rho_1}\frac{\sum_{l=1}^\infty\lambda_lW_l^0(\rho)^2}{\rho(1-\rho)},
\end{equation*}
where the right hand side is well-defined if $\sum_l\lambda_l<\infty$ or simply, when $d$ is bounded. Coverage probabilities using $\widetilde S_1$ are shown in brackets in Table 1. As the original $S_1$  already performs quite well, the gain of using $\widetilde S_1$ is not significant. 

\subsubsection{Power}

The following theorems demonstrate the power consistency of proposed tests under contiguous alternatives. 

\begin{theorem}[Power Consistency]\label{thm:S1_power}\label{thm:S2_power} We have

(a) \textbf{For $\bm{S_1}$}: If $\sqrt{n}\|\mu_0-\mu_1\|_{\mathcal{H}}\rightarrow \infty$, then
	\begin{equation*}
	\P_{H_A}\left(S_1>q^{(1)}_{\alpha}\right)\rightarrow 1,\quad n\rightarrow\infty
	\end{equation*}
	with $q^{(1)}_{\alpha}$ the upper $\alpha$-th quantile of asymptotic null of $S_1$.
	
(b) \textbf{For $\bm{S_2}$}: If $\mu_0=\mu_1$  and $\sqrt{n}|v_0-v_1|\rightarrow \infty$, then
	\begin{equation*}
	\P_{H_A}\left(S_2>q^{(2)}_{\alpha}\right)\rightarrow 1,\quad n\rightarrow\infty
	\end{equation*}
	with $q^{(2)}_{\alpha}$ the upper $\alpha$-th quantile of asymptotic null of $S_2$.
\end{theorem}
\begin{remark}
We show the power consistency of $S_1$ and $S_2$ even when the difference between $F_0$ and $F_1$ measured in terms of  $\|\mu_0-\mu_1\|_{\mathcal{H}}$ or $|v_0-v_1|$ shrinks to zero at a rate slower than $n^{-1/2}$.
Conclusion (2) shows that $S_2$ is a complement to $S_1$ in the sense that it is useful when change in distributions cannot be captured by $\E \phi(y)$.
\end{remark}

 Theorem \ref{thm:S1_power} guarantees that with more data, we will eventually detect the change point as long as it is captured by $\mu$ or $v$. We will see next that localization accuracy also depends on $\mu,v$. Here $\mu,v$ are determined by $d$ and thus, the choice of $d$ is critical for the performance of proposed statistics. However, the choice of proper distance is not the main focus of this article, and thus we simply assume that it is given.

\subsubsection{Minimax Localization Rate}
This section is concerned with the consistency of localization after detection of a change point.
\begin{theorem}[Localization consistency]\label{thm:consistency_S1}\label{thm:consistency_S2}
	We have
	
	(a) \textbf{For $\bm{S_1}$}: If $\mu_0,\,\mu_1$ are fixed and $\mu_0\neq \mu_1$, then for any $x>0$ and $n$, with probability at least $1-22e^{-x}$, we have
	\begin{equation*}
	\left|\frac{\widehat \tau-\tau^*}{n}\right|\leq \frac{C_1}{n}\left(\frac{M^2x}{\left\|\mu_0-\mu_1\right\|_{\mathcal{H}}^2}+\frac{M^4x}{\left\|\mu_0-\mu_1\right\|_{\mathcal{H}}^4}\right),
	\end{equation*}
	where $\widehat \tau$ is the estimated change point using statistics $S_1$ and $C_1$ is some constant depending on $\rho_0,\,\rho_1,\,\rho^*$.
	
	(b) \textbf{For $\bm{S_2}$}: If $\mu_0=\mu_1$, $v_0,\,v_1$ are fixed and $v_0\neq v_1$, then for any $x>0$, when sample size $n$ is sufficiently large, with probability at least $1-16e^{-x}$, we have
	\begin{equation*}
	\left|\frac{\widehat \tau-\tau^*}{n}\right| \leq \frac{ C_2}{n}\left(\frac{M^4x}{\left|v_0-v_1\right|^2}+\frac{M^2x}{\left|v_0-v_1\right|}\right),
	\end{equation*}
	where $\widehat \tau$ is the estimated change point using statistics $S_2$ and $C_2$ is some constant depending on $\rho_0,\,\rho_1,\,\rho^*$.
\end{theorem}

Theorem \ref{thm:consistency_S1} shows that under suitable $d$, the location of estimated change point is rate optimal \citep{brunel2014convex}.

\subsubsection{Combining $S_1, S_2$ for Unknown Type of Change}
\label{sec:combine_S1_S2}

To increase power, we could combine $S_1, S_2$ to detect unknown type of changes. However, differently from (unweighted) graph-based CPD methods where $Z_w,Z_{\diff}$ are always independent \citep{chu2019asymptotic}, $S_1, S_2$ are independent only under some restrictive conditions (Proposition \ref{prop:S1_S2_ind} in the Appendix).
Thus, combining $S_1, S_2$ in general case is much more complicated. Notice that the order of $T_1(t)$ and $T_2(t)$ are different. The work from \citet{dubey2019fr} motivates one simple way to combine $S_1$ and $S_2$:
\begin{align}\label{def:stat_S3}
    S_3 = \max_{n_0\leq t\leq n_1} \frac{1}{4\hsn^2}{\frac{t(n-t)}{n}}\left[4T_1^2(t) +T_2^2(t)\right].
\end{align}
As corollaries of previous theorems, we can get asymptotic null distribution, power, and localization consistency of $S_3$. The results are similar to those in \citet{dubey2019fr} and for brevity are included in the Appendix  \ref{sec_appendix:S3_guarantees}. Comparing Corollary \ref{thm:S3_power} or Theorem 3 in \citet{dubey2019fr} against Theorem \ref{thm:S1_power}, we find that when using $S_3$ instead of $S_1$, we are capable of idenfitying both changes in location or scale, but pay a price of increasing the order of magnitude of local alternatives $\|\mu_0-\mu_1\|_{\mathcal{H}}$ we can detect from $n^{-1/2}$ to $n^{-1/4}$.


\subsection{Multiple Change Points}
Theoretical guarantees of the proposed method in multiple change points setting are provided. 
We show that asymptotically, under some mild conditions, we can identify both the correct number and locations of change points: 
\begin{theorem}\label{thm:multiple_CP} In multiple change points setting, we have
    
    (a) \textbf{For $\bm{S_1}$}: If $\exists\,\, c_1>0$ s.t. $\|\mu_l-\mu_{l+1}\|_{\mathcal{H}}\geq c_1$ for any $l=0,1,\cdots,K-1$, Algorithm \ref{alg:multiple_CP} for $S_1$ yields: as $n\rightarrow\infty,\alpha\rightarrow0,n\alpha\rightarrow\infty$,
    	\begin{align}\label{eq:multiple_CP_consistency}
    	&\P\left(|\hat{\mathcal{D}}|=|\mathcal{D}|\right)\rightarrow 1,\quad \text{and}\quad
    	\forall\,\, k\in\hat{\mathcal{D}},\quad\min_{\tau\in\mathcal{D}}\left|k/n-\tau/n\right|=o_p(1).
    	\end{align}
    	
    (b) \textbf{For $\bm{S_2}$}: If $\mu_0=\mu_1=\cdots\mu_{K}$ and there exists some constant $c_2$ s.t. $\left|\sigma_l^2-\sigma_{l+1}^2\right|\geq c_2>0$ for any $l=0,1,\cdots,K-1$, Algorithm \ref{alg:multiple_CP} for $S_2$ also yields \eqref{eq:multiple_CP_consistency} as $n\rightarrow\infty,\alpha\rightarrow0,\sqrt{n}\alpha\rightarrow\infty$.

\end{theorem}

\section{Power and Localization Comparison}\label{sec:simulation}

\begin{table}\centering
\begin{subtable}{\textwidth}\centering
\label{table:power_Gaussian_mean}
\begin{center}
\begin{small}
\scalebox{.7}{
\begin{sc}
\begin{tabular}{lcccccccr} 
\toprule
dim     &   $\mu_1$  & $\sigma_1$ & $Z_w$           & $M$            & $S_1$          & $S_2$           & $S_3$  & $D$           \\
\midrule
1  & 0.8 & - & 0.41 (15.12) & 0.40 (17.68)  & \textbf{0.85 (6.01)} & 0.05 (29.78) & 0.23 (18.77) & 0.27 (16.52)  \\ 
10  & 0.3  & -& 0.41 (19.98) & 0.35 (21.77) & \textbf{0.66 (8.43)} & 0.03 (28.09) & 0.04 (26.73) & 0.06 (31.46)  \\ 
50  & 0.2 &-& 0.33 (14.13) & 0.27 (18.07) & \textbf{0.90 (4.11)}  & 0.01 (33.00)    & 0.02 (28.64) & 0.05 (39.92)  \\ 
100  & 0.2 &-& 0.57 (11.59) & 0.46 (13.64) & \textbf{0.98 (2.40)}  & 0.05 (34.21) & 0.07 (24.88) & 0.07 (41.45)  \\ 
500  & 0.1 &-& 0.33 (17.24) & 0.26 (18.69) & \textbf{0.75 (7.21)} & 0.02 (38.07) & 0.02 (38.67) & 0.06 (45.29)  \\ 
\midrule
1&-&2 & 0.39 (15.71) & 0.39 (18.78)  & 0.07 (37.36) & {0.49 (13.46)} & {0.44 (13.49)} & \textbf{0.59 (12.09)}  \\ 
10&-&1.2 & 0.13 (19.98) & 0.45 (11.64) & {0.02 (32.53)} & {0.77 (8.27)} & {0.77 (8.28)} & \textbf{0.80 (7.89)}  \\ 
50&-&1.06 & 0.06 (25.73) & 0.33 (13.56) & {0.05 (33.71)} & \textbf{0.54 (13.34)} & {0.52 (14.49)} & {0.44 (17.20)}  \\ 
100&-&1.05 & 0.04 (25.16) & 0.24 (13.94) & {0.03 (32.88)} & \textbf{0.64 (11.20)} & {0.63 (11.25)} & {0.23 (16.95)}  \\ 
500&-&1.03 & 0.12 (23.50) & 0.70 (8.89) & {0.09 (30.67)} & \textbf{0.82 (5.79)} & {0.81 (6.28)} & {0.24 (30.64)}  \\ 
\midrule
1&2&2 & 0.53 (10.81) & 0.53 (12.82)  & 0.39 (16.51) & {0.47 (9.30)} & {0.52 (9.18)} & \textbf{0.65 (7.65)}  \\ 
10&0.6&1.2 & 0.45 (10.24) & 0.55 (7.98) & {0.44 (9.05)} & {0.81 (6.53)} & {0.83 (5.77)} & \textbf{0.90 (5.75)}  \\ 
50&0.4&1.06 & 0.41 (11.61) & 0.43 (12.32) & \textbf{0.80 (8.09)} & {0.49 (9.62)} & {0.51 (7.93)} & {0.36 (18.88)}  \\
100&0.4&1.05 & 0.00 (42.15) & 0.00 (42.15) & \textbf{1.00 (2.65)} & 0.62 (9.03) & {0.71 (6.54)} & {0.25 (19.48)}  \\ 
500&0.2&1.03 & 0.47 (9.06) & 0.60 (9.96) & \textbf{0.74 (4.93)} & {0.74 (5.08)} & {0.74 (4.96)} & {0.19 (30.47)}  \\ 
1 (Poisson) & $\times2$  & $\times\sqrt{2}$ &1.00 (11.41)    & 1.00 (11.36)    & \textbf{0.99 (3.69)} & 0.02 (20.16) & 0.60 (6.64)   & 0.69 (5.98)   \\
\bottomrule
\end{tabular}
\end{sc}
}
\end{small}
\end{center}  
\caption{\small Euclidean, high-dimensional data.}
\end{subtable}

\begin{subtable}{\textwidth}\centering
\label{table:power_network}
\begin{center}
\begin{small}
\scalebox{.63}{
\begin{sc}
\begin{tabular}{lcccccccr} 
\toprule
$F_1$   & $Z_w$           & $M$            & $S_1$          & $S_2$           & $S_3$           & $D$            \\
\midrule
$p_1=0.3$ & 0.62 (12.93) & 0.49 (12.19)  & \textbf{0.98 (4.18)} & {0.30 (17.45)} & {0.33 (15.61)} & {0.29 (18.47)}  \\ 
$p_1=0.4$ & 1.00 (2.62) & 0.97 (2.66)  & \textbf{1.00 (0.92)} & {0.62 (8.03)} & {0.80 (4.72)} & {0.63 (11.55)}  \\ 
$p_1=0.5$ & 1.00 (1.01) & 1.00 (1.01)  & \textbf{1.00 (0.31)} & {0.67 (9.43)} & {0.96 (2.89)} & {0.85 (5.44)}  \\ 
\bottomrule
\end{tabular}
\end{sc}
}
\end{small}
\end{center}
\caption{\small Network data. }
\end{subtable}

\begin{subtable}{\textwidth}\centering
\label{table:power_network}
\begin{center}
\begin{small}
\scalebox{.63}{
\begin{sc}
\begin{tabular}{lcccccccr} 
\toprule
$F_1$   & $Z_w$           & $M$            & $S_1$          & $S_2$           & $S_3$           & $D$            \\
\midrule
$\mu=0.03$ & 0.08 (30.02) & 0.09 (28.22)  & \textbf{0.09 (25.81)} & {0.05 (41.94)} & {0.05 (41.60)} & {0.08 (48.06)}  \\ 
$\mu=0.05$ & 0.12 (24.19) & 0.11 (25.51)  & \textbf{0.44 (14.97)} & {0.04 (46.70)} & {0.03 (47.04)} & {0.05 (43.30)}  \\ 
$\mu=0.08$ & 0.68 (9.02) & 0.56 (11.05)  & \textbf{1.00 (1.46)} & {0.02 (49.76)} & {0.03 (49.08)} & {0.03 (40.24)}  \\ 
$\mu=0.10$ & 0.99 (4.00) & 0.95 (4.50)  & \textbf{1.00 (0.46)} & {0.02 (44.32)} & {0.02 (44.32)} & {0.07 (45.68)}  \\ 
\bottomrule
\end{tabular}
\end{sc}
}
\end{small}
\end{center}
\caption{\small Function data. }
\end{subtable}
\caption{\small Comparison of power and localization accuracy for each method in AMOC setting. We report the proportion of experiments with p-value smaller than 0.05, out of 100 replications. Number in brackets are the localization error $|\hat \tau-\tau^*|$ averaged over 100 replications.}
\label{table:AMOC}
\end{table}

We investigate the power and localization accuracy for the proposed statistics in simulated datasets. We compare the proposed statistics ($S_1,S_2,S_3$, all with higher order corrections introduced in Section \ref{sec:asymptotic_null}) against $Z_w,\,M$ and that of \citet{dubey2019fr}, which we refer to as $D$. Notice that $S_3$ is the same as $D$ (see Definition \ref{def:stat_S3}), except that we are using higher order corrections introduced in Section \ref{sec:asymptotic_null}. 
For graph-based methods, as suggested by \citet{chen2017new}, we use 5-MST to construct the graph. We investigate high-dimensional data (normal distributed unless specifically noted), Erdos-Renyi random graph, and functional data. Denote $\bm 1_d\in\mathbb{R}^d$ as the vector of all 1's, $I_d\in\mathbb{R}^{d\times d}$ as the identity matrix.

\subsection{AMOC Setting}
We set $n=100$, $\tau^*=33$ and report: (1) the empirical probability of correctly detecting the change point and (2) the $L_1$ error in locating the change point, both averaged over 100 simulations. For fairness, we use $1\,000$ permutations to compute the p-values for all methods.   Results are summarized in Table \ref{table:AMOC}.

For Euclidean data, we set $d(y,y')=\|y-y'\|^2$ where $\|\cdot\|$ is Euclidean distance. We compare (multivariate) Gaussian and Poisson distribution. 
When only mean changes, we set $F_0=N(0\times \bm 1_d, I_d)$, $F_1=\mu_1\times \bm 1_d+F_0$.
We observe that $S_1$ consistently outperforms other methods by a large margin, both in terms of power and localization accuracy. $S_3$ and $D$, although combining information from $S_1$, have inferior performance because they put larger weights (of higher order) on $S_2$. 
When only scale changes, we set $F_0=N(0\times \bm 1_d, I_d)$, $F_1=\sigma_1F_0$. We see that in low to moderate dimensions, $S_2,S_3$ and $D$ have the best performance; in high dimensions, higher order corrections become important and thus, $S_2$ and $S_3$ become superior than $D$.
When both mean and scale change, we set $F_0=N(0\times \bm 1_d, I_d)$, $F_1=\mu_1\times \bm 1_d+\sigma_1F_0$ on normal data and $F_0=\text{Poisson}(2)$, $F_1=\text{Poisson}(4)$ on Poisson data.  We see that in low dimensions, $S_1,S_2,D$ are performing well, and $D$ is slightly better than $S_1,S_2$. However when dimensionality grows high, $D$ becomes inferior and depending on the actual magnitude of changes, it is often one of $S_1,S_2,S_3$ that performs the best.

For network data, we use Erdos-Renyi random graph with 10 nodes. Before change point, an edge is formed independently between two nodes with probability $p_0=0.1$. After change point, a community emerges among the first 3 nodes, the probability of forming an edge within which becomes $p_1$. The probability of forming an edge among other pairs remains $0.1$. We use $d(y,y')=\|y-y'\|_F^2$ where $\|\cdot\|_F$ is the Frobenius norm and $y$ is the adjacency matrix where an edge is represented by 1 and otherwise 0. As suggested by Table \ref{table:AMOC}, $S_1$ outperforms all other methods.

For functional data, suppose each $y_i$ is a noisy observation of a discretized function at $1000$ equally spaced grids. Set $y_i(x)=\sin(x)+0.5N(0,1),x\in[0,2\pi]$ for $i=1,2,\cdots,\tau^*$ and $y_i(x)=\sin(x+\mu)+0.5N(0,1),x\in[0,2\pi]$ for $i=\tau^*+1,\cdots,n$. We use distance $d(y,y')=\int_{0}^{2\pi}\left|y(x)-y'(x)\right|^2dx$. Table \ref{table:AMOC} reveals that $S_1$ consistently has the best performance.

\begin{table}\centering
\begin{subtable}{.49\textwidth}\centering
\begin{center}
\begin{small}
\scalebox{.7}{
\begin{sc}
\begin{tabular}{lcccccccr} 
\toprule
dim& $\mu_1$ & $\mu_2$  & $Z_w$           & $M$            & $S_1$          & $S_2$           & $S_3$           & $D$             \\ 
\midrule
1   & 2  &1   & 0.87 & 0.72  & $\bm{0.92}$ & 0.55 & 0.78 & 0.78  \\ 
10  &0.5 &0.2 & 0.67 & 0.57  & $\bm{0.73}$ & 0.34 & 0.34 & 0.37  \\ 
50  &0.5 &0.2 & 0.67 & 0.57  & $\bm{0.73}$ & 0.34 & 0.54 & 0.37  \\ 
100 &0.3 &0.1 & 0.61 & 0.56  & $\bm{0.99}$  & 0.34 & 0.45 & 0.38  \\ 
500 &0.2 &0.1 & 0.57 & 0.57  & $\bm{0.88}$ & 0.34 & 0.34 & 0.34 \\ 
\bottomrule
\end{tabular}
\end{sc}
}
\end{small}
\end{center}    
\caption{\small High-dimensional data, mean change.}
\label{table:multiple_power_Gaussian_mean}
\end{subtable}
\begin{subtable}{.49\textwidth}\centering
\begin{center}
\begin{small}
\scalebox{.7}{
\begin{sc}
\begin{tabular}{lcccccccr} 
\toprule
dim&$\sigma_1$  &$\sigma_2$ & $Z_w$           & $M$            & $S_1$          & $S_2$           & $S_3$           & $D$          \\
\midrule
1 &2&$\sqrt{2}$ & 0.81& 0.71& 0.41& $\bm{0.92}$ & $\bm{0.92}$ & $\bm{0.92}$  \\ 
10&$1.2$ &1.2 & 0.48 & 0.63 & 0.41& $\bm{0.95}$ & $\bm{0.95}$ & $\bm{0.97}$  \\ 
50& $1.06$ & 1.06 & 0.34& 0.63& 0.34& $\bm{0.83}$ & 0.78& 0.76 \\ 
100 & $1.05$ & 1.05 & 0.49& 0.55& 0.34 & $\bm{0.92}$ & $\bm{0.92}$ & ${0.71}$  \\ 
500 & $1.03$ & 1.03 & 0.38 & 0.67& 0.34 & $\bm{0.89 }$ & $\bm{0.89}$ & ${0.69}$ \\
\bottomrule
\end{tabular}
\end{sc}
}
\end{small}
\end{center}  
\caption{\small High-dimensional data, scale change.}
\label{table:multiple_power_Gaussian_scale}
\end{subtable}

\begin{subtable}{\textwidth}\centering
\begin{center}
\begin{small}
\scalebox{.85}{
\begin{sc}
\begin{tabular}{lccccccccccr} 
\toprule
dim& $\mu_1$&$\sigma_1$&$\mu_2$ & $\sigma_2$   & $Z_w$           & $M$            & $S_1$          & $S_2$           & $S_3$           & $D$          \\
\midrule
1 &2&2&1&$\sqrt{2}$& 0.82& 0.72 & 0.67& 0.47& 0.47 & $\bm{0.72}$  \\ 
10 & 0.6 & $1.2$ & 0.3 & 1.2 &0.56& 0.61 & 0.82& $\bm{0.95}$ & $\bm{0.95}$ & $\bm{0.95}$  \\ 
50 & 0.4 & $1.06$ & 0.2 & 1.06 & 0.93 &0.71 & $\bm{0.996}$ & ${0.78}$ & $0.86$ & ${0.81}$  \\ 
100 & 0.4 & $1.05$ & 0.2 &1.05 & 0.55& 0.50& $0.92$ & $\bm{0.95}$ & $\bm{0.95}$ & ${0.69}$  \\ 
500 & 0.2 & $1.03$ & 0.1 & 1.03 & 0.72 & 0.72 & $0.92$ & $\bm{0.95}$ & $\bm{0.95}$ & ${0.73}$  \\ 
1 (Poisson) & $\times 1.5$& $\times \sqrt{1.5}$  & $\times 1$& $\times 1$  &0.51 &0.51& $\bm{0.91}$ & 0.41& 0.73 & 0.76\\
\bottomrule
\end{tabular}
\end{sc}
}
\end{small}
\end{center}  
\caption{\small High-dimensional data, both mean and scale change.}
\label{table:multiple_power_Gaussian_both}
\end{subtable}

\begin{subtable}{.49\textwidth}\centering
\begin{center}
\begin{small}
\scalebox{.85}{
\begin{sc}
\begin{tabular}{lcccccccr} 
\toprule
$F_1$   & $Z_w$           & $M$            & $S_1$          & $S_2$           & $S_3$           & $D$            \\
\midrule
$p_1=0.3$ &0.40&0.40 & $\bm{0.70}$ & 0.52& 0.52 & 0.52  \\ 
$p_1=0.4$ & 0.59 & 0.48 & $\bm{0.91}$ & $0.53$ &$0.53$& 0.49 \\ 
$p_1=0.5$ & 0.83 & 0.66& $\bm{0.93}$ & 0.74 & 0.87& 0.75 \\
\bottomrule
\end{tabular}
\end{sc}
}
\end{small}
\end{center}  
\caption{\small Network data.}
\label{table:multiple_power_network}
\end{subtable}
\begin{subtable}{.49\textwidth}\centering
\begin{center}
\begin{small}
\scalebox{.85}{
\begin{sc}
\begin{tabular}{lcccccccr} 
\toprule
$F_1$   & $Z_w$           & $M$            & $S_1$          & $S_2$           & $S_3$           & $D$            \\
\midrule
$\mu=0.03$ &0.34 &0.34& $\bm{0.41}$ & ${0.34}$ & 0.34&0.34 \\ 
$\mu=0.05$ & 0.34& 0.34 & $\bm{0.47}$ & 0.34&0.34& 0.34\\ 
$\mu=0.08$ & 0.34 & 0.34& $\bm{0.81}$ & 0.34& 0.34&0.34 \\ 
$\mu=0.1$ &0.55& 0.55 & $\bm{0.83}$ &0.34 & 0.34 & 0.34 \\ 
\bottomrule
\end{tabular}
\end{sc}
}
\end{small}
\end{center} 
\caption{\small Functional data }
\label{table:multiple_power_function}
\end{subtable}
\caption{Comparison of Rand Index \citep{rand1971objective} (averaged over 100 replications) computed by R package ``fossil" \citep{vavrek2015fossil} for different methods in multiple change points setting. We set $n=150,\,\tau_1^*=40,\tau_2^*=100$. }
\label{table:multiple}
\end{table}

\subsection{Multiple Change Points Setting}
We set $n=150$, $\tau_1^*=40,\tau_2^*=100$. Following \citet{matteson2014nonparametric}, we use Rand Index defined below to measure the performance of each method. 
\begin{definition}[Rand Index]\label{dfn:rand}
For any two clusterings $U,V$ of $n$ observations, the Rand Index is defined as
\begin{equation*}
    \text{Rand}=\frac{\# I_1+\#I_2}{
\begin{pmatrix}
2\\
n
\end{pmatrix}}
\end{equation*}
where $\# I_1$ is the number of pairs in the same cluster under $U$ and $V$, and $\# I_2$ is the number of pairs in different clusters under $U$ and $V$. 
\end{definition}
Rand Index incorporates information from both power and localization accuracy, and higher value means better performance.  In practice we compute Rand Index by R package ``fossil" \citep{vavrek2015fossil}. Originally $Z_w,M$ and $D$ are not designed for multiple change points setting, but we generalize them using a similar binary segmentation procedure as in Algorithm \ref{alg:multiple_CP} (we use $n_{\min}=20$, $\alpha=0.05$).  Results are shown in Table \ref{table:multiple}.

For Euclidean data, when there is only mean change, we set $F_0=N(0\times \bm 1_d, I_d)$, $F_1=\mu_1\times \bm 1_d+F_0$, $F_2=\mu_2\times \bm 1_d+F_0$.  Table \ref{table:multiple_power_Gaussian_mean} shows that $S_1$  has the best performance. When there is only scale change, we set $F_0=N(0\times \bm 1_d, I_d)$, $F_1=\sigma_1F_0$, $F_2=\sigma_2F_0$.  Observe that $S_2$ has the best performance. When there are changes in both mean and scale, we set $F_0=N(0\times \bm 1_d, I_d)$, $F_1=\mu_1\times \bm 1_d+\sigma_1F_0$, $F_2=\mu_2\times \bm 1+\sigma_2F_0$ on normal data and $F_0=F_2=\text{Poisson}(4)$, $F_1=\text{Poisson}(6)$ on Poisson data.  Depending on the actual magnitude of changes, one of $S_1, S_2,S_3$ has the best performance. And the gain of using $S_3$ instead of $D$ becomes larger as dimensionality becomes larger.

For network data, we use Erdos-Renyi random graph with 10 nodes. For $F_0$ and $F_2$, an edge is formed independently between two nodes with probability $p_0=0.1$. For $F_1$, a community emerges among the first 3 nodes, the probability of forming an edge within which becomes $p_1$. The probability of forming an edges among other pairs remains $p_0=0.1$. We observe in Table \ref{table:multiple_power_network} that $S_1$ has the best performance.

For functional data, for $F_0$, we set $y_i(x)=\sin(x)+0.5N(0,1),x\in[0,2\pi]$; for $F_1$, $y_i(x)=\sin(x+2\mu)+0.5N(0,1),x\in[0,2\pi]$; for $F_2$, $y_i(x)=\sin(x+\mu)+0.5N(0,1),x\in[0,2\pi]$. The other settings are identical to the AMOC setting. We observe in Table \ref{table:multiple_power_function} that $S_1$ has the best performance.

\paragraph{Conclusion}
In both AMOC and multiple change points setting, the proposed statistics outperform baselines. If we know the type of change, using the corresponding $S_1$ (or $S_2$) is highly advantageous. Higher order corrections are useful, especially under high dimensions where it greatly improves both power and localization accuracy. For unknown type of change, which one of $S_1,S_2,S_3$ performs best is dependent on the actual distribution. We can either apply $S_3$, considering its relative robust performance across different types of changes; or we can apply $S_1$ and $S_2$ separately.

\section{Real Data Analysis}\label{sec:real_data}
\begin{figure}
    \centering
    \includegraphics[width=.22\linewidth]{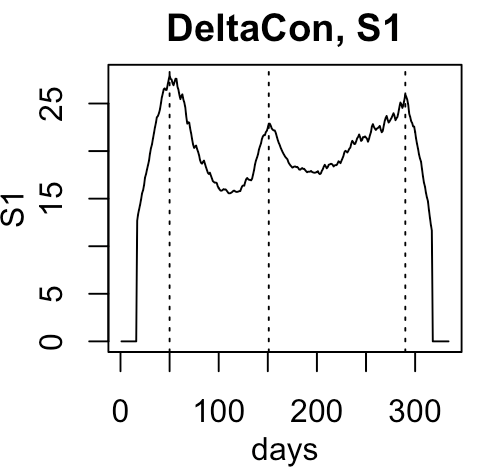}
    \includegraphics[width=.22\linewidth]{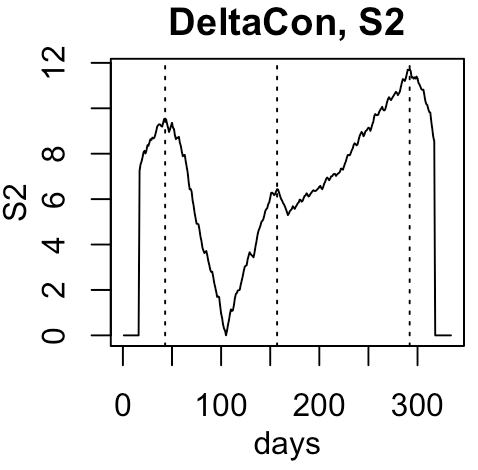}
    \caption{\small Plot of $S_1,S_2$ on MIT proximity network data. The peaks correspond to the change points detected by Algorithm \ref{alg:multiple_CP}.}
    \label{fig:my_label}
\end{figure}
The MIT proximity network is extracted from the MIT Reality Mining dataset \citep{pentland2009inferring}, which consists of the proximity network for $m=93$ faculty and graduate students recorded via cell phone Bluetooth scan every five minutes. From the raw data, we extracted a sequence of daily binary networks $\{y_i\}_{i=1}^n\in\mathbb{R}^{m\times m}$ from July 2004 to June 2005, where a link between two subjects means that they are scanned together at least once during that day. We used DELTACON \citep{koutra2013deltacon} to measure the distance $d(y_i,y_j)$, which is defined as
$$
d\left(y_i,y_j\right)=\left[\sum_{k=1}^m\sum_{l=1}^m\left(\sqrt{q_{i,kl}}-\sqrt{q_{j,kl}}\right)^2\right]^{1/2}, 
$$
where $Q_i:=[q_{i,kl}]_{k,l=1}^m=\left(I_m+\epsilon_i^2U_i-\epsilon_i y_i\right)^{-1}\in\mathbb{R}^{m\times m}$,  $\epsilon_i=\frac{1}{1+\max_k(c^i_{k,k})}$ with $c^i_{k,k}$ the degree of the $k$-th subject in the $i$-th network, and $U_i=\text{diag}(c^i_{1,1},\cdots,c^i_{m,m})$. 

The scan statistic on the original sequence is shown in Figure \ref{fig:my_label}.
Using Algorithm \ref{alg:multiple_CP}, we identify the $50^{\text{th}}$ (2004/9/6), $151^{\text{st}}$ (2004/12/17), and $290^{\text{th}}$ (2005/5/4) day as change points. They correspond to the first day of class (2004/9/8), end of exam week (2004/12/17), and the last day of classes (2005/5/12), all with p-value approximately equal to 0. Using $S_2$ identifies very similar change points. Different distance measures (Frobenius, NetSimile \citep{berlingerio2012netsimile}) led to similar results.

\section{Discussion and Conclusion}
\label{sec:discuss}
 
We propose nonparametric scan statistics for the detection and localization of change points based on the graph-based CPD framework. The proposed statistics are applicable to both AMOC and multiple change points setting. We provide analytic forms to control type I error of the proposed statistics, as well as prove their power consistency and minimax localization rate.  This work also establishes connections among various CPD methods. In particular, we found that the graph-based statistics $Z_w, Z_{\diff}$ \citep{chu2019asymptotic} exhibit similar forms as the familiar CUSUM statistic, which justifies the empirical observations on their performance. 


 The performance of the statistics is determined by both the magnitude of change and the distance measure. Ideally the distance $d$ should be able to capture all possible changes in the distribution. In the extreme case where the change in distribution is not reflected by $d$ (more precisely, the feature map $\phi$ associated with $d$), the proposed statistics will lack power.  Thus, distance selection or distance learning from data is an important topic which needs further investigation.


\bibliography{main}

\begin{thebibliography}{8}
\providecommand{\natexlab}[1]{#1}
\providecommand{\url}[1]{\texttt{#1}}
\expandafter\ifx\csname urlstyle\endcsname\relax
  \providecommand{\doi}[1]{doi: #1}\else
  \providecommand{\doi}{doi: \begingroup \urlstyle{rm}\Url}\fi

\bibitem[Author(2021)]{anonymous}
Author, N.~N.
\newblock Suppressed for anonymity, 2021.

\bibitem[Duda et~al.(2000)Duda, Hart, and Stork]{DudaHart2nd}
Duda, R.~O., Hart, P.~E., and Stork, D.~G.
\newblock \emph{Pattern Classification}.
\newblock John Wiley and Sons, 2nd edition, 2000.

\bibitem[Kearns(1989)]{kearns89}
Kearns, M.~J.
\newblock \emph{Computational Complexity of Machine Learning}.
\newblock PhD thesis, Department of Computer Science, Harvard University, 1989.

\bibitem[Langley(2000)]{langley00}
Langley, P.
\newblock Crafting papers on machine learning.
\newblock In Langley, P. (ed.), \emph{Proceedings of the 17th International
  Conference on Machine Learning (ICML 2000)}, pp.\  1207--1216, Stanford, CA,
  2000. Morgan Kaufmann.

\bibitem[Michalski et~al.(1983)Michalski, Carbonell, and
  Mitchell]{MachineLearningI}
Michalski, R.~S., Carbonell, J.~G., and Mitchell, T.~M. (eds.).
\newblock \emph{Machine Learning: An Artificial Intelligence Approach, Vol. I}.
\newblock Tioga, Palo Alto, CA, 1983.

\bibitem[Mitchell(1980)]{mitchell80}
Mitchell, T.~M.
\newblock The need for biases in learning generalizations.
\newblock Technical report, Computer Science Department, Rutgers University,
  New Brunswick, MA, 1980.

\bibitem[Newell \& Rosenbloom(1981)Newell and Rosenbloom]{Newell81}
Newell, A. and Rosenbloom, P.~S.
\newblock Mechanisms of skill acquisition and the law of practice.
\newblock In Anderson, J.~R. (ed.), \emph{Cognitive Skills and Their
  Acquisition}, chapter~1, pp.\  1--51. Lawrence Erlbaum Associates, Inc.,
  Hillsdale, NJ, 1981.

\bibitem[Samuel(1959)]{Samuel59}
Samuel, A.~L.
\newblock Some studies in machine learning using the game of checkers.
\newblock \emph{IBM Journal of Research and Development}, 3\penalty0
  (3):\penalty0 211--229, 1959.

\end{thebibliography}


\begin{thebibliography}{}

\bibitem[\protect\citeauthoryear{Arlot, Celisse, and Harchaoui}{Arlot
  et~al.}{2012}]{arlot2012kernel}
Arlot, S., A.~Celisse, and Z.~Harchaoui (2012).
\newblock A kernel multiple change-point algorithm via model selection.
\newblock {\em arXiv:1202.3878\/}.

\bibitem[\protect\citeauthoryear{Arlot, Celisse, and Harchaoui}{Arlot
  et~al.}{2019}]{arlot2019kernel}
Arlot, S., A.~Celisse, and Z.~Harchaoui (2019).
\newblock A kernel multiple change-point algorithm via model selection.
\newblock {\em Journal of machine learning research\/}~{\em 20\/}(162).

\bibitem[\protect\citeauthoryear{Berlinet and Thomas-Agnan}{Berlinet and
  Thomas-Agnan}{2011}]{berlinet2011reproducing}
Berlinet, A. and C.~Thomas-Agnan (2011).
\newblock {\em Reproducing kernel Hilbert spaces in probability and
  statistics}.
\newblock Springer Science \& Business Media.

\bibitem[\protect\citeauthoryear{Berlingerio, Koutra, Eliassi-Rad, and
  Faloutsos}{Berlingerio et~al.}{2012}]{berlingerio2012netsimile}
Berlingerio, M., D.~Koutra, T.~Eliassi-Rad, and C.~Faloutsos (2012).
\newblock Netsimile: A scalable approach to size-independent network
  similarity.
\newblock {\em arXiv:1209.2684\/}.

\bibitem[\protect\citeauthoryear{Brunel}{Brunel}{2014}]{brunel2014convex}
Brunel, V.-E. (2014).
\newblock Convex set detection.
\newblock {\em arXiv:1404.6224\/}.

\bibitem[\protect\citeauthoryear{Celisse, Marot, Pierre-Jean, and
  Rigaill}{Celisse et~al.}{2018}]{celisse2018new}
Celisse, A., G.~Marot, M.~Pierre-Jean, and G.~Rigaill (2018).
\newblock New efficient algorithms for multiple change-point detection with
  reproducing kernels.
\newblock {\em Computational Statistics \& Data Analysis\/}~{\em 128},
  200--220.

\bibitem[\protect\citeauthoryear{Chang, Li, Yang, and P{\'o}czos}{Chang
  et~al.}{2019}]{chang2019kernel}
Chang, W.-C., C.-L. Li, Y.~Yang, and B.~P{\'o}czos (2019).
\newblock Kernel change-point detection with auxiliary deep generative models.
\newblock {\em arXiv:1901.06077\/}.

\bibitem[\protect\citeauthoryear{Chen}{Chen}{2019}]{chen2019change}
Chen, H. (2019).
\newblock Change-point detection for multivariate and non-euclidean data with
  local dependency.
\newblock {\em arXiv:1903.01598\/}.

\bibitem[\protect\citeauthoryear{Chen et~al.}{Chen
  et~al.}{2019}]{chen2019sequential}
Chen, H. et~al. (2019).
\newblock Sequential change-point detection based on nearest neighbors.
\newblock {\em The Annals of Statistics\/}~{\em 47\/}(3), 1381--1407.

\bibitem[\protect\citeauthoryear{Chen, Chen, and Su}{Chen
  et~al.}{2018}]{chen2018weighted}
Chen, H., X.~Chen, and Y.~Su (2018).
\newblock A weighted edge-count two-sample test for multivariate and object
  data.
\newblock {\em Journal of the American Statistical Association\/}~{\em
  113\/}(523), 1146--1155.

\bibitem[\protect\citeauthoryear{Chen and Friedman}{Chen and
  Friedman}{2017}]{chen2017new}
Chen, H. and J.~H. Friedman (2017).
\newblock A new graph-based two-sample test for multivariate and object data.
\newblock {\em Journal of the American statistical association\/}~{\em
  112\/}(517), 397--409.

\bibitem[\protect\citeauthoryear{Chen, Zhang, et~al.}{Chen
  et~al.}{2015}]{chen2015graph}
Chen, H., N.~Zhang, et~al. (2015).
\newblock Graph-based change-point detection.
\newblock {\em The Annals of Statistics\/}~{\em 43\/}(1), 139--176.

\bibitem[\protect\citeauthoryear{Chen, Zhang, and Chu}{Chen
  et~al.}{2015}]{chen2015gseg}
Chen, H., N.~Zhang, and L.~Chu (2015).
\newblock gseg: Graph-based change-point detection (g-segmentation). r package,
  version 0.1.

\bibitem[\protect\citeauthoryear{Chu and Chen}{Chu and
  Chen}{2018}]{chu2018sequential}
Chu, L. and H.~Chen (2018).
\newblock Sequential change-point detection for high-dimensional and
  non-euclidean data.
\newblock {\em arXiv:1810.05973\/}.

\bibitem[\protect\citeauthoryear{Chu, Chen, et~al.}{Chu
  et~al.}{2019}]{chu2019asymptotic}
Chu, L., H.~Chen, et~al. (2019).
\newblock Asymptotic distribution-free change-point detection for multivariate
  and non-euclidean data.
\newblock {\em The Annals of Statistics\/}~{\em 47\/}(1), 382--414.

\bibitem[\protect\citeauthoryear{Desobry, Davy, and Doncarli}{Desobry
  et~al.}{2005}]{desobry2005online}
Desobry, F., M.~Davy, and C.~Doncarli (2005).
\newblock An online kernel change detection algorithm.
\newblock {\em IEEE Transactions on Signal Processing\/}~{\em 53\/}(8),
  2961--2974.

\bibitem[\protect\citeauthoryear{Doukhan}{Doukhan}{2012}]{doukhan2012mixing}
Doukhan, P. (2012).
\newblock {\em Mixing: properties and examples}, Volume~85.
\newblock Springer Science \& Business Media.

\bibitem[\protect\citeauthoryear{Dubey and M{\"u}ller}{Dubey and
  M{\"u}ller}{2019}]{dubey2019fr}
Dubey, P. and H.-G. M{\"u}ller (2019).
\newblock Frechet change point detection.
\newblock {\em arXiv:1911.11864\/}.

\bibitem[\protect\citeauthoryear{Friedman and Rafsky}{Friedman and
  Rafsky}{1979}]{friedman1979multivariate}
Friedman, J.~H. and L.~C. Rafsky (1979).
\newblock Multivariate generalizations of the wald-wolfowitz and smirnov
  two-sample tests.
\newblock {\em The Annals of Statistics\/}, 697--717.

\bibitem[\protect\citeauthoryear{Garreau, Arlot, et~al.}{Garreau
  et~al.}{2018}]{garreau2018consistent}
Garreau, D., S.~Arlot, et~al. (2018).
\newblock Consistent change-point detection with kernels.
\newblock {\em Electronic Journal of Statistics\/}~{\em 12\/}(2), 4440--4486.

\bibitem[\protect\citeauthoryear{Gretton, Borgwardt, Rasch, Sch{\"o}lkopf, and
  Smola}{Gretton et~al.}{2012}]{gretton2012kernel}
Gretton, A., K.~M. Borgwardt, M.~J. Rasch, B.~Sch{\"o}lkopf, and A.~Smola
  (2012).
\newblock A kernel two-sample test.
\newblock {\em Journal of Machine Learning Research\/}~{\em 13\/}(Mar),
  723--773.

\bibitem[\protect\citeauthoryear{Gretton, Fukumizu, Harchaoui, and
  Sriperumbudur}{Gretton et~al.}{2009}]{gretton2009fast}
Gretton, A., K.~Fukumizu, Z.~Harchaoui, and B.~K. Sriperumbudur (2009).
\newblock A fast, consistent kernel two-sample test.
\newblock In {\em Advances in neural information processing systems}, pp.\
  673--681.

\bibitem[\protect\citeauthoryear{Harchaoui, Moulines, and Bach}{Harchaoui
  et~al.}{2009}]{harchaoui2009kernel}
Harchaoui, Z., E.~Moulines, and F.~R. Bach (2009).
\newblock Kernel change-point analysis.
\newblock In {\em Advances in neural information processing systems}, pp.\
  609--616.

\bibitem[\protect\citeauthoryear{Harchaoui, Vallet, Lung-Yut-Fong, and
  Capp{\'e}}{Harchaoui et~al.}{2009}]{harchaoui2009regularized}
Harchaoui, Z., F.~Vallet, A.~Lung-Yut-Fong, and O.~Capp{\'e} (2009).
\newblock A regularized kernel-based approach to unsupervised audio
  segmentation.
\newblock In {\em 2009 IEEE International Conference on Acoustics, Speech and
  Signal Processing}, pp.\  1665--1668. IEEE.

\bibitem[\protect\citeauthoryear{Henze}{Henze}{1988}]{henze1988multivariate}
Henze, N. (1988).
\newblock A multivariate two-sample test based on the number of nearest
  neighbor type coincidences.
\newblock {\em The Annals of Statistics\/}, 772--783.

\bibitem[\protect\citeauthoryear{Huang, Kong, and Huang}{Huang
  et~al.}{2014}]{huang2014high}
Huang, S., Z.~Kong, and W.~Huang (2014).
\newblock High-dimensional process monitoring and change point detection using
  embedding distributions in reproducing kernel hilbert space.
\newblock {\em IIE Transactions\/}~{\em 46\/}(10), 999--1016.

\bibitem[\protect\citeauthoryear{Koutra, Vogelstein, and Faloutsos}{Koutra
  et~al.}{2013}]{koutra2013deltacon}
Koutra, D., J.~T. Vogelstein, and C.~Faloutsos (2013).
\newblock Deltacon: A principled massive-graph similarity function.
\newblock In {\em Proceedings of the 2013 SIAM International Conference on Data
  Mining}, pp.\  162--170. SIAM.

\bibitem[\protect\citeauthoryear{Lee, Na, and Na}{Lee
  et~al.}{2003}]{lee2003cusum}
Lee, S., O.~Na, and S.~Na (2003).
\newblock On the cusum of squares test for variance change in nonstationary and
  nonparametric time series models.
\newblock {\em Annals of the Institute of Statistical Mathematics\/}~{\em
  55\/}(3), 467--485.

\bibitem[\protect\citeauthoryear{Li, Xie, Dai, and Song}{Li
  et~al.}{2015}]{li2015m}
Li, S., Y.~Xie, H.~Dai, and L.~Song (2015).
\newblock M-statistic for kernel change-point detection.
\newblock In {\em Advances in Neural Information Processing Systems}, pp.\
  3366--3374.

\bibitem[\protect\citeauthoryear{Li, Li, Krishnan, and Liu}{Li
  et~al.}{2011}]{li2011large}
Li, Z., P.~Li, A.~Krishnan, and J.~Liu (2011).
\newblock Large-scale dynamic gene regulatory network inference combining
  differential equation models with local dynamic bayesian network analysis.
\newblock {\em Bioinformatics\/}~{\em 27\/}(19), 2686--2691.

\bibitem[\protect\citeauthoryear{Liu and Chen}{Liu and
  Chen}{2020}]{liu2020fast}
Liu, Y.-W. and H.~Chen (2020).
\newblock A fast and efficient change-point detection framework for modern
  data.
\newblock {\em arXiv:2006.13450\/}.

\bibitem[\protect\citeauthoryear{Lu, Liang, Li, and Wu}{Lu
  et~al.}{2011}]{lu2011high}
Lu, T., H.~Liang, H.~Li, and H.~Wu (2011).
\newblock High-dimensional odes coupled with mixed-effects modeling techniques
  for dynamic gene regulatory network identification.
\newblock {\em Journal of the American Statistical Association\/}~{\em
  106\/}(496), 1242--1258.

\bibitem[\protect\citeauthoryear{Matteson and James}{Matteson and
  James}{2014}]{matteson2014nonparametric}
Matteson, D.~S. and N.~A. James (2014).
\newblock A nonparametric approach for multiple change point analysis of
  multivariate data.
\newblock {\em Journal of the American Statistical Association\/}~{\em
  109\/}(505), 334--345.

\bibitem[\protect\citeauthoryear{Meckes}{Meckes}{2013}]{meckes2013positive}
Meckes, M.~W. (2013).
\newblock Positive definite metric spaces.
\newblock {\em Positivity\/}~{\em 17\/}(3), 733--757.

\bibitem[\protect\citeauthoryear{Page}{Page}{1954}]{page1954continuous}
Page, E.~S. (1954).
\newblock Continuous inspection schemes.
\newblock {\em Biometrika\/}~{\em 41\/}(1/2), 100--115.

\bibitem[\protect\citeauthoryear{Pentland, Eagle, and Lazer}{Pentland
  et~al.}{2009}]{pentland2009inferring}
Pentland, A., N.~Eagle, and D.~Lazer (2009).
\newblock Inferring social network structure using mobile phone data.
\newblock {\em Proceedings of the National Academy of Sciences (PNAS)\/}~{\em
  106\/}(36), 15274--15278.

\bibitem[\protect\citeauthoryear{Rand}{Rand}{1971}]{rand1971objective}
Rand, W.~M. (1971).
\newblock Objective criteria for the evaluation of clustering methods.
\newblock {\em Journal of the American Statistical association\/}~{\em
  66\/}(336), 846--850.

\bibitem[\protect\citeauthoryear{Rice and Zhang}{Rice and
  Zhang}{2019}]{rice2019consistency}
Rice, G. and C.~Zhang (2019).
\newblock Consistency of binary segmentation for multiple change-points
  estimation with functional data.
\newblock {\em arXiv:2001.00093\/}.

\bibitem[\protect\citeauthoryear{Rosenbaum}{Rosenbaum}{2005}]{rosenbaum2005exact}
Rosenbaum, P.~R. (2005).
\newblock An exact distribution-free test comparing two multivariate
  distributions based on adjacency.
\newblock {\em Journal of the Royal Statistical Society: Series B (Statistical
  Methodology)\/}~{\em 67\/}(4), 515--530.

\bibitem[\protect\citeauthoryear{Sejdinovic, Sriperumbudur, Gretton, and
  Fukumizu}{Sejdinovic et~al.}{2013}]{sejdinovic2013equivalence}
Sejdinovic, D., B.~Sriperumbudur, A.~Gretton, and K.~Fukumizu (2013).
\newblock Equivalence of distance-based and rkhs-based statistics in hypothesis
  testing.
\newblock {\em The Annals of Statistics\/}, 2263--2291.

\bibitem[\protect\citeauthoryear{Sinn, Ghodsi, and Keller}{Sinn
  et~al.}{2012}]{sinn2012detecting}
Sinn, M., A.~Ghodsi, and K.~Keller (2012).
\newblock Detecting change-points in time series by maximum mean discrepancy of
  ordinal pattern distributions.
\newblock {\em arXiv:1210.4903\/}.

\bibitem[\protect\citeauthoryear{Song and Chen}{Song and
  Chen}{2020}]{song2020asymptotic}
Song, H. and H.~Chen (2020).
\newblock Asymptotic distribution-free change-point detection for data with
  repeated observations.
\newblock {\em arXiv:2006.10305\/}.

\bibitem[\protect\citeauthoryear{Talih and Hengartner}{Talih and
  Hengartner}{2005}]{talih2005structural}
Talih, M. and N.~Hengartner (2005).
\newblock Structural learning with time-varying components: tracking the
  cross-section of financial time series.
\newblock {\em Journal of the Royal Statistical Society: Series B (Statistical
  Methodology)\/}~{\em 67\/}(3), 321--341.

\bibitem[\protect\citeauthoryear{Tewes}{Tewes}{2017}]{tewes2017change}
Tewes, J. (2017).
\newblock {\em Change-point tests and the bootstrap under long-and short-range
  dependence}.
\newblock Ph.\ D. thesis, Ruhr-Universit{\"a}t Bochum.

\bibitem[\protect\citeauthoryear{Vavrek}{Vavrek}{2015}]{vavrek2015fossil}
Vavrek, M. (2015).
\newblock Fossil: palaeoecological and palaeogeographical analysis tools. r
  package version 0.3. 7.

\end{thebibliography}
\bibliographystyle{chicago}

\newpage
\appendix
\onecolumn

\section{Additional Theoretical Results}\label{sec_appendix:def}

\begin{prop}\label{prop_appendix:widetilte}
$\widetilde k$ defined in Equation \eqref{eq:k_tilde} does not depend on centering point $y_0$.
\end{prop}
\begin{proof}
For any kernel $k^{(y_0)}$ induced from $d$, from the Moore-Aronszajn Theorem \citep{berlinet2011reproducing}, 
there exists an RKHS $\mathcal{H}_k$ with reproducing kernel $k^{(y_0)}$. We call $\phi^{(y_0)}: y\mapsto k^{(y_0)}(\cdot,y)$
the feature map of $k^{(y_0)}$. 
Notice $\langle\phi^{(y_0)}(y_i),\phi^{(y_0)}(y_j)\rangle_{\mathcal{H}_{k^{(y_0)}}}=k^{(y_0)}(y_i,y_j)$. 
From Lemma \ref{lemma:distance_induced_kernel}, for different centering points $y_0$ and $y_0'$, if $\phi^{(y_0)}(\cdot)$ is a feature map for $k^{(y_0)}$, then $\phi^{(y_0')}$ defined as $\phi^{(y_0')}(y)=\phi^{(y_0)}(y)-\phi^{(y_0)}(y_0')$ is a feature map for $k^{(y_0')}$. This implies $\phi^{(y_0')}(y)-\E\phi^{(y_0')}(y)=\phi^{(y_0)}(y)-\E\phi^{(y_0)}(y)$ for any $y_0,y_0'$ and thus, by definition, $\widetilde k$ does not depend on $y_0$. 

\end{proof}

\begin{prop}\label{prop:S1_S2_ind}
	Under the null, if  $\forall l\in\mathbb{Z}_+,\,\, \E\,\left[\|\phi(y)-\E\phi(y)\|_{\mathcal{H}}^2\left(\phi_l(y)-\E\phi_l(y)\right)\right]=0$ holds, then $S_1$ and $S_2$ are asymptotically independent.
\end{prop}
\begin{proof}
See Section \ref{sec_appendix:indep_proof}.
\end{proof}
\begin{remark}
	Notice that if $y_i\in\mathbb{R}^p$, each coordinate of $y_i$ is independent and follows a symmetric distribution, and $d$ is defined as the Euclidean distance, Proposition \ref{prop:S1_S2_ind} shows that $S_1$ and $S_2$ are asymptotically independent. 
\end{remark}

\section{Technical Proofs}
This section includes proofs to all theoretical results in the main text. First, let us introduce some additional notations.

\subsection{Additional Notations}
Denote $F^i$ as the distribution that $y_i$ follows, i.e., $F^i=F_0$ for all $i$ under the null and $F^i=F_k$ for $i=\tau_k^*+1,\cdots,\tau^*_{k+1}$ under the alternative. We denote $z_i=\phi(y_i)$, $\mu^i=\E_{F^i}\phi(Y)$, and $\epsilon_i=z_i-\mu^i$. Denote $\bar\epsilon_n=\frac{1}{n}\sum_{i=1}^{n}\epsilon_i$, $\bar\epsilon_{t-}=\frac{1}{t}\sum_{i=1}^{t}\epsilon_i$ and $\bar\epsilon_{t+}=\frac{1}{n-t}\sum_{i=t+1}^{n}\epsilon_i$. 
Define $\left(\bm{\mu}^*\right)_{t-}=(\mu^1,\mu^2,\cdots,\mu^t)^\top\in\mathcal{H}^t$,  $\left(\bm{\mu}^*\right)_{t+}=(\mu^{t+1},\mu^{t+2},\cdots,\mu^n)^\top\in\mathcal{H}^{n-t}$,  $\left(\bm{\mu}\right)_{t-}=(\bar \mu_{t-},\bar \mu_{t-},\cdots,\bar \mu_{t-})^\top\in \mathcal{H}^t$ where $\bar \mu_{t-}=\frac{1}{t}\sum_{i=1}^{t}\mu^i$, and  $\left(\bm{\mu}\right)_{t+}=(\bar \mu_{t+},\bar \mu_{t+},\cdots,\bar \mu_{t+})^\top\in\mathcal{H}^{n-t}$ where $\bar \mu_{t+}=\frac{1}{n-t}\sum_{i=t+1}^{n}\mu^i$. The norm in spaces $\mathcal{H}^t$ and $\mathcal{H}^{n-t}$ are defined in the same way as that in $\mathcal{H}$. 
Let
$s^2=\lim_{n\rightarrow \infty}\frac{1}{n}\sum_{i=1}^n\Var_{F^i}(\|\phi(Y)-\frac{1}{n}\sum_{i=1}^n\mu^i\|^2)$. Notice that  $s^2=\Var_{F_0}(\|\epsilon\|^2)$ if there is no change point.  
Write  $\bm\epsilon_{t-}=(\epsilon_1,\cdots,\epsilon_t)^\top$, and $\bm\epsilon_{t+}=(\epsilon_{t+1},\cdots,\epsilon_n)^\top$. We define the operator $\Pi$ as
$\Pi (\bm\epsilon_{t-})= \arg\min_{f=(f_1,f_2,\cdots,f_t)\in \mathcal{H}^t, f_1=f_2=\cdots=f_t}\{\|f-\left(\bm\epsilon\right)_{t-}\|^2\}.$ From  Appendix A.1 of \citet{arlot2012kernel}, we know that $\Pi (\bm\epsilon_{t-})=\left(\bar\epsilon_{t-}, \bar\epsilon_{t-}, \cdots,\bar\epsilon_{t-}\right)^\top$. 
We use $\xrightarrow{p}$ to denote convergence in probability. 

\subsection{Some Useful Results}
The following are some useful results which will be utilized in later proofs.
\begin{lemma}[Proposition 1 from \cite{arlot2012kernel}]\label{appendix_lemma:bound_noise_kernel_bd}
	If (1) $\exists M>0$ such that $\forall i\in\{1,2,\cdots,n\}$, $\|\widetilde k(y_i,y_i)\|^2\leq M^2,\,\,\text{a.s.}$ (2) $y_i$'s are independent, then, for any $x>0$, 
	$$
	\P\left( \left|\,\,\frac{1}{n}\left\|\sum_{i=1}^{n}\epsilon_i\right\|^2 - \frac{1}{n}\E \left\|\sum_{i=1}^{n}\epsilon_i\right\|^2 \,\,\right|\leq \frac{14M^2}{3}\left(x+2\sqrt{2x}\right)\right)\geq 1-2e^{-x}.	
	$$
\end{lemma}
From Equation (19) in \citet{arlot2012kernel}, $\frac{1}{n}\E \|\sum_{i=1}^{n}\epsilon_i\|^2\leq CM^2$ where $C$ is the number of change points +1. Thus, we have
\begin{equation}\label{eq:deviation_bounded}
\P\left( \frac{1}{\sqrt{n}}\left\|\sum_{i=1}^{n}\epsilon_i\right\|
\leq
\sqrt{CM^2+\frac{14M^2}{3}\left(x+2\sqrt{2x}\right)}
\right)\geq 1-2e^{-x}.	
\end{equation}

\begin{lemma}[Lemma 7.10 from \citet{garreau2018consistent}]\label{lemma:deviation_exp_bounded}
	If (1) there exists a positive constant $V$ s.t. $\max_{1\leq i\leq n}\E\|\epsilon_i\|^2\leq V$,  (2) $y_i$'s are independent, then, for any $x>0$,
	$$
	\P\left(\left\|\sum_{i=1}^{n}\epsilon_i\right\|\leq e^{x/2}\sqrt{nV}\right)\geq 1-e^{-x}.
	$$
\end{lemma}

Now we are ready to present proofs of the results in the main article.
\subsection{Proof of Theorem \ref{thm:S1_null}}
\subsubsection{On $S_1$}
Conclusion (a) is a direct consequence of the following Lemma:
\begin{lemma}\label{lemma:S1_null_unscaled}
	Under the null, if 
	
	(1) $d$ is a semi-metric of negative type, and 
	
	(2) $\E_y\left|\widetilde k(y,y)\right|^{2+\delta}<+\infty$ for some $\delta>0$,
	
	(3) $\E_{y,y'}\left|\widetilde k(y,y')\right|^{2}<\infty$, 
	
	then for any ${0< \rho< 1}$, as $n\rightarrow \infty$, we have
	\begin{equation*}
	n\rho^2(1-\rho)^2\left(\bar d_{A(\lceil n\rho\rceil)}-\frac{1}{2}\bar d_{B_1(\lceil n\rho\rceil )}-\frac{1}{2}\bar d_{B_2(\lceil n\rho\rceil )}\right)\xrightarrow{w} \sum_{l=1}^\infty\lambda_l\left(W_l^0(\rho)^2-\rho(1-\rho)\right),
	\end{equation*}
	where $W_l^0(\cdot)$'s are independent Brownian bridges, $\lambda_l$'s are eigenvalues of $\widetilde k$ defined in Equation \eqref{eq:eigen_decomp}.
\end{lemma}
\begin{proof}
	From assumption (1), we know that Lemma \ref{lemma:distance_induced_kernel} holds. Together with assumption (3), we know that eigen-decomposition \eqref{eq:eigen_decomp} holds. 
	
	Notice that assumption (2) is equivalent to $\E\|\phi(y)\|^{4+2\delta}<\infty$. Since our data are i.i.d under the null,
	from Theorem 16 in \citet{tewes2017change}, by directly treating $\{\phi(y_i)\}$ as the observations in Hilbert space $\mathcal{H}$, we have
	\begin{equation*}
	\frac{1}{\sqrt{n}}\sum_{i=1}^{\lceil n\rho \rceil}(\phi(y_i)-\mu)\xrightarrow{w}\bm W(\rho),
	\end{equation*}
	where $\bm W(\rho)$ is a Brownian motion in $\mathcal{H}$ and $\bm W(1)$ has the covariance operator $\Sigma$: $\mathcal{H}\rightarrow \mathcal{H}$, defined by
	$$
	\langle\Sigma \phi(y), \phi(y')\rangle
	=\E_{y''} \left[\langle \phi(y'')-\E\phi(y''), \phi(y')\rangle\langle \phi(y'')-\E\phi(y''), \phi(y')\rangle\right], \quad\forall y,y'\in \mathcal{H}.
	$$
	From the definition of $\phi(\cdot)$ and $\widetilde k$ in Section \ref{sec:connections}, we have
	$$
	\langle\Sigma \phi(y), \phi(y')\rangle
	=\E_{y''}\sum_{l,m}\phi_l(y)\phi_m(y')\left[\phi_l(y'')\phi_m(y'')\right]
	=\sum_l\lambda_l\phi_l(y)\phi_l(y'),
	$$
	as long as the last quantity is well-defined.
	Thus, we know $\bm W(\rho)=\left(\sqrt{\lambda_1}W_1(\rho),\sqrt{\lambda_2}W_2(\rho),\cdots\right)^\top$ where $W_l(\rho)$ and $W_m(\rho)$ are independent Brownian motions if $l\neq m$. Thus, a direct consequence of Corollary 5.2.1 in \citet{tewes2017change} is that 
	\begin{equation}\label{appendix_eq:S1_null_eq1}
	    \frac{1}{n}\left\|\frac{n-t}{n}\sum_{i=1}^{t}\phi(y_i)-\frac{t}{n}\sum_{i=t+1}^{n}\phi(y_i)\right\|^2
	\xrightarrow{w} \left\|\bm W(\rho)-\rho \bm W(1)\right\|^2
	=\sum_l\lambda_l W^0_l(\rho)^2,
	\end{equation}
	where $t=\lceil n\rho\rceil$.
	
	After some tedious calculations, we know
	\begin{equation}\label{appendix_eq:S1_null_eq2}
	\begin{split}
	&\frac{1}{n}\left\|\frac{n-t}{n}\sum_{i=1}^{t}\phi(y_i)-\frac{t}{n}\sum_{i=t+1}^{n}\phi(y_i)\right\|^2
	=\frac{(n-t)^2t^2}{n^3}\left[\frac{1}{t(n-t)} d_{A(t)}-\frac{1}{2t^2}d_{B_1(t)}-\frac{1}{2(n-t)^2}d_{B_2(t)}\right]
	\\=&\frac{(n-t)^2t^2}{n^3}\left[\frac{1}{t(n-t)} d_{A(t)}-\frac{1}{2t(t-1)}d_{B_1(t)}-\frac{1}{2(n-t)^2}d_{B_2(t)}\right]
	+\frac{(n-t)^2}{2n^3(t-1)}d_{B_1(t)}+\frac{t^2}{2n^3(n-t-1)}d_{B_1(t)}
	\\\stackrel{(a)}{=}&
	n\rho^2(1-\rho)^2\left[\bar d_{A(\lceil n\rho\rceil)}-\frac{1}{2}\bar d_{B_1(\lceil n\rho\rceil)}-\frac{1}{2}\bar d_{B_2(\lceil n\rho\rceil)}\right]
	+\frac{(n-t)^2t}{n^3(t-1)}\sum_{i=1}^{t}\|\hat \epsilon_i\|^2+\frac{t^2(n-t)}{n^3(n-t-1)}\sum_{i=t+1}^{n}\|\hat \epsilon_i\|^2,
	\end{split}
	\end{equation}
	where $\hat\epsilon_i=\phi(y_i)-\bar\phi(y)_{t-}$ for $i=1,2,\cdots,t$ and $\hat\epsilon_i=\phi(y_i)-\bar\phi(y)_{t+}$ for $i=t+1,t+2,\cdots,n$. Here (a) follows from the fact that $d_{B_1(t)}=2t\sum_{i=1}^{t}\|\epsilon_i-\bar\epsilon_{t-}\|^2$ and $d_{B_2(t)}=2(n-t)\sum_{i=t+1}^{n}\|\epsilon_i-\bar\epsilon_{t+}\|^2$. 
	
	Since $t\rightarrow\infty,n-t\rightarrow\infty$ as $n\rightarrow \infty$, we know that
	\begin{equation}\label{appendix_eq:S1_null_eq3}
	    \begin{split}
	        \frac{1}{t}\sum_{i=1}^{t}\|\hat \epsilon_i\|^2 -\E\|\epsilon\|^2
	&= \frac{1}{t}\sum_{i=1}^{t}\left(\|\hat \epsilon_i\|^2 -\|\epsilon_i\|^2\right) + \frac{1}{t}\sum_{i=1}^{t}\|\epsilon_i\|^2-\E\|\epsilon\|^2
	\\&= \left\|\frac{1}{t}\sum_{i=1}^t\epsilon_i\right\|^2 + \frac{1}{t}\sum_{i=1}^{t}\left(\|\epsilon_i\|^2-\E\|\epsilon\|^2\right)
	\xrightarrow{p}
	0,
	    \end{split}
	\end{equation}
	where the convergence in probability follows from Lemma \ref{lemma:deviation_exp_bounded} and law of large numbers (assumption 2 implies the boundedness of $\Var(\|\epsilon_i\|^2)$). 
	Similarly we have
	\begin{equation}\label{appendix_eq:S1_null_eq4}
	    \frac{1}{n-t}\sum_{i=t+1}^{n}\|\hat \epsilon_i\|^2 -\E\|\epsilon\|^2
	= \frac{1}{n-t}\sum_{i=t+1}^{n}\left(\|\hat \epsilon_i\|^2 -\|\epsilon_i\|^2\right) + \frac{1}{n-t}\sum_{i=t+1}^{n}\|\epsilon_i\|^2-\E\|\epsilon\|^2
	\xrightarrow{p}
	0.
	\end{equation}
	Since  $\E\|\epsilon\|^2=\sum_l\lambda_l$, combining Equation \eqref{appendix_eq:S1_null_eq1}, \eqref{appendix_eq:S1_null_eq2}, \eqref{appendix_eq:S1_null_eq3} and \eqref{appendix_eq:S1_null_eq4},
	we have 
	$$
	n\rho^2(1-\rho)^2\left[\bar d_{A(\lceil n\rho\rceil)}-\frac{1}{2}\bar d_{B_1(\lceil n\rho\rceil)}-\frac{1}{2}\bar d_{B_2(\lceil n\rho\rceil)}\right]
	\xrightarrow{w}
	\sum_l\lambda_l \left(W^0_l(\rho)^2 - \rho(1-\rho)\right).
	$$
	Now we want to make sure that $\sum_l\lambda_l \left(W^0_l(\rho)^2 - \rho(1-\rho)\right)$ is well defined. Notice that 
	$$
	\E\left[\sum_l\lambda_l \left(W^0_l(\rho)^2 - \rho(1-\rho)\right)\right]=0,
	$$
	$$
	\Var\left(\sum_l\lambda_l \left(W^0_l(\rho)^2 - \rho(1-\rho)\right)\right)=2\sum_l\lambda_l^2(1-\rho)^2\rho^2\stackrel{(b)}{<}+\infty,
	$$
	where (b) follows from assumption (3) because assumption (3) is equivalent to $$\int_y\int_{y'}\tilde k(y,y')dF_0(y)dF_0(y')<+\infty,$$
	which implies that $\sum_l\lambda_l^2<+\infty$.
\end{proof}

\subsubsection{On $S_2$}
Conclusion (b) is a direct consequence of Lemma \ref{lemma:S2_null_unscaled}.
\begin{lemma}\label{lemma:S2_null_unscaled}
	Under the null, if distance $d$ satisfies 
	
	(1) $d$ is a semi-metric of negative type,  
	
	(2) $\E_{y} \widetilde k(y,y)\leq M^2$,  and 
	
	(3) $\E_y|\widetilde k(y,y)-\E_{y}\widetilde k(y,y)|^{2+\delta}<+\infty$ for some $\delta>0$, 
	
	then, for any ${0<\rho<1}$,
	\begin{equation}
	\frac{\sqrt{n}\rho(1-\rho)}{2\widehat s_n}(\bar d_{B_1(\lceil n\rho\rceil)}-\bar d_{B_2(\lceil n\rho\rceil)})\xrightarrow{w} W^0(\rho),\quad n\rightarrow \infty.
	\end{equation}
\end{lemma}
\begin{proof}
	The proof is similar to the proof of Theorem 2.1 in \citet{lee2003cusum}.
	Write $t=\lceil n\rho\rceil$.
	\begin{equation*}\small
	\begin{split}
	&\frac{1}{\sqrt{n}\hsn}\frac{t(n-t)}{2n}\left(\db-\dbb\right)
	=\frac{1}{\sqrt{n}}\frac{s_n}{\hsn}\frac{1}{s_n}\left[ \frac{t}{t-1}\frac{n-t}{n}\sum_{i=1}^t\|\epsilon_i\|^2-\frac{n-t}{n-t-1}\frac{t}{n}\sum_{i=t+1}^n\|\epsilon_i\|^2\right]
	\\&+\frac{1}{\hsn}\frac{1}{\sqrt{n}}\frac{t}{t-1}\frac{n-t}{n}\sum_{i=1}^t\left(\|z_i-\bar z_{t-}\|^2-\|\epsilon_i\|^2\right)-\frac{1}{\hsn}\frac{1}{\sqrt{n}}\frac{n-t}{n-t-1}\frac{t}{n}\sum_{i=t+1}^n\left(\|z_i-\bar z_{t+}\|^2-\|\epsilon_i\|^2\right)
	\\&=U_1+U_2+U_3,
	\end{split}
	\end{equation*}	
	where
	\begin{align*}
	U_1&=\frac{1}{\sqrt{n}}\frac{s_n}{\hsn}\frac{1}{s_n}\left[ \frac{t}{t-1}\frac{n-t}{n}\sum_{i=1}^t\|\epsilon_i\|^2-\frac{n-t}{n-t-1}\frac{t}{n}\sum_{i=t+1}^n\|\epsilon_i\|^2\right],
	\\
	U_2&=\frac{1}{\hsn}\frac{1}{\sqrt{n}}\frac{t}{t-1}\frac{n-t}{n}\sum_{i=1}^t\left(\|z_i-\bar z_{t-}\|^2-\|\epsilon_i\|^2\right),
	\\U_3&=-\frac{1}{\hsn}\frac{1}{\sqrt{n}}\frac{n-t}{n-t-1}\frac{t}{n}\sum_{i=t+1}^n\left(\|z_i-\bar z_{t+}\|^2-\|\epsilon_i\|^2\right).
	\end{align*}
	Now we derive the asymptotic property of each of them separately.
	
	First we show that $U_1\xrightarrow{w} W^0\left(\rho\right)$. Notice that assumption (3) in Lemma \ref{lemma:S2_null_unscaled} implies $\E|\|\epsilon\|^2-\E\|\epsilon\|^2|^{2+\delta}<+\infty$ for some $\delta>0$. Thus, by treating $\|\epsilon_i\|$ as a (univariate) variable, it is a direct consequence from Lemma 3.1 of \citet{doukhan2012mixing} that 
	$$
	\frac{1}{\sqrt{n}s}\left[\frac{t}{t-1}\frac{n-t}{n}\sum_{i=1}^t\|\epsilon_i\|^2-\frac{n-t}{n-t-1}\frac{t}{n}\sum_{i=t+1}^n\|\epsilon_i\|^2\right]
	\xrightarrow{w} W^0(\rho).
	$$
	Combined with Lemma \ref{prop:proof_of_S2_null}, we know that $U_1(t)\xrightarrow{w}W^0(\rho)$.
	
	Then we show that $U_2\xrightarrow{P} 0$. Notice that
	\begin{equation*}
	\begin{split}
	U_2
	=&\frac{1}{\sqrt{n}}\sum_{i=1}^{t}\left(\|z_{i}-\bar{z}_{t-}\|^{2}-\|\epsilon_{i}\|^{2}\right)	
	=\frac{1}{\sqrt{n}}\sum_{i=1}^{t}\left(\|\epsilon_{i}-\bar{\epsilon}_{t-}\|^{2}-\|\epsilon_{i}\|^{2}\right)
	=-\frac{1}{\sqrt{n}}\frac{1}{t}\|\sum_{i=1}^t\epsilon_i\|^2
	\xrightarrow{p}0,
	\end{split}
	\end{equation*}
	where the convergence in probability follows Lemma \ref{lemma:deviation_exp_bounded}.
	
	The fact that $U_3\xrightarrow{p}0$ can proved in a similar way.
	
	To sum, this means 
	$$\frac{1}{\sqrt{n} s}\frac{t(n-t)}{2n}\left(\db-\dbb\right)\xrightarrow{w} W^0(\rho).$$ 
	From Lemma \ref{prop:proof_of_S2_null}, we have 
	$$
	\frac{\hsn}{s}\xrightarrow{p} 1.
	$$
	Thus,
	$$\frac{1}{\sqrt{n}\hsn}\frac{t(n-t)}{2n}\left(\db-\dbb\right)\xrightarrow{w} W^0(\rho).$$
\end{proof}

\begin{lemma}\label{prop:proof_of_S2_null}
	Suppose $d$ is a semi-metric of negative type. If there exists a positive constant $M$ s.t. $\max_{1\leq i\leq n}\E_{F^i}\left(\widetilde k(y,y)-\E_{F^i}\widetilde k(y,y)\right)^2\leq M^2$, we have that $\hsn\xrightarrow{p} s$. 
\end{lemma}
\begin{remark}
Notice that this Lemma holds for both the alternative and the null.
\end{remark}
\begin{proof}
	The proof follows from proof of Lemma 3.3 in \citet{lee2003cusum}.
	Notice that
	$
	\hsn^2=\frac{1}{n}\sum_{i=1}^n\left(\|\hat\epsilon_i\|^2-\hat m_n\right)^2
	$
	where $\hat m_n=\frac{1}{n}\sum_{i=1}^n\|\hat\epsilon_i\|^2$ and $\hat\epsilon_i=z_i-\bar z$. Denote $\wt\epsilon_i=z_i-\frac{1}{n}\sum_{i=1}^{n}\mu^i$, $\widetilde m_n=\frac{1}{n}\sum_{i=1}^n\|\wt\epsilon_i\|^2$. Denote $\tsn^2 =\frac{1}{n}\sum_{i=1}^{n}(\|\wt\epsilon_{i}\|^{2}-\wt m_{n})^{2}$.  Notice that $\hat\epsilon_i=\wt\epsilon_i-\bar\epsilon_n$. Notice that $\E\|\epsilon_i\|^2\leq M^2$ is equivalent to $\E\|\phi(y_i)\|^2\le M^2$.
	
	Notice that
	\begin{equation*}
	\begin{split}
	{\hsn}^2
	=&\frac{1}{n}\sum_{i=1}^{n}(\|\hat{\epsilon_{i}}\|^{2}-\|\wt\epsilon_{i}\|^{2}+\|\wt\epsilon_{i}\|^{2}-\wt m_{n}+\wt m_{n}-\hat{m}_n)^{2}
	\\=&\frac{1}{n}\sum_{i=1}^{n}(\|\wt\epsilon_{i}\|^{2}-\wt m_{n})^{2}+\left(\|\hat{\epsilon}_{i}\|^{2}-\|\wt\epsilon_{i}\|^{2}\right)^{2}
	+(\wt m_{n}-\hat{m}_n)^{2}
	\\&+2(\|\hat{\epsilon}_{i}\|^{2}-\|\wt\epsilon_{i}\|^{2})(\|\wt\epsilon_{i}\|^{2}-\wt m_{n})+
	2(\|\hat{\epsilon}_{i}\|^{2}-\|\wt\epsilon_{i}\|^{2})(\wt m_{n}-\hat{m}_n)+2(\|\wt\epsilon_{i}\|^{2}-\wt m_{n})(\wt m_{n}-\hat{m}_{n})
	\\=&\tsn^2+R_{1}+R_{2}+2R_{3}+2R_{4}+2R_{5},
	\end{split}
	\end{equation*}
	where
	\begin{align*}
	&R_1=\frac{1}{n}\sum_{i=1}^{n}\left(\|\hat{\epsilon}_{i}\|^{2}-\|\wt\epsilon_{i}\|^{2}\right)^{2},\quad
	R_2=\frac{1}{n}\sum_{i=1}^{n}(\wt m_{n}-\hat{m}_n)^{2},\quad
	R_3=\frac{1}{n}\sum_{i=1}^{n}2(\|\hat{\epsilon}_{i}\|^{2}-\|\wt\epsilon_{i}\|^{2})(\|\wt\epsilon_{i}\|^{2}-\wt m_{n}),\\
	&R_4=\frac{1}{n}\sum_{i=1}^{n}2(\|\hat{\epsilon}_{i}\|^{2}-\|\wt\epsilon_{i}\|^{2})(\wt m_{n}-\hat{m}_n),\quad
	R_5=\frac{1}{n}\sum_{i=1}^{n}2(\|\wt\epsilon_{i}\|^{2}-\wt m_{n})(\wt m_{n}-\hat{m}_{n}).
	\end{align*}
	Now we bound each of them separately.
	Firstly,
	\begin{equation*}
	\begin{split}
	R_{1}
	&=\frac{1}{n}\sum_{i=1}^{n}\left(\|\hat{\epsilon_{i}}\|^{2}-\|\wt\epsilon_{i}\|^{2}\right)^{2}
	=\frac{1}{n}\sum_{i=1}^{n}\left(\|\wt\epsilon_{i}-\bar\epsilon_{n}\|^{2}-\|\wt\epsilon_{i}\|^{2}\right)^{2}
	\\&=\frac{1}{n}\sum_{i=1}^{n}\left(\|\bar\epsilon_n\|^{2}+2\langle\wt\epsilon_{i},\,\bar\epsilon_n\rangle\right)^{2}
	\leq\frac{1}{n}\sum_{i=1}^{n}\left[2\|\bar\epsilon_n\|^{4}+2\left(2\langle\wt\epsilon_{i},\,\bar\epsilon_n\rangle\right)^{2}\right]
	\\&\leq 2\|\bar\epsilon_{n}\|^{4}+\left(\frac{8}{n}\sum_{i=1}^{n}\|\wt\epsilon_{i}\|^{2}\right)\times\|\bar\epsilon_n\|^{2}.
	\end{split}
	\end{equation*}
	Since 
	\begin{align*}
	 &\frac{1}{n}\sum_{i=1}^n\|\wt\epsilon_i\|^2
	 \xrightarrow{p}\E\|\wt\epsilon_i\|^2,\quad\text{where}
	 \\
	 &\E\|\wt\epsilon_i\|^2
	 =\E\left\|\mu^i-\frac{1}{n}\sum_{i=1}^n\mu^i+\epsilon_i\right\|^2
	 =\E\left\|\phi(y_i)-\frac{1}{n}\sum_{i=1}^n\mu^i\right\|^2
	 \leq 2\E\left\|\phi(y_i)\|^2+2\E\|\frac{1}{n}\sum_{i=1}^n\mu^i\right\|^2
	 \\&\leq2\E\|\phi(y_i)\|^2+2\E\left\|\frac{1}{n}\sum_{i=1}^n(\mu^i+\epsilon_i)\right\|^2
	 \leq C,   
	\end{align*}
	and $\|\bar\epsilon_n\|\xrightarrow{p}0$ (Lemma \ref{lemma:deviation_exp_bounded}), we have $R_1\xrightarrow{p}0$.
	Then,
	\begin{equation*}
	\begin{split}
	R_{2}&
	=\left(\frac{1}{n}\sum_{i=1}^{n}\left\|\hat{\epsilon}_i\right\|^{2}-\frac{1}{n}\sum_{i=1}^{n}\|\wt\epsilon_{i}\|^{2}\right)^{2}
	=\left(\frac{1}{n}\sum_{i=1}^{n}\|\bar\epsilon_n \|^{2}+\left\langle\frac{2}{n}\sum_{i=1}^{n}\epsilon_{i},\,-\bar \epsilon_n \right\rangle\right)^{2}
	=\|\bar \epsilon_n \|^{4}.
	\end{split}
	\end{equation*}
	Since $\|\bar \epsilon_n \|\xrightarrow{p}0$,  we have $R_2\xrightarrow{p} 0$. 
	Then,
	\begin{equation*}
	\begin{split}
	|R_{3}|
	&=\left|\frac{1}{n}\sum_{i=1}^{n}(\|\hat{\epsilon_{i}}\|^{2}-\|\wt\epsilon_{i}\|^{2})(\|\wt\epsilon_{i}\|^{2}-\wt m_{n})\right|
	\\&\leq
	\sqrt{\frac{1}{n}\sum_{i=1}^{n}\left(\|\hat{\epsilon_{i}}\|^{2}-\|\wt\epsilon_{i}\|^{2}\right)^{2}\frac{1}{n}\sum_{i=1}^{n}\left(\|\wt\epsilon_{i}\|^{2}-\wt m_{n}\right)^{2}}
	\\&\leq\sqrt{R_{1}\left(\frac{1}{n}\sum_{i=1}^{n}\|\wt\epsilon_{i}\|^{4}-\wt m_{n}^{2}\right)}
	=\sqrt{R_1\widetilde s_n^2}
	.
	\end{split}
	\end{equation*}
	Recall that when there is no change point, we have $\rho^*=1$. Then,
	from Law of Large Numbers, we have 
	\begin{equation}\label{appendix_eq:S2_null_eq1}
	\widetilde s_n^2
	=\frac{1}{n}\sum_{i=1}^{n}\|\wt\epsilon_{i}\|^{4}-\wt m_{n}^{2}
	\xrightarrow{p}
	s^2,
	\end{equation}
	where $s^2$ is a bounded positive constant because for any $i$,
	\begin{equation*}
	\begin{split}
	&\Var_{F^i}\left(\left\|\phi(y)-\frac{1}{n}\sum_{i=1}^n\mu^i\right\|^2\right)
	=\E_{F^i}\left(\|\phi(y)\|^2-\E_{F^i}\|\phi(y)\|^2-2\left\langle\frac{1}{n}\sum_{i=1}^n\mu^i,\epsilon\right\rangle\right)^2
	\\\leq &2\E_{F^i}\left(\|\phi(y)\|^2-\E_{F^i}\|\phi(y)\|^2\right)^2+2\E_{F^i}\left(2\left\langle\frac{1}{n}\sum_{i=1}^n\mu^i,\epsilon\right\rangle\right)^2
	\\\leq &2\E_{F^i}\left(\widetilde k(y,y)-\E_{F^i}\widetilde k(y,y)\right)^2+8\E_{F^i}\|\epsilon\|^2\left(\frac{1}{n}\sum_{i=1}^n\left\|\mu^i\right\|\right)^2
	\\\stackrel{(a)}{\le} &2\E_{F^i}\left(\widetilde k(y,y)-\E_{F^i}\widetilde k(y,y)\right)^2+8\E_{F^i}\widetilde k(y,y)\frac{1}{n}\sum_{i=1}^n\left\|\mu^i\right\|^2
	<+\infty.
	\end{split}
	\end{equation*}
	where (a) follows from the fact that $\E_{F^i}\|\epsilon\|^2=\E_{F^i}\|\phi(Y)-\E_{F^i}\phi(Y)\|^2\le\E_{F^i}\|\phi(Y)\|^2=\E_{F^i}\widetilde k(y,y),$ and $\left(\frac{1}{n}\sum_{i=1}^n\left\|\mu^i\right\|\right)^2\le \frac{1}{n}\sum_{i=1}^n\left\|\mu^i\right\|^2=\frac{1}{n}\sum_{i=1}^n\|\E_{F^i}\phi(Y)\|^2\le \frac{1}{n}\sum_{i=1}^n\E_{F^i}\|\phi(Y)\|^2=\frac{1}{n}\sum_{i=1}^n\E_{F^i}\widetilde k(Y,Y)$. 
	Combined with $R_1\xrightarrow{p}0$, we have $R_3\xrightarrow{p}0$.
	\begin{equation*}
	\begin{split}
	|R_{4}|
	=\left|\frac{2}{n}\sum_{i=1}^{n}(\|\hat{\epsilon}_i\|^{2}-\|\epsilon_{i}\|^{2})(\tilde m_{n}-\hat m_{n})\right|
	=2|\tilde m_{n}-\hat m_{n}|\times\left|\frac{1}{n}\sum_{i=1}^{n}\|\hat{\epsilon_{i}}\|^{2}-\E\|\epsilon_{i}\|^{2}\right|
	\xrightarrow{p}0,
	\end{split}
	\end{equation*}
	where the convergence in probability follows from the fact that $R_2=\left(\tilde m_n-\hat m_n\right)^2\xrightarrow{p} 0$, and  $\frac{1}{n}\sum_{i=1}^{n}\|\hat{\epsilon}_i\|^{2}-\E\|\epsilon_{i}\|^{2}\xrightarrow{p} 0$ (law of large numbers).
	$$
	|R_{5}|=\left|\frac{1}{n}\sum_{i=1}^{n}\left(\|\epsilon_{i}\|^{2}-\tilde m_{n}\right)\left(\tilde m_{n}-\hat m_{n}\right)\right|=0,
	$$
	where the last equality follows from the definition of $\tilde m_{n}$.
	
	Combining the above, we know that $\hsn-\tsn\xrightarrow{p}0$ and thus, $\frac{\tsn}{\hsn}\xrightarrow{P}1$. Equation \eqref{appendix_eq:S2_null_eq1} says $\tsn\xrightarrow{p}s^2$. Thus, we have $\hsn\xrightarrow{p}s^2$. This completes the proof.
\end{proof}

\subsection{Proof of Theorem \ref{thm:S1_power}}\label{appendix_sec:proof_S1_power}
Conclusion (1) is a  direct consequence of Theorem \ref{thm:S1_alternative} and (2) is a direct consequence of Theorem \ref{thm:S2_alternative}
.
\begin{theorem}[Alternative distribution for $S_1$]\label{thm:S1_alternative}
	In AMOC setting, under the alternative, if (1) $d$ is a semi-metric of negative type,  (2) there exists positive constant $M$ such that for all $i\in\{1,2,\cdots,n\}$, $\wt k(y_i,y_i)\leq M^2$ a.s., (3) there exists $\bm \Delta^{(1)}\in\mathcal{H}$ s.t. $\|\sqrt{n}(\mu_0-\mu_1)-\bm\Delta^{(1)}\|\rightarrow0$, then
	$$
	S_1
	\xrightarrow{w}
	\max_{\rho\in[\rho_0,\rho_1]}\left(\frac{\sum_l\left(\sqrt{\lambda_l}W^0(\rho)+\xi(\rho)\Delta_l^{(1)}\right)^2-\delta(\rho)}{\rho(1-\rho)}\right),
	$$
	where
	\begin{equation*}
	\delta(\rho)=
	\begin{cases}
	(1-\rho)\left((1-\rho)\rho^* v_0+(\rho-\rho^*+\rho\rho^*)v_1\right),\quad \text{if}\quad \rho^*\leq\rho
	\\
	\rho\left(\rho(1-\rho^*)v_1+(\rho\rho^*-2\rho+1)v_0\right),\quad \text{if}\quad \rho^*>\rho
	\end{cases},
	\end{equation*}
	and 
	$$\xi(\rho)=\begin{cases}
	\rho(1-\rho^*),\quad \text{if}\quad \rho\leq\rho^*\\
	(1-\rho)\rho^*,\quad\text{if}\quad \rho>\rho^*
	\end{cases}.
	$$
\end{theorem}

\begin{theorem}[Alternative distribution for $S_2$]\label{thm:S2_alternative}
	In AMOC setting, under the alternative, if (1) $d$ is a semi-metric of negative type, (2) there exists positive constant $M$ such that for all $i\in\{1,2,\cdots,n\}$, $\wt k(y_i,y_i)\leq M^2$ a.s., (3) $\sqrt{n}(v_0-v_1)\rightarrow \Delta_v^{(2)}$ and (4) $\sqrt{n}\|\mu_0-\mu_1\|^2\rightarrow\Delta_{\mu}^{(2)}$, then 
	$$S_2\xrightarrow{w}\max_{\rho\in[\rho_0,\rho_1]}\left(\frac{|G+\Delta^{(2)}|}{\sqrt{\rho(1-\rho)}}\right),$$
	where $G$ is some Gaussian process and 
	$$
	\Delta^{(2)}=
	\begin{cases}
	\frac{1}{s}\rho^*(1-\rho)\left(\Delta_v^{(2)}+\frac{\rho-\rho^*}{\rho}\Delta_{\mu}^{(2)}\right),\rho\geq\rho^*\\
	\frac{1}{s}(1-\rho^*)\rho\left(\Delta_v^{(2)}-\frac{\rho^*-\rho}{1-\rho}\Delta_{\mu}^{(2)}\right), \rho<\rho^*
	\end{cases}.
	$$
\end{theorem}

\subsubsection{Proof of Theorem \ref{thm:S1_alternative}}
\begin{proof}
Denote $t=\lceil n\rho\rceil$. For each $\rho\in[\rho_0,\rho_1]$, we show that
\begin{align}\label{eq_appendix:alternative_S2}
n\rho^2(1-\rho)^2\left[\bar d_{A(\lceil n\rho\rceil)}-\frac{1}{2}\bar d_{B_1(\lceil n\rho\rceil)}-\frac{1}{2}\bar d_{B_2(\lceil n\rho\rceil)}\right]
\xrightarrow{w}
\frac{\sum_l\left(\sqrt{\lambda_l}W^0(\rho)+\xi(\rho)\Delta_l^{(1)}\right)^2-\delta(\rho)}{\rho(1-\rho)}.    
\end{align}

In order to show Equation \eqref{eq_appendix:alternative_S2}, we utilize the following relationship:
\begin{equation*}
\begin{split}
&n\rho^2(1-\rho)^2\left[\bar d_{A(\lceil n\rho\rceil)}-\frac{1}{2}\bar d_{B_1(\lceil n\rho\rceil)}-\frac{1}{2}\bar d_{B_2(\lceil n\rho\rceil)}\right]
\\=&\frac{1}{n}\left\|\frac{n-t}{n}\sum_{i=1}^{t}\phi(y_i)-\frac{t}{n}\sum_{i=t+1}^{n}\phi(y_i)\right\|^2
-\frac{(n-t)^2t}{n^3(t-1)}\sum_{i=1}^{t}\|\hat \epsilon_i\|^2-\frac{t^2(n-t)}{n^3(n-t-1)}\sum_{i=t+1}^{n}\|\hat \epsilon_i\|^2
\\=&U_1-U_2-U_3,
\end{split}
\end{equation*}
where
\begin{align*}
    &U_1=\frac{1}{n}\left\|\frac{n-t}{n}\sum_{i=1}^{t}\phi(y_i)-\frac{t}{n}\sum_{i=t+1}^{n}\phi(y_i)\right\|^2,
    \\
    &U_2=\frac{(n-t)^2t}{n^3(t-1)}\sum_{i=1}^{t}\|\hat \epsilon_i\|^2,\quad
    U_3=\frac{t^2(n-t)}{n^3(n-t-1)}\sum_{i=t+1}^{n}\|\hat \epsilon_i\|^2.
\end{align*}
and $\hat\epsilon_i=\phi(y_i)-\bar\phi(y)_{t-}$ for $i=1,2,\cdots,t$ and $\hat\epsilon_i=\phi(y_i)-\bar\phi(y)_{t+}$ for $i=t+1,t+2,\cdots,n$. 
Now we derives asymptotic property for each of $U_1,U_2,U_3$ separately.

Firstly, from corollary 5.2.2 of \citet{tewes2017change}, if $\|\sqrt{n}(\mu_0-\mu_1)-\bm\Delta^{(1)}\|\rightarrow0$, then
we have
\begin{align}\label{eq_appendix:alternative_S2_U1}
U_1=\frac{1}{n}\left\|\frac{n-t}{n}\sum_{i=1}^{t}\phi(y_i)-\frac{t}{n}\sum_{i=t+1}^{n}\phi(y_i)\right\|^2
\xrightarrow{w}
\sum_l\left(\sqrt{\lambda_l}W^0(\rho)+\xi(\rho)\Delta_l^{(1)}\right)^2.
\end{align}

Secondly, write $\widetilde\epsilon_i=z_i-\bar \mu_{t-}$ for all $i=1,2,\cdots,t$, and notice that
\begin{equation}\label{eq_appendix:alternative_S2_term1}
\begin{split}
&\frac{1}{t}\sum_{i=1}^{t}(\|\hat \epsilon_i\|^2 -\E_{F^i}\|\widetilde\epsilon\|^2)
= \frac{1}{t}\sum_{i=1}^{t}\left(\|\hat \epsilon_i\|^2 -\|\widetilde\epsilon_i\|^2\right) + \frac{1}{t}\sum_{i=1}^{t}(\|\widetilde \epsilon_i\|^2-\E_{F^i}\|\widetilde \epsilon_i\|^2)
\\
=& \frac{1}{t}\sum_{i=1}^{t}\left(\|\bar \epsilon_t\|^2 -2\langle\mu^i-\bar\mu_{t-}+\epsilon_i,\bar\epsilon_t\rangle\right) 
\\&+ \frac{1}{t}\sum_{i=1}^{t}\left[\|\epsilon_i\|^2+\|\mu^i-\bar\mu_{t-}\|^2+2\langle\mu^i-\bar\mu_{t-},\epsilon_i\rangle-\E_{F^i}\left(\|\epsilon_i\|^2+\|\mu^i-\bar\mu_{t-}\|^2+2\langle\mu^i-\bar\mu_{t-},\epsilon_i\rangle\right)\right]
\\
=& -\|\bar \epsilon_t\|^2
+
\frac{1}{t}\sum_{i=1}^{t}\left[\|\epsilon_i\|^2-\E_{F^i}\|\epsilon_i\|^2\right] +\frac{1}{t}\sum_{i=1}^{t}\left[2\langle\mu^i-\bar\mu_{t-},\epsilon_i\rangle\right]
\xrightarrow{p}0,
\end{split}
\end{equation}
where the last convergence follows from the fact that $\|\bar\epsilon_t\|=\O_p(t^{-1/2})$. 

Similarly, denote $\widetilde\epsilon_i=z_i-\bar \mu_{t+}$ for all $i=t+1,\cdots,n$, we have
\begin{equation}\label{eq_appendix:alternative_S2_term2}
\begin{split}
&\frac{1}{n-t}\sum_{i=t+1}^{n}(\|\hat \epsilon_i\|^2 -\E_{F^i}\|\widetilde\epsilon\|^2)
\xrightarrow{p}0,
\end{split}
\end{equation}
Combining Equation \eqref{eq_appendix:alternative_S2_term1} and \eqref{eq_appendix:alternative_S2_term2}, we know that
\begin{align*}
    U_2
    &\xrightarrow{p}
    (1-\rho)^2\rho\lim_{t\rightarrow\infty}\frac{1}{t}\sum_{i=1}^t\E_{F^i}\|z_i-\bar\mu_{t-}\|^2
    =(1-\rho)^2\rho\lim_{t\rightarrow\infty}\frac{1}{t}\sum_{i=1}^t[\E_{F^i}\|\epsilon_i\|^2+\|\mu^i-\bar\mu_{t-}\|^2],
    \\
    U_3
    &\xrightarrow{p}
    \rho^2(1-\rho)\lim_{t\rightarrow\infty}\frac{1}{n-t}\sum_{i=t+1}^n\E_{F^i}\|z_i-\bar\mu_{t+}\|^2
    =\rho^2(1-\rho)\lim_{t\rightarrow\infty}\frac{1}{n-t}\sum_{i=t+1}^n[\E_{F^i}\|\epsilon_i\|^2+\|\mu^i-\bar\mu_{t+}\|^2].
\end{align*}

Thus, 
\begin{equation}\label{appendix_eq:proof_S1_power_deltarho}
\begin{split}
U_2+U_3\xrightarrow{p}&\,\,
\mathbb{I}(\rho^*\leq \rho)(1-\rho)\left(\rho^*(1-\rho)\frac{\rho-\rho^*}{\rho}\|\mu_0-\mu_1\|^2+(1-\rho)\rho^* v_0+(\rho-\rho^*+\rho\rho^*)v_1\right)
\\&+
\mathbb{I}(\rho^*> \rho)\rho\left(\rho(\rho^*-\rho)\frac{1-\rho^*}{1-\rho}\|\mu_0-\mu_1\|^2+\rho(1-\rho^*)v_1+(\rho\rho^*-2\rho+1)v_0\right)
\\&=\delta(\rho).
\end{split}
\end{equation}

Combining Equation \eqref{appendix_eq:proof_S1_power_deltarho} and \eqref{eq_appendix:alternative_S2_U1}, we know that Equation \eqref{eq_appendix:alternative_S2} holds. This leads to
$$
S_1
\xrightarrow{w}
\max_{\rho\in[\rho_0,\rho_1]}\left(\frac{\sum_l\left(\sqrt{\lambda_l}W^0(\rho)+\xi(\rho)\Delta_l^{(1)}\right)^2-\delta(\rho)}{\rho(1-\rho)}\right).
$$
\end{proof}

\subsubsection{Proof of Theorem \ref{thm:S2_alternative}}
\begin{proof}

Define $\wt\epsilon_i=z_i-\mu_{t-}$ for all $i=1,2,\cdots, t$ and $\wt\epsilon_i=z_i-\mu_{t+}$ for all $i=t+1,\cdots, n$. Notice that
\begin{equation*}\small
\begin{split}
&\frac{1}{\sqrt{n}\hsn}\frac{t(n-t)}{2n}\left(\db-\dbb\right)
=\frac{1}{\sqrt{n}}\frac{s}{\hsn}\frac{1}{s}\left[ \frac{t}{t-1}\frac{n-t}{n}\sum_{i=1}^t\|\wt\epsilon_i\|^2-\frac{n-t}{n-t-1}\frac{t}{n}\sum_{i=t+1}^n\|\wt\epsilon_i\|^2\right]
\\&+\frac{1}{\hsn}\frac{1}{\sqrt{n}}\frac{t}{t-1}\frac{n-t}{n}\sum_{i=1}^t\left(\|z_i-\bar z_t\|^2-\|\wt\epsilon_i\|^2\right)-\frac{1}{\hsn}\frac{1}{\sqrt{n}}\frac{n-t}{n-t-1}\frac{t}{n}\sum_{i=t+1}^n\left(\|z_i-\bar z_{n-t}\|^2-\|\wt\epsilon_i\|^2\right)
\\&=U_1+U_2-U_3,
\end{split}
\end{equation*}	
where
\begin{align*}
    &U_1=\frac{1}{\sqrt{n}}\frac{s}{\hsn}\frac{1}{s}\left[ \frac{t}{t-1}\frac{n-t}{n}\sum_{i=1}^t\|\wt\epsilon_i\|^2-\frac{n-t}{n-t-1}\frac{t}{n}\sum_{i=t+1}^n\|\wt\epsilon_i\|^2\right],
    \\
    &U_2=\frac{1}{\hsn}\frac{1}{\sqrt{n}}\frac{t}{t-1}\frac{n-t}{n}\sum_{i=1}^t\left(\|z_i-\bar z_t\|^2-\|\wt\epsilon_i\|^2\right),\quad
    U_3=\frac{1}{\hsn}\frac{1}{\sqrt{n}}\frac{n-t}{n-t-1}\frac{t}{n}\sum_{i=t+1}^n\left(\|z_i-\bar z_{n-t}\|^2-\|\wt\epsilon_i\|^2\right).
\end{align*}
Now we bound $U_1,U_2,U_3$ separately.

For $U_1$,
we have
\begin{equation*}
\begin{split}
U_1=&
\frac{1}{\sqrt{n}}\frac{s}{\hsn}\frac{1}{s}\left[ \frac{t}{t-1}\frac{n-t}{n}\sum_{i=1}^t\left(\|\wt\epsilon_i\|^2-\E\|\wt\epsilon_i\|^2\right)-\frac{n-t}{n-t-1}\frac{t}{n}\sum_{i=t+1}^n\left(\|\wt\epsilon_i\|^2-\E\|\wt\epsilon_i\|^2\right)\right] 
\\&+ \frac{1}{\sqrt{n}\hsn}\tau^*(1-\rho)\left(v_0-v_1+\frac{\rho-\rho^*}{\rho}\|\mu_0-\mu_1\|^2\right)\bm 1(\rho\geq\rho^*)
\\&+ \frac{1}{\sqrt{n}\hsn}(n-\tau^*)\rho\left(v_0-v_1-\frac{\rho^*-\rho}{1-\rho}\|\mu_0-\mu_1\|^2\right)\bm 1(\rho<\rho^*).
\end{split}
\end{equation*}
Since $d$ is bounded, and $\hsn\xrightarrow{p}s$ from Lemma \ref{prop:proof_of_S2_null}, we know 
$$
\frac{1}{\sqrt{n}}\frac{s}{\hsn}\frac{1}{s}\left[ \frac{t}{t-1}\frac{n-t}{n}\sum_{i=1}^t\left(\|\wt\epsilon_i\|^2-\E\|\wt\epsilon_i\|^2\right)-\frac{n-t}{n-t-1}\frac{t}{n}\sum_{i=t+1}^n\left(\|\wt\epsilon_i\|^2-\E\|\wt\epsilon_i\|^2\right)\right] 
\xrightarrow{w}
G,
$$
where $G$ is some Gaussian process.  Thus,
\begin{align}\label{appendix_eq:s2_power_u1}
    U_2\xrightarrow{w}G+\Delta^{(2)}.
\end{align}

For $U_2$, we have
\begin{equation}\label{eq_appendix:alternative_S2_u2}
\begin{split}
U_2&=\frac{1}{\hsn}\frac{1}{\sqrt{n}}\frac{t}{t-1}\frac{n-t}{n}\sum_{i=1}^t\left(\|z_i-\bar z_t\|^2-\|\wt\epsilon_i\|^2\right)
\\&=
\frac{1}{\hsn}\frac{1}{\sqrt{n}}\frac{t}{t-1}\frac{n-t}{n}\sum_{i=1}^t\left(\|z_i-\bar \mu_t-\bar\epsilon_t\|^2-\|z_i-\bar\mu_t\|^2\right)
\\&=\frac{1}{\hsn}\frac{1}{\sqrt{n}}\frac{t^2}{t-1}\frac{n-t}{n}(-\|\bar\epsilon_t\|^2)
\xrightarrow{p}0,
\end{split}
\end{equation}
because $\|\bar \epsilon_{t}\|=\O_p(t^{-1/2})$ from Lemma \ref{appendix_lemma:bound_noise_kernel_bd}. Similarly we have 
\begin{align}\label{eq_appendix:alternative_S2_u3}
 U_3\xrightarrow{p}0.   
\end{align}

Combining Equation \eqref{appendix_eq:s2_power_u1}, \eqref{eq_appendix:alternative_S2_u2}, \eqref{eq_appendix:alternative_S2_u3}, we get
$$S_2\xrightarrow{w}\max_{\rho\in[\rho_0,\rho_1]}\left(\frac{|G+\Delta^{(2)}|}{\sqrt{\rho(1-\rho)}}\right).$$
\end{proof}

\subsection{Proof of Theorem \ref{thm:consistency_S1}}

\subsubsection{For $S_1$}
\begin{proof}
We will utilize the following equation: $\forall t\in\{1,2,\cdots,n\}$,
\begin{equation*}
\begin{split}
T_1(t)
=& \frac{1}{t}\sum_{i=1}^{t}\wt k(y_i,y_i)+\frac{1}{n-t}\sum_{i=t+1}^{n}\wt k(y_i,y_i)-2\langle \bar\phi(y)_{t-},\bar\phi(y)_{t+}\rangle
\\&-\frac{1}{t-1}\left(\left\|(\bm{\mu}^*)_{t-}-(\bm{\mu})_{t-}\right\|^2\,+\,2\left\langle(\bm{\mu}^*)_{t-}-(\bm{\mu})_{t-},\,(\bm{\epsilon})_{t-}\right\rangle\,-\,\|\Pi(\bm{\epsilon})_{t-}\|^2\,+\,\|(\bm{\epsilon})_{t-}\|^2\right)
\\&-\frac{1}{n-t-1}\left(\|(\bm{\mu}^*)_{t+}-(\bm{\mu})_{t+}\|^2\,+\,2\left\langle(\bm{\mu}^*)_{t+}-(\bm{\mu})_{t+},\,(\bm{\epsilon})_{t+}\right\rangle\,-\,\|\Pi(\bm{\epsilon})_{t+}\|^2\,+\,\|(\bm{\epsilon})_{t+}\|^2\right).
\end{split}
\end{equation*}
Suppose $\tau^*<\wh\tau$. Plugging the above equation into the basic inequality $S_1(\wh \tau)\geq S_1(\tau^*)$,  we have
\begin{align}
0\,\leq &
\,\,\frac{\tau^*-\wh \tau}{n^2}\sum_{i=1}^{\tau^*}\wt k(y_i,y_i)+\frac{n-\wh \tau-\tau^*}{n^2}\sum_{i=\tau^*+1}^{\wh\tau}\wt k(y_i,y_i)+\frac{\wh \tau-\tau^*}{n^2}\sum_{i=\wh \tau+1}^{n}\wt k(y_i,y_i)
\label{appendix_eq:consistency_s1_line1}
\\&
+ \frac{\tau^*-\wh \tau}{n^2}\sum_{i=1}^{\tau^*}\|\epsilon_i\|^2+\frac{n-\wh \tau-\tau^*}{n^2}\sum_{i=\tau^*+1}^{\wh \tau}\|\epsilon_i\|^2+\frac{\wh \tau-\tau^*}{n^2}\sum_{i=\wh \tau+1}^{n}\|\epsilon_i\|^2
\label{appendix_eq:consistency_s1_line2}
\\&
-\frac{\wh \tau(n-\wh \tau)}{n^2(\wh \tau-1)}\left\|\firstmu-\firstempmu\right\|^2-\frac{\hat t(n-\hat t)}{n^2(n-\hat t-1)}\left\|\secondmu-\secondempmu\right\|^2
\label{appendix_eq:consistency_s1_line3}
\\&
-\frac{2\wh\tau(n-\wh\tau)}{n^2(\wh\tau-1)}\left\langle\firstmu-\firstempmu,\firstep\right\rangle-\frac{2\wh\tau(n-\wh\tau)}{n^2(n-\wh\tau-1)}\left\langle\secondmu-\secondempmu,\secondep\right\rangle
\label{appendix_eq:consistency_s1_line4}
\\&
+\frac{\wh\tau(n-\wh\tau)}{n^2(\wh\tau-1)}\left\|\Pi\firstempep\right\|^2+\frac{\wh\tau(n-\wh\tau)}{n^2(n-\wh\tau-1)}\left\|\Pi\secondempep\right\|^2
\label{appendix_eq:consistency_s1_line5}
\\&
-\frac{\tau^*(n-\tau^*)}{n^2(\tau^*-1)}\left\|\Pi\firstep\right\|^2-\frac{\tau^*(n-\tau^*)}{n^2(n-\tau^*-1)}\left\|\secondep\right\|^2
\label{appendix_eq:consistency_s1_line6}
\\&
-\frac{2\wh\tau(n-\wh\tau)}{n^2}\left\langle\bar\phi(y)_{\wh\tau-},\,\bar\phi(y)_{\wh\tau+}\right\rangle
+\frac{2\tau^*(n-\tau^*)}{n^2}\left\langle\bar\phi(y)_{\tau^*-},\,\bar\phi(y)_{\tau^*+}\right\rangle.
\label{appendix_eq:consistency_s1_line7}
\end{align}
Now we will bound each line seperately. 

Define $U_1=\frac{\tau^*-\wh\tau}{n^2}\tau^*\|\mu_0\|^2+\frac{(\wh\tau-\tau^*)(2n-2\wh\tau-\tau^*)}{n^2}\|\mu_1\|^2$.
Then for line \eqref{appendix_eq:consistency_s1_line1} and line \eqref{appendix_eq:consistency_s1_line2}, we have for any $x>0$, 
\begin{equation*}
\begin{split}
&\eqref{appendix_eq:consistency_s1_line1}+\eqref{appendix_eq:consistency_s1_line2}
\\=&
\frac{\tau^*-\wh\tau}{n^2}\sum_{i=1}^{\tau^*}\left(\|\mu_0\|^2+2\langle\mu_0,\epsilon_i\rangle\right)+\frac{n-\wh\tau-\tau^*}{n^2}\sum_{i=\tau^*+1}^{\wh\tau}\left(\|\mu_1\|^2+2\langle\mu_1,\epsilon_i\rangle\right)+\frac{\wh\tau-\tau^*}{n^2}\left(\|\mu_1\|^2+2\langle\mu_1,\epsilon_i\rangle\right)
\\=&
\frac{\tau^*-\wh\tau}{n^2}t^*\|\mu_0\|^2+\frac{(\wh\tau-\tau^*)(2n-2\wh\tau-\tau^*)}{n^2}\|\mu_1\|^2
\\&+2\frac{\tau^*-\wh\tau}{n^2}\sum_{i=1}^{\tau^*}\langle\mu_0,\epsilon_i\rangle
+2\frac{n-\wh\tau-\tau^*}{n^2}\sum_{i=\tau^*+1}^{\wh\tau}\langle\mu_1,\epsilon_i\rangle
+2\frac{\wh\tau-\tau^*}{n^2}\sum_{i=\wh\tau+1}^{n}\langle\mu_1,\epsilon_i\rangle,
\end{split}
\end{equation*}
where from Cauchy-Schwarz Inequality, Lemma \ref{appendix_lemma:bound_noise_kernel_bd} and the fact that $\|\mu^i\|^2\le \E_{F^i}\|\phi(y)\|^2=\E_{F^i}\wt k(y,y)\le M^2$, we know 
\begin{align*}
&\frac{\tau^*-\wh\tau}{n^2}\sum_{i=1}^{\tau^*}\langle\mu_0,\epsilon_i\rangle
\le
\frac{|\tau^*-\wh\tau|}{n^2}\|\mu_0\|\left\|\sum_{i=1}^{\tau^*}\epsilon_i\right\|
\le\frac{|\tau^*-\wh\tau|}{n^2}
\sqrt{\frac{14\tau^*}{3}}(\sqrt{x}+\sqrt{2})M^2,\quad\text{w.p. at least }1-2e^{-x}.
\\
&\frac{n-\wh\tau-\tau^*}{n^2}\sum_{i=\tau^*+1}^{\wh\tau}\langle\mu_1,\epsilon_i\rangle
\le 
\frac{|n-\tau^*-\wh\tau|}{n^2}\sqrt{\frac{14(\wh\tau-\tau^*)}{3}}(\sqrt{x}+\sqrt{2})M^2,
\quad\text{w.p. at least }1-2e^{-x}.
\\
&\frac{\wh\tau-\tau^*}{n^2}\sum_{i=\wh\tau+1}^{n}\langle\mu_1,\epsilon_i\rangle
\le 
\frac{|\wh\tau-\tau^*|}{n^2}\sqrt{\frac{14(n-\wh\tau)}{3}}(\sqrt{x}+\sqrt{2})M^2,
\quad\text{w.p. at least }1-2e^{-x}.
\end{align*}

Thus, with probability at least $1-6e^{-x}$, we have
\begin{equation}\label{appendix_eq:consistency_s1_term1}
\begin{split}
&\eqref{appendix_eq:consistency_s1_line1}+\eqref{appendix_eq:consistency_s1_line2}
\le
U_1+6M^2(\sqrt{x}+\sqrt{2})\frac{1}{\sqrt{n}}\sqrt{\frac{\wh\tau-\tau^*}{n}}.
\end{split}
\end{equation}

For line \eqref{appendix_eq:consistency_s1_line3} and line \eqref{appendix_eq:consistency_s1_line4}, from Proposition 3 of \citet{arlot2012kernel}, we have  for any $\theta>0$ and $x>0$, with probability at least $1-4e^{-x}$,
\begin{equation*}
\begin{split}
\eqref{appendix_eq:consistency_s1_line4}
\,\,\,{\leq}\,\,\,
&\frac{2\wh\tau(n-\wh\tau)}{n^2(\wh\tau-1)}\left[\theta\|\firstmu-\firstempmu\|^2+\left(\frac{3v_0}{2}+\frac{4M^2}{3}\right)x\right]
\\&+
\frac{2\wh\tau(n-\wh\tau)}{n^2(\wh\tau-1)}\left[\theta\|\secondmu-\secondempmu\|^2+\left(\frac{3v_1}{2}+\frac{4M^2}{3}\right)x\right].
\end{split}
\end{equation*}
Take $\theta=\frac{1}{3}$, we have with probability at least $1-4e^{-x}$,
\begin{equation}\label{appendix_eq:consistency_s1_term2}
\begin{split}
\eqref{appendix_eq:consistency_s1_line3}+\eqref{appendix_eq:consistency_s1_line4}
\stackrel{(a)}{\leq}
-\frac{1}{3}\frac{\wh\tau(n-\wh\tau)}{n^2(\wh\tau-1)}(\wh\tau-\tau^*)\|\mu_0-\mu_1\|^2+\frac{17}{2n}M^2x,
\end{split}
\end{equation}
where (a) follows from the fact that $v_0=\E_{F^i}\|\epsilon\|^2\le \E_{F^i}\|\mu^i+\epsilon_i\|^2=\E_{F^i}\wt k(y,y)\leq M^2$.

For line \eqref{appendix_eq:consistency_s1_line5} and line \eqref{appendix_eq:consistency_s1_line6}, from Lemma \ref{appendix_lemma:bound_noise_kernel_bd}, we know that for any $x>0$, with probability at least $1-2e^{-x}$, we have
\begin{equation}\label{appendix_eq:consistency_s1_term3}
\eqref{appendix_eq:consistency_s1_line5} + \eqref{appendix_eq:consistency_s1_line6}
\leq\left(\frac{\wh\tau(n-\wh\tau)}{n^2(\wh\tau-1)}+\frac{\wh\tau(n-\wh\tau)}{n^2(n-\wh\tau-1)}\right)\left(M^2+\frac{14}{3}M^2(x+2\sqrt{2x})\right)
\leq \frac{1}{n}\frac{28}{3}M^2(\sqrt{x}+\sqrt{2})^2.
\end{equation}

For line \eqref{appendix_eq:consistency_s1_line7}, we have
\begin{equation*}
\begin{split}
&-\frac{2\wh\tau(n-\wh\tau)}{n^2}\left\langle\bar\phi(y)_{\wh\tau-},\,\bar\phi(y)_{\wh\tau+}\right\rangle
=
-\frac{2}{n^2}\left\langle \tau^*\mu_0+(\wh\tau-\tau^*)\mu_1+\sum_{i=1}^{\wh\tau}\epsilon_i,\,(n-\wh\tau)\mu_1+\sum_{i=\wh\tau+1}^{n}\epsilon_i\right\rangle
\\=
&
-\frac{2}{n^2}\left[\tau^*(n-\wh\tau)\langle\mu_0,\mu_1\rangle+
(\wh\tau-\tau^*)(n-\wh\tau)\|\mu_1\|^2+
\left\langle \tau^*\mu_0+(\wh\tau-\tau^*)\mu_1,\,\sum_{i=\wh\tau+1}^{n}\epsilon_i\right\rangle \right]
\\&
-\frac{2}{n^2}\left[\left\langle \sum_{i=1}^{\wh\tau}\epsilon_i,\,(n-\wh\tau)\mu_1\right\rangle+
\left\langle \sum_{i=1}^{\wh\tau}\epsilon_i,\,\sum_{i=\wh\tau+1}^{n}\epsilon_i\right\rangle
\right],
\end{split}
\end{equation*}
and 
\begin{equation*}
\begin{split}
&\frac{2\tau^*(n-\tau^*)}{n^2}\langle\bar\phi(y)_{\tau^*-},\,\bar\phi(y)_{\tau^*+}\rangle
\\=
&
\frac{2}{n^2}\left[\tau^*(n-\tau^*)\langle\mu_0,\mu_1\rangle+
\left\langle \tau^*\mu_0,\,\sum_{i=\tau^*+1}^{n}\epsilon_i\right\rangle 
+\left\langle \sum_{i=1}^{\tau^*}\epsilon_i,\,(n-\tau^*)\mu_1\right\rangle+
\left\langle \sum_{i=1}^{\tau^*}\epsilon_i,\,\sum_{i=\tau^*+1}^{n}\epsilon_i\right\rangle
\right].
\end{split}
\end{equation*}
Thus, 
\begin{equation*}
\begin{split}
\eqref{appendix_eq:consistency_s1_line7}
=
&
\,\,\,\frac{2}{n^2}\left[\tau^*(\wh\tau-\tau^*)\left\langle\mu_0,\mu_1\right\rangle-
(\wh\tau-\tau^*)(n-\wh\tau)\|\mu_1\|^2+
\left\langle \sum_{i=1}^{\tau^*}\epsilon_i,\,(\wh\tau-\tau^*)\mu_1\right\rangle\right.
\\&\left.+
\left\langle \sum_{i=\tau^*+1}^{\wh\tau}\epsilon_i,\,\tau^*\mu_0-(n-\wh\tau)\mu_1\right\rangle
 \right]
\\&
+\frac{2}{n^2}\left[\left\langle \sum_{i=\wh\tau+1}^{n}\epsilon_i,\,-(\wh\tau-\tau^*)\mu_1\right\rangle+\left\langle \sum_{i=1}^{\tau^*}\epsilon_i,\,\sum_{i=\tau^*+1}^{n}\epsilon_i\right\rangle
-\left\langle \sum_{i=1}^{\wh\tau}\epsilon_i,\,\sum_{i=\wh\tau+1}^{n}\epsilon_i\right\rangle
\right],
\end{split}
\end{equation*}
where from Lemma \ref{appendix_lemma:bound_noise_kernel_bd}, the following holds: for any $x>0$,
\begin{align*}
&\frac{2}{n^2}\left\langle \sum_{i=1}^{\tau^*}\epsilon_i,\,(\wh\tau-\tau^*)\mu_1\right\rangle 
\le
\frac{2}{n^2}(\wh\tau-\tau^*)\|\mu_1\|\left\|\sum_{i=1}^{\tau^*}\epsilon_i\right\| 
\le
\frac{2}{n^2}(\wh\tau-\tau^*)M^2\sqrt{\frac{14\tau^*}{3}}(\sqrt{x}+\sqrt{2}),
\\&\quad\text{w.p. at least }1-2e^{-x}.
\\&
\frac{2}{n^2}\left\langle \sum_{i=\tau^*+1}^{\wh\tau}\epsilon_i,\,\tau^*\mu_0-(n-\wh\tau)\mu_1\right\rangle
\le 
\frac{2}{n^2}(\tau^*+n-\wh\tau)M^2\sqrt{\frac{14(\wh\tau-\tau^*)}{3}}(\sqrt{x}+\sqrt{2}),
\\&\quad\text{w.p. at least }1-2e^{-x}.
\\&
\frac{2}{n^2}\left\langle \sum_{i=\wh\tau+1}^{n}\epsilon_i,\,-(\wh\tau-\tau^*)\mu_1\right\rangle
\le 
\frac{2}{n^2}(\wh\tau-\tau^*)M^2\sqrt{\frac{14(\wh\tau-\tau^*)}{3}}(\sqrt{x}+\sqrt{2}),
\\&\quad\text{w.p. at least }1-2e^{-x}.
\\&
\frac{2}{n^2}\left\langle \sum_{i=1}^{\tau^*}\epsilon_i,\,\sum_{i=\tau^*+1}^{n}\epsilon_i\right\rangle
\le 
\frac{28}{3}\frac{1}{n^2}M^2\sqrt{(n-\tau^*)\tau^*}(\sqrt{x}+\sqrt{2})^2,
\quad\text{w.p. at least }1-2e^{-x}.
\\&
\frac{2}{n^2}\left\langle \sum_{i=1}^{\wh\tau}\epsilon_i,\,\sum_{i=\wh\tau+1}^{n}\epsilon_i\right\rangle
\le 
\frac{28}{3}\frac{1}{n^2}M^2\sqrt{(n-\wh\tau)\wh\tau}(\sqrt{x}+\sqrt{2})^2,
\quad\text{w.p. at least }1-2e^{-x}.
\end{align*}

Thus, we have with probability at least $1-10e^{-x}$,
\begin{equation}\label{appendix_eq:consistency_s1_term4}
\begin{split}
\eqref{appendix_eq:consistency_s1_line7}
\le&
\,\,\,\frac{2}{n^2}\left[\tau^*(\wh\tau-\tau^*)\langle\mu_0,\mu_1\rangle-
(\wh\tau-\tau^*)(n-\wh\tau)\|\mu_1\|^2+
(\wh\tau-\tau^*)\sqrt{\frac{14\tau^*}{3}}(\sqrt{x}+\sqrt{2})M^2
\right]
\\&
+\frac{2}{n^2}\left[
(\tau^*+n-\wh\tau)\sqrt{\frac{14(\wh\tau-\tau^*)}{3}}(\sqrt{x}+\sqrt{2})M^2
+(\wh\tau-\tau^*)\sqrt{\frac{14(\wh\tau-\tau^*)}{3}}(\sqrt{x}+\sqrt{2})M^2
\right]
\\&
+\frac{2}{n^2}\left[
\frac{14}{3}\sqrt{\tau^*(n-\tau^*)}M^2(\sqrt{x}+\sqrt{2})^2
+\frac{14}{3}\sqrt{\wh\tau(n-\wh\tau)}M^2(\sqrt{x}+\sqrt{2})^2
\right]
\\\stackrel{(a)}{\leq}&\,\,\,
\frac{2}{n^2}\tau^*(\wh\tau-\tau^*)\langle\mu_0,\mu_1\rangle-
\frac{2}{n^2}(\wh\tau-\tau^*)(n-\wh\tau)\|\mu_1\|^2
\\&+6\sqrt{\frac{14}{3}}M^2(\sqrt{x}+\sqrt{2})\frac{1}{\sqrt{n}}\sqrt{\frac{\wh\tau-\tau^*}{n}}+\frac{56}{3}\frac{M^2(\sqrt{x}+\sqrt{2})^2}{n},
\end{split}
\end{equation}
where (a) follows from the fact that $\frac{\wh\tau-\tau^*}{n}\leq \sqrt{\frac{\wh\tau-\tau^*}{n}}$.

Thus, combining \eqref{appendix_eq:consistency_s1_term1}, \eqref{appendix_eq:consistency_s1_term2}, \eqref{appendix_eq:consistency_s1_term3} and \eqref{appendix_eq:consistency_s1_term4}, we have: with probability at least $1-22e^{-x}$,
\begin{equation*}
\begin{split}
0\leq &
-\frac{1}{2}\frac{\wh\tau-\tau^*}{n}(1-\rho_1+\rho_0)\|\mu_0-\mu_1\|^2+
\frac{1}{n}\left(\frac{17}{2}x+28(\sqrt{x}+\sqrt{2})^2\right)M^2
\\&+12\sqrt{\frac{14}{3}}M^2(\sqrt{x}+\sqrt{2})\frac{1}{\sqrt{n}}\sqrt{\frac{\wh\tau-\tau^*}{n}}.
\end{split}
\end{equation*}
The solution to the above inequality is that
\begin{equation*}
\begin{split}
\frac{\wh\tau-\tau^*}{n}
\leq
\frac{1}{n}\left(\frac{12\sqrt{14/3}M^2(\sqrt{x}+\sqrt{2})}{(1-\rho_1+\rho_0)\|\mu_0-\mu_1\|^2}\right)^2
+
\frac{2}{n}\frac{17x/2+28(\sqrt{x}+\sqrt{2})^2}{(1-\rho_1+\rho_0)\|\mu_0-\mu_1\|^2}M^2.
\end{split}
\end{equation*}
When $\wh\tau<\tau^*$, we can prove a similar inequality. Thus for all $\wh\tau$, the desired conclusion holds.
\end{proof}

\subsubsection{For $S_2$}
\begin{proof}
First notice that $\forall t\in\{1,2,\cdots,n\}$, under the assumption that $\mu_0=\mu_1$, we have
\begin{equation*}
\bar d_{B_1(t)}-\bar d_{B_2(t)}
=\frac{2}{t-1}\left(-\|\Pi(\bm\epsilon)_{t-}\|^2+\|(\bm\epsilon)_{t-}\|^2\right)
-\frac{2}{n-t-1}\left(-\|\Pi(\bm\epsilon)_{t+}\|^2+\|(\bm\epsilon)_{t+}\|^2\right).
\end{equation*}
Plug this into the basic inequality, $S_2(\wh\tau)\geq S_1(\tau^*)$, we have
\begin{equation*}\footnotesize
\begin{split}
&\text{LHS}:=
\\&\sqrt{\frac{(n-\tau^*)\tau^*}{n^2}}
\times\left|
-\frac{1}{\tau^*(\tau^*-1)}\left\|\sum_{i=1}^{\tau^*}\epsilon_i\right\|^2+\frac{1}{\tau^*-1}\sum_{i=1}^{\tau^*}\|\epsilon_i\|^2
+\frac{1}{(n-\tau^*)(n-\tau^*-1)}\left\|\sum_{i=\tau^*+1}^{n}\epsilon_i\right\|^2-\frac{1}{n-\tau^*-1}\sum_{i=\tau^*+1}^{n}\|\epsilon_i\|^2
\right|
\\&\leq 
\sqrt{\frac{(n-\wh\tau)\wh\tau}{n^2}}
\times\left|
-\frac{1}{\wh\tau(\wh\tau-1)}\left\|\sum_{i=1}^{\wh\tau}\epsilon_i\right\|^2+\frac{1}{\wh\tau-1}\sum_{i=1}^{\wh\tau}\|\epsilon_i\|^2
+\frac{1}{(n-\wh\tau)(n-\wh\tau-1)}\left\|\sum_{i=\wh\tau+1}^{n}\epsilon_i\right\|^2-\frac{1}{n-\wh\tau-1}\sum_{i=\wh\tau+1}^{n}\|\epsilon_i\|^2
\right|
:= \text{RHS}.
\end{split}
\end{equation*}
We will deal with LHS and RHS separately.
Now suppose that $v_0>v_1$. When $\wh\tau\leq \tau^*$, we have, for any $x>0$, when $n$ is sufficiently large, the term inside the absolute value for RHS is positive with probability at least $1-e^{-x}$, i.e., with probability at least $1-e^{-x}$,
$$
-\frac{1}{\wh\tau(\wh\tau-1)}\left\|\sum_{i=1}^{\wh\tau}\epsilon_i\right\|^2+\frac{1}{\wh\tau-1}\sum_{i=1}^{\wh\tau}\|\epsilon_i\|^2
+\frac{1}{(n-\wh\tau)(n-\wh\tau-1)}\left\|\sum_{i=\wh\tau+1}^{n}\epsilon_i\right\|^2-\frac{1}{n-\wh\tau-1}\sum_{i=\wh\tau+1}^{n}\|\epsilon_i\|^2\ge0,
$$
because Lemma \ref{appendix_lemma:bound_noise_kernel_bd} implies
\begin{align*}
&-\frac{1}{\wh\tau(\wh\tau-1)}\left\|\sum_{i=1}^{\wh\tau}\epsilon_i\right\|^2+\frac{1}{\wh\tau-1}\sum_{i=1}^{\wh\tau}\|\epsilon_i\|^2
+\frac{1}{(n-\wh\tau)(n-\wh\tau-1)}\left\|\sum_{i=\wh\tau+1}^{n}\epsilon_i\right\|^2-\frac{1}{n-\wh\tau-1}\sum_{i=\wh\tau+1}^{n}\|\epsilon_i\|^2
\\&=
-\frac{1}{\wh\tau(\wh\tau-1)}\left\|\sum_{i=1}^{\wh\tau}\epsilon_i\right\|^2+\frac{1}{(n-\wh\tau)(n-\wh\tau-1)}\left\|\sum_{i=\wh\tau+1}^{n}\epsilon_i\right\|^2
\\&+\frac{1}{\wh\tau-1}\sum_{i=1}^{\wh\tau}(\|\epsilon_i\|^2-v_i)-\frac{1}{n-\wh\tau-1}\sum_{i=\wh\tau+1}^{n}(\|\epsilon_i\|^2-v_i)
+\frac{\wh\tau}{\wh\tau-1}(v_0-v_1)+\frac{v_1}{n-\wh\tau-1}
\\&
=\O_p\left(n^{-1/2}\right)+\frac{\wh\tau}{\wh\tau-1}(v_0-v_1)+\frac{v_1}{n-\wh\tau-1}.
\end{align*}
and $v_0-v_1,v_1>0$. 

Thus, for any $x>0$, when $n$ is sufficiently large, with probability at least $1-e^{-x}$, we have
\begin{equation*}\label{appendix_eq:S2_consistency_RHS}
\begin{split}
\text{RHS}
=&
\sqrt{\frac{(n-\wh\tau)\wh\tau}{n^2}}
\times\left(
-\frac{1}{\wh\tau(\wh\tau-1)}\left\|\sum_{i=1}^{\wh\tau}\epsilon_i\right\|^2+\frac{1}{(n-\wh\tau)(n-\wh\tau-1)}\left\|\sum_{i=\wh\tau+1}^{n}\epsilon_i\right\|^2
+\frac{1}{\wh\tau-1}\sum_{i=1}^{\wh\tau}(\|\epsilon_i\|^2-v_i)\right) 
\\&+\sqrt{\frac{(n-\wh\tau)\wh\tau}{n^2}}
\times\left(-\frac{1}{n-\wh\tau-1}\sum_{i=\wh\tau+1}^{n}(\|\epsilon_i\|^2-v_i)
+\frac{t^*}{\wh\tau-1}(v_0-v_1)+\frac{v_1}{n-\wh\tau-1}\right),
\end{split}
\end{equation*}
and consequently from Lemma \ref{appendix_lemma:bound_noise_kernel_bd}, we know that with probability at least $1-5e^{-x}$,
\begin{equation}\label{appendix_eq:S2_consistency_RHS}
\begin{split}
\text{RHS}
\le&
\sqrt{\frac{(n-\wh\tau)\wh\tau}{n^2}}
\times\left(
\frac{1}{(\wh\tau-1)}\frac{14}{3}(\sqrt{x}+\sqrt{2})^2M^2+\frac{1}{(n-\wh\tau-1)}\frac{14}{3}(\sqrt{x}+\sqrt{2})^2M^2\right.
\\&\left.+\frac{1}{\wh\tau-1}\sum_{i=1}^{\wh\tau}(\|\epsilon_i\|^2-v_i)\right) 
\\&+\sqrt{\frac{(n-\wh\tau)\wh\tau}{n^2}}
\times\left(-\frac{1}{n-\wh\tau-1}\sum_{i=\wh\tau+1}^{n}(\|\epsilon_i\|^2-v_i)
+\frac{\tau^*}{\wh\tau-1}(v_0-v_1)+\frac{v_1}{n-\wh\tau-1}\right),
\end{split}
\end{equation}

Similarly, for any $x>0$, when $n$ is sufficiently large, with probability at least $1-5e^{-x}$, we have
\begin{equation}\small\label{appendix_eq:S2_consistency_LHS}
\begin{split}
\text{LHS}
=&
\sqrt{\frac{(n-\tau^*)\tau^*}{n^2}}
\times
\left|
-\frac{1}{\tau^*(\tau^*-1)}\left\|\sum_{i=1}^{\tau^*}\epsilon_i\right\|^2+\frac{1}{(n-\tau^*)(n-\tau^*-1)}\left\|\sum_{i=\tau^*+1}^{n}\epsilon_i\right\|^2\right.
\\&\left.+\frac{1}{\tau^*-1}\sum_{i=1}^{\tau^*}(\|\epsilon_i\|^2-v_i)-\frac{1}{n-\tau^*-1}\sum_{i=\tau^*+1}^{n}(\|\epsilon_i\|^2-v_i)
+\frac{\tau^*}{\tau^*-1}v_0-\frac{(n-\tau^*)}{n-\tau^*-1}v_1
\right|
\\\geq&
\sqrt{\frac{(n-\tau^*)\tau^*}{n^2}}
\times\left(
-\frac{1}{(\tau^*-1)}\frac{14}{3}(\sqrt{x}+\sqrt{2})^2M^2-\frac{1}{(n-\tau^*-1)}\frac{14}{3}(\sqrt{x}+\sqrt{2})^2M^2\right.
\\&\left.-\frac{1}{\tau^*-1}\sum_{i=1}^{\tau^*}(\|\epsilon_i\|^2-v_i)\right) 
\\&+\sqrt{\frac{(n-\tau^*)\tau^*}{n^2}}
\times\left(-\frac{1}{n-\tau^*-1}\sum_{i=\tau^*+1}^{n}(\|\epsilon_i\|^2-v_i)
+(v_0-v_1)-\frac{v_1}{n-\tau^*-1}\right.
\\&\left.-\frac{v_0}{\tau^*-1}\right),
\end{split}
\end{equation}
Combing Equation \eqref{appendix_eq:S2_consistency_RHS} and Equation \eqref{appendix_eq:S2_consistency_LHS}, for any $x>0$, when $n$ is sufficiently large, with probability at least $1-10e^{-x}$, we have: 
\begin{align*}
&\left(\sqrt{\frac{(n-\tau^{*})\tau^{*}}{n^{2}}}-\sqrt{\frac{(n-\widehat{\tau})\widehat{\tau}}{n^{2}}}\frac{\tau^{*}}{\widehat{\tau}}\right)|v_{0}-v_{1}|
\\\leq&
\,\,\,\sqrt{\frac{(n-\widehat{\tau})\widehat{\tau}}{n^{2}}}\left[\frac{1}{n-\widehat{\tau}-1}v_{1}+\frac{14}{3}\left(\frac{1}{n_{0}}+\frac{1}{n-n_{1}}\right)(\sqrt{x}+\sqrt{2})^{2}M^{2}\right]
\\&+\sqrt{\frac{(n-\tau^{*})\tau^{*}}{n^{2}}}\left[\frac{1}{n-\tau^{*}-1}v_{1}+\frac{1}{\tau^{*}-1}v_{0}+\frac{14}{3}\left(\frac{1}{n_{0}}+\frac{1}{n-n_{1}}\right)(\sqrt{x}+\sqrt{2})^{2}M^{2}\right]
\\&+\sqrt{\frac{(n-\widehat{\tau})\widehat{\tau}}{n^{2}}}\left[\frac{\tau^{*}}{\widehat{\tau}-1}-\frac{\tau^{*}}{\widehat{\tau}}\right]|v_{0}-v_{1}|
+\left(\sqrt{\frac{(n-\widehat{\tau})\widehat{\tau}}{n^{2}}}\frac{1}{\widehat{\tau}-1}-\sqrt{\frac{(n-\tau^{*})\tau^{*}}{n^{2}}}\frac{1}{\tau^{*}-1}\right)\sum_{i=1}^{\tau^{*}}\left(\|\epsilon_{i}\|^{2}-v_{i}\right)
\\&+\left(\sqrt{\frac{(n-\widehat{\tau})\widehat{\tau}}{n^{2}}}\frac{1}{\widehat{\tau}-1}+\sqrt{\frac{(n-\tau^{*})\tau^{*}}{n^{2}}}\frac{1}{n-\tau^{*}-1}\right)\sum_{i=\tau^{*}+1}^{\widehat{\tau}}\left(\|\epsilon_{i}\|^{2}-v_{i}\right)
\\&-\left(\sqrt{\frac{(n-\widehat{\tau})\widehat{\tau}}{n^{2}}}\frac{1}{n-\widehat{\tau}-1}-\sqrt{\frac{(n-\tau^{*})\tau^{*}}{n^{2}}}\frac{1}{n-\tau^{*}-1}\right)\sum_{i=\widehat{\tau}+1}^{n}\left(\|\epsilon_{i}\|^{2}-v_{i}\right).
\end{align*}
Utilizing Proposition 4 in \cite{arlot2012kernel} and the fact that $0<v_0,v_1\leq M^2$, for any $x>0$, when $n$ is sufficiently large, with probability at least $1-16e^{-x}$, we have: 
\begin{equation*}
\begin{split}
&\left(\sqrt{\frac{(n-\tau^{*})\tau^{*}}{n^{2}}}-\sqrt{\frac{(n-\widehat{\tau})\widehat{\tau}}{n^{2}}}\frac{\tau^{*}}{\widehat{\tau}}\right)|v_{0}-v_{1}|
\\\le&\,\,\,\frac{1}{2}\left[\frac{2}{n-n_{1}}M^{2}+\frac{2}{n_{0}}M^{2}+\frac{28}{3}\left(\frac{1}{n_{0}}+\frac{1}{n-n_{1}}\right)(\sqrt{x}+\sqrt{2})^{2}M^{2}\right]
\\&+\frac{2}{n}\left(\sqrt{\frac{n-\widehat{\tau}}{\widehat{\tau}}}-\sqrt{\frac{n-\tau^{*}}{\tau^{*}}}\right)\left(\sqrt{2\tau^{*}M^{4}x}+\frac{M^{2}}{3}x\right)
\\&+\frac{2}{n}\left(\sqrt{\frac{n-\widehat{\tau}}{\widehat{\tau}}}+\sqrt{\frac{\tau^{*}}{n-\tau^{*}}}\right)\left(\sqrt{2(\widehat{\tau}-\tau^{*})M^{4}x}+\frac{M^{2}}{3}x\right)
\\&+\frac{2}{n}\left(\sqrt{\frac{\widehat{\tau}}{n-\widehat{\tau}}}-\sqrt{\frac{\tau^{*}}{n-\tau^{*}}}\right)\left(\sqrt{2(n-\widehat{\tau})M^{4}x}+\frac{M^{2}}{3}x\right)
\\\leq&\,\,\,\frac{1}{2}\left[\frac{2}{n-n_{1}}M^{2}+\frac{2}{n_{0}}M^{2}+\frac{28}{3}\left(\frac{1}{n_{0}}+\frac{1}{n-n_{1}}\right)(\sqrt{x}+\sqrt{2})^{2}M^{2}\right]
\\&+\frac{2}{n}\left(2\sqrt{\frac{n-\widehat{\tau}}{\widehat{\tau}}}-\sqrt{\frac{n-\tau^{*}}{\tau^{*}}}+\sqrt{\frac{\widehat{\tau}}{n-\widehat{\tau}}}\right)\frac{M^{2}}{3}x
\\&+\frac{2}{n}\left(\sqrt{\frac{n-\tau^{*}}{n-\widehat{\tau}}}-\sqrt{\frac{\tau^{*}}{\widehat{\tau}}}\right)\sqrt{n-\widehat{\tau}}\left(\sqrt{\frac{\widehat{\tau}}{n-\tau^{*}}}-1\right)\sqrt{2M^{4}x}
\\&+\frac{2}{n}\left(\sqrt{\frac{n-\widehat{\tau}}{\widehat{\tau}}}+\sqrt{\frac{\tau^{*}}{n-\tau^{*}}}\right)\sqrt{\frac{\widehat{\tau}-\tau^{*}}{n}}\sqrt{2M^{4}x}.
\end{split}
\end{equation*}

Let 
\begin{align*}
c_1=&\frac{1}{2}\left[\frac{2}{n-n_1}M^2+\frac{2}{n_0}M^2+\frac{28}{3}\left(\frac{1}{n_0}+\frac{1}{n-n_1}\right)(\sqrt{x}+\sqrt{2})^2M^2\right] 
\\&+ \frac{2}{n}\left(2\sqrt{\frac{n-\wh\tau}{\wh\tau}}-\sqrt{\frac{n-\tau^*}{\tau^*}}+\sqrt{\frac{\wh\tau}{n-\wh\tau}}\right)\frac{M^2}{3}x,    
\end{align*}
$$c_2=\sqrt{\frac{\tau^*(n-\wh\tau)}{n^2}}(v_0-v_1)-\frac{2}{n}\sqrt{n-\wh\tau}\left(\sqrt{\frac{\wh\tau}{n-\tau^*}}-1\right)\sqrt{2M^4 x},$$ $$w=\frac{\wh\tau-\tau^*}{n},\quad
c_3=\frac{2}{\sqrt{n}}\left(\sqrt{\frac{n-\wh\tau}{\wh\tau}}+\sqrt{\frac{\tau^*}{n-\tau^*}}\right)\sqrt{2M^4x}.$$
Then, we have: for any $x>0$, when $n$ is sufficiently large, with probability at least $1-16e^{-x}$,
\begin{equation*}
c_2\left(\sqrt{1+\frac{w}{1-w-\rho^*}}-\sqrt{1-\frac{w}{w+\rho^*}}\right)
\leq c_1+c_3\sqrt{w}.
\end{equation*}
Using the fact that 
\begin{equation*}
\begin{split}
c_2\left(\sqrt{1+\frac{w}{1-w-\rho^*}}-\sqrt{1-\frac{w}{w+\rho^*}}\right)
\geq 
c_2\left(\sqrt{1+w}-1\right)
\stackrel{(a)}{\geq}
c_2(\sqrt{2}-1)w,
\end{split}
\end{equation*}
where (a) follows from the fact that $\sqrt{1+w}\geq 1+(\sqrt{2}-1)w$ for any $w>0$. We have for any $x>0$, when $n$ is sufficiently large, with probability at least $1-16e^{-x}$,
\begin{equation}\label{appendix_eq:s2_consistency_inequality}
c_2(\sqrt{2}-1)w\leq c_1+c_3\sqrt{w}.
\end{equation}
Notice that
\begin{equation*}
\begin{split}
c_1\leq
&\frac{1}{n}\left[\frac{1}{1-\rho_1}M^2+\frac{1}{\rho_0}M^2+\frac{14}{3}\left(\frac{1}{\rho_0}+\frac{1}{1-\rho_1}\right)(\sqrt{x}+\sqrt{2})^2M^2\right] 
+ \frac{2}{n}\left(2\sqrt{\frac{1-\rho_0}{\rho_0}}+\sqrt{\frac{\rho_1}{1-\rho_1}}\right)\frac{M^2}{3}x,
\end{split}
\end{equation*}
\begin{equation*}
\begin{split}
c_2\geq
\sqrt{\rho^*(1-\rho_1)}(v_0-v_1)-\frac{2}{\sqrt{n}}\sqrt{1-\rho_0}\left(\sqrt{\frac{\rho_1}{1-\rho^*}}-1\right)\sqrt{2M^4 x},
\end{split}
\end{equation*}
\begin{equation*}
\begin{split}
c_3
\leq
\frac{2}{\sqrt{n}}\left(\sqrt{\frac{1-\rho_0}{\rho_0}}+\sqrt{\frac{\rho^*}{1-\rho^*}}\right)\sqrt{2M^4x}.
\end{split}
\end{equation*}
Thus, $c_2>c_3\sqrt{\rho^*(1-\rho_1)}(v_0-v_1)$ when $n$ is sufficiently large.

So Equation \eqref{appendix_eq:s2_consistency_inequality} yields: for any $x>0$, when $n$ is sufficiently large, with probability at least $1-16e^{-x}$,
\begin{equation*}
\begin{split}
w
\leq
&\left(\frac{c_3}{2c_2(\sqrt{2}-1)}+\sqrt{\frac{c_1}{c_2(\sqrt{2}-1)}+\frac{c_3^2}{4c_2^2(\sqrt{2}-1)^2}}\right)^2
\stackrel{(b)}{\leq}
\frac{c_3^2}{c_2^2(\sqrt{2}-1)^2}+\frac{2c_1}{c_2(\sqrt{2}-1)}
\\\leq&
\widetilde C_0 \frac{1}{n}\left(\frac{M^4x}{(v_0-v_1)^2}+\frac{M^2x}{|v_0-v_1|}\right),
\end{split}
\end{equation*}
where (b) utilizes the fact that $(x+y)^2\leq 2(x^2+y^2)$ for all $x,y$. 
\end{proof}

\subsection{Derivation for Higher Order Correction for $S_2$ (Equation \eqref{eq:S2_tilde}, \eqref{eq:S2_higher_order_correctness})}
\label{appendix_sec:S2_correction}
First we derive Equation \eqref{eq:S2_tilde}:
From Proof of Lemma \ref{lemma:S2_null_unscaled}, we know that 
\begin{equation*}
    \begin{split}
        \frac{\sqrt{n}\rho(1-\rho)}{2\hsn}T_2=U_1+U_2+U_3,
    \end{split}
\end{equation*}
where
$$U_1=\frac{1}{\sqrt{n}s_n}\left[\frac{t}{t-1}\frac{n-t}{n}\sum_{i=1}^t\|\epsilon_i\|^2-\frac{n-t}{n-t-1}\frac{t}{n}\sum_{i=t+1}^n\|\epsilon_i\|^2\right]\xrightarrow{w}W^0(\rho),$$ 
$$U_2=-\frac{1}{\sqrt{n}\hsn}\frac{t}{t-1}\left(1-\frac{t}{n}\right)\frac{1}{t}\sum_{i=1}^t\|\epsilon_i\|^2
\approx -\frac{1}{\sqrt{n}\hsn}\frac{t}{t-1}\left(1-\frac{t}{n}\right)\E\|\epsilon\|^2,$$
$$U_3=-\frac{1}{\sqrt{n}\hsn}\frac{n-t}{n-t-1}\frac{t}{n}\frac{1}{n-t}\sum_{i=t+1}^n\|\epsilon_i\|^2
\approx -\frac{1}{\sqrt{n}\hsn}\frac{n-t}{n-t-1}\frac{t}{n}\E\|\epsilon\|^2.$$
Thus, $U_2=\O_p(n^{-1/2})$ and $U_3=\O_p(n^{-1/2})$. By replacing the true mean with the estimated sample version, we can get Equation \eqref{eq:S2_tilde}, which cancels the $\O_p(n^{-1/2})$ term from $U_2$ and $U_3$.

Equation \eqref{eq:S2_higher_order_correctness} corrects for the $\O_p(n^{-1/2})$ coming from $U_1$:
Write $Z\left(\frac{t}{n}\right)=\frac{U_1}{\sqrt{\rho(1-\rho)}}$. Following \citet{chen2015graph}, we have
\begin{equation*}
    \begin{split}
        &\P\left(\max_{n_0\leq t\leq n_1}Z\left(\frac{t}{n}\right)>b\right)
        \\&=\frac{1}{b}\sum_{n_0\leq t\leq n_1}\int_{x=0}^{\infty}p\left(Z\left(\frac{t}{n}\right)=b+\frac{1}{b}x\right)\P\left(\max_{n_0\leq s\leq n_1}Z\left(\frac{s}{n}\right)<b\C Z\left(\frac{t}{n}\right)=b+\frac{1}{b}x\right)dx.
    \end{split}
\end{equation*}
We approximate $p\left(Z\left(\frac{t}{n}\right)=b+\frac{1}{b}x\right)$ using 3$^{rd}$ order Edgeworth Expansion and approximate $\P\left(\max_{n_0\leq s\leq n_1}Z\left(\frac{s}{n}\right)<b\C Z\left(\frac{t}{n}\right)\right)$ using a random walk.

Notice that $Z\left(\frac{t}{n}\right)$ is a sum of independent, non-identical distributed random variables, so we can apply Edgeworth Expansion and get when $b\rightarrow\infty$, $x^2/(2b^2)$ is negligible to $x$ and $x/b$ is neglibible to $b$, so let $V$ be the skewness of $Z\left(\frac{t}{n}\right)$, then
\begin{equation*}
    \begin{split}
        &p\left(Z\left(\frac{t}{n}\right)=b+\frac{1}{b}x\right)
        \approx \phi\left(b+\frac{1}{b}x\right)+\frac{1}{\sqrt{n}}\frac{1}{6}V\left[\left(b+\frac{x}{b}\right)^3-3\left(b+\frac{x}{b}\right)\right]\phi\left(b+\frac{x}{b}\right)
        \\= &\phi(b)e^{-x^2/(2b^2)-x}\left[1+\frac{1}{6\sqrt{n}}V\left(b+\frac{x}{b}\right)\left(\left(b+\frac{x}{b}\right)^2-3\right)\right]
        \\\approx & \phi(b)e^{-x}\left[1+\frac{1}{6\sqrt{n}}Vb\left(b^2-3\right)\right].
    \end{split}
\end{equation*}

To approximate $\P\left(\max_{n_0\leq s\leq n_1}Z\left(\frac{s}{n}\right)<b\C Z\left(\frac{t}{n}\right)\right)$, notice that 
\begin{equation*}
    \begin{split}
        b\left(Z\left(\frac{s}{n}\right)-Z\left(\frac{t}{n}\right)\right)\C Z\left(\frac{t}{n}\right)=b+\frac{x}{b}\sim N\left(-f_{t/n,-}'(0)|\frac{s-t}{n}|b^2,2f_{t/n,-}'(0)|\frac{s-t}{n}|b^2\right),
    \end{split}
\end{equation*}
where 
\begin{equation*}
    \begin{split}
        f_{x,-}'(0)=\frac{\partial}{\partial \delta}\text{corr}(Z\left(0\right), Z\left(\delta\right))\C_{\delta=0}
        =\frac{1}{2x(1-x)}.
    \end{split}
\end{equation*}

So let $W_m^{(t)}$ be a random walk with $W_1^{(t)}\sim N(\mu^{(t)},(\sigma^2)^{(t)})$ where $\mu^{(t)}=\frac{1}{n}f_{t/n,-}(0)b^2$ and $(\sigma^2)^{(t)}=2\mu^{(t)}$. We have
$$
\P\left(\max_{n_0\leq s\leq n_1}Z\left(\frac{s}{n}\right)<b\C Z\left(\frac{t}{n}\right)\right)
\approx \P\left(\max_{n_0\leq s<t}-W^{(t)}_{t-s}<-x\right)
\approx \P\left(\min_{m\geq 1}W^{(t)}_m>x\right).
$$
Combining the above, we have
\begin{equation*}
    \begin{split}
        &\P\left(\max_{n_0\leq t\leq n_1}Z\left(\frac{t}{n}\right)>b\right)
        \approx \frac{\phi(b)}{b}\left[1+\frac{1}{6\sqrt{n}}Vb\left(b^2-3\right)\right]\sum_{n_0\leq t\leq n_1}\int_{x=0}^{\infty}e^{-x}\P\left(\min_{m\geq 1}W^{(t)}_m>x\right)dx
        \\&=\frac{\phi(b)}{b}\left[1+\frac{1}{6\sqrt{n}}Vb\left(b^2-3\right)\right]\sum_{n_0\leq t\leq n_1}(-f_{t/n,-}(0))b^2\nu\left(b\sqrt{-\frac{2}{n}f_{t/n,-}(0)}\right)
        \\&=b\phi(b)\int_{\rho_)}^{\rho_1}\left[1+\frac{1}{6\sqrt{n}}Vb\left(b^2-3\right)\right](-f_{x,-}(0))\nu\left(b\sqrt{-2f_{x,-}(0)}\right)dx.
    \end{split}
\end{equation*}
And calculation shows (recall that $v=\E\|\epsilon\|^2$)
$$V=\frac{n-2t}{\sqrt{t(n-t)}}\times\text{Skewness of }\|\epsilon\|^2
=\frac{n-2t}{\sqrt{t(n-t)}}\times\frac{\E\left[\|\epsilon\|^6-3\|\epsilon\|^4v+3\|\epsilon\|^2v^2-v^3\right]}{\left[\E\left(\|\epsilon\|^4-2\|\epsilon\|^2v+v^2\right)\right]^{1.5}}.$$
From symmetry of $Z\left(\frac{t}{n}\right)$, we know $\P\left(\max_{n_0\leq t\leq n_1}|Z\left(\frac{t}{n}\right)|>b\right)=2\P\left(\max_{n_0\leq t\leq n_1}Z\left(\frac{t}{n}\right)>b\right)$. By replacing the corresponding true moments by the sample version in $K$, we get Equation \eqref{eq:S2_higher_order_correctness}.

\subsection{Proof of Proposition \ref{prop:S1_S2_ind}}\label{sec_appendix:indep_proof}

\begin{proof}
First notice that 
$$\lim_{n\rightarrow\infty}nT_1(t)=\lim_{n\rightarrow\infty}n\|\bar\phi(y)_{t-}-\bar\phi(y)_{t+}\|^2+o_p(1).$$ 
And
\begin{equation*}
    \begin{split}
        \lim_{n\rightarrow\infty}\sqrt{n}T_2(t)
        =&\lim_{n\rightarrow\infty}\sqrt{n}\left[\frac{1}{t}\sum_{i=1}^t\left(\|\phi(y_i)-\bar\phi(y)_{t-}\|^2-\E\|\phi(y_i)-\bar\phi(y)_{t-}\|^2\right)\right]
        \\&-\sqrt{n}\left[\frac{1}{n-t}\sum_{i=t+1}^n\left(\|\phi(y_i)-\bar\phi(y)_{t-}\|^2-\E\|\phi(y_i)-\bar\phi(y)_{t-}\|^2\right)\right]+o_p(1).
    \end{split}
\end{equation*}

Notice that
$$\sqrt{n}\left(\bar\phi_l(y)_{t-}-\E\bar\phi_l(y)_{t-}\right),$$ $$\sqrt{n}\left(\bar\phi_l(y)_{t+}-\E\bar\phi_l(y)_{t+}\right),$$ $$\frac{1}{\sqrt{n}}\sum_{i=1}^t\left[\|\phi(y_i)-\bar\phi(y)_{t-}\|^2-\E\|\phi(y_i)-\bar\phi(y)_{t-}\|^2\right],$$  $$\frac{1}{\sqrt{n}}\sum_{i=t+1}^n\left[\|\phi(y_i)-\bar\phi(y)_{t+}\|^2-\E\|\phi(y_i)-\bar\phi(y)_{t+}\|^2\right],$$ 
are all asymptotically Gaussian with mean 0. Thus, we only need to check that their covariance converges to 0. Since our data are i.i.d, we only need to check that the pairs: 
$$\sqrt{n}\left(\bar\phi_l(y)_{t-}-\E\bar\phi_l(y)_{t-}\right)
\quad \text{and}\quad \frac{1}{\sqrt{n}}\sum_{i=1}^t\|\phi(y_i)-\bar\phi(y)_{t-}\|^2-\E\|\phi(y_i)-\bar\phi(y_i)_{t-}\|^2,$$  $$\sqrt{n}\left(\bar\phi_l(y)_{t+}-\E\bar\phi_l(y)_{t+}\right)\quad\text{and}\quad \frac{1}{\sqrt{n}}\sum_{i=t+1}^n\|\phi(y_i)-\bar\phi(y)_{t+}\|^2-\E\|\phi(y_i)-\bar\phi(y)_{t+}\|^2,$$ 
are asymptotically uncorrelated for any $l$.
Since under the null, 
\begin{equation*}
    \begin{split}
        &\text{Cov}\left(\sqrt{n}\left(\bar\phi_l(y)_{t-}-\E\bar\phi_l(y)_{t-}\right),\frac{1}{\sqrt{n}}\sum_{i=1}^t\|\phi(y_i)-\bar\phi(y)_{t-}\|^2-\E\|\phi(y_i)-\bar\phi(y)_{t-}\|^2\right)
        \\&=\sum_{i=1}^t\E\|\phi(y_i)-\bar\phi(y)_{t-}\|^2\left(\bar\phi_l(y)_t-\E\bar\phi_l(y)_{t-}\right)
        =\left(1-\frac{1}{t}\right)\E\|\phi(y_i)\|^2\phi_l(y_i)=0.
    \end{split}
\end{equation*}
Similarly we have
\begin{equation*}
    \begin{split}
        &\text{Cov}\left(\sqrt{n}\left(\bar\phi_l(y)_{t+}-\E\bar\phi_l(y)_{t+}\right),\frac{1}{\sqrt{n}}\sum_{i=t+1}^n\|\phi(y_i)-\bar\phi(y)_{t+}\|^2-\E\|\phi(y_i)-\bar\phi(y)_{t+}\|^2\right)
        =0.
    \end{split}
\end{equation*}
Thus, we get the desired conclusion.
\end{proof}

\subsection{Theoretical Guarantees for $S_3$}\label{sec_appendix:S3_guarantees}
\begin{corollary}[Asymptotic null distribution for $S_3$]\label{thm:S3_null}
	Under $H_0$, if distance $d$ satisfies 
	
	(1) $d$ is a semi-metric of negative type, 
	
	
	(2) $\E_y|\widetilde k(y,y)|^{2+\delta}+\E_{y,y'}|\widetilde k(y,y')|^{2} <+\infty$ for some $\delta>0$,
	
	(3) $\E_y|\widetilde k(y,y)-\E \widetilde k(y,y)|^{2+\delta}<+\infty$ for some $\delta'>0$, 
	\\then as $n\rightarrow\infty$,
	\begin{equation}\label{eq:S3_asymp_null}
	S_3\xrightarrow{w} 
	\max_{\rho_0\leq \rho\leq \rho_1}\frac{1}{{\rho(1-\rho)}}\left(W^0(\rho)\right)^2.
	\end{equation}
\end{corollary}
\begin{proof}
	Corollary \ref{thm:S3_null} is a direct consequence of Theorem \ref{thm:S1_null}.
\end{proof}

\begin{corollary}[Localization Consistency for $S_3$]\label{thm:localization_S3}
	In AMOC setting, under $H_A$, suppose $d$ is a semi-metric of negative type, and there exists some positive constant $M$ such that for all $i\in\{1,\cdots,n\}$, $\wt k(y_i,y_i)\leq M^2$, a.s., then
	$$
	\left|\frac{\widehat \tau-\tau^*}{n}\right|=o_p(1),
	$$
	where $\widehat \tau$ is the estimated change point using statistics $S_3$.
\end{corollary}
\begin{proof}
From the proof of Theorem \ref{thm:S1_power} (Section \ref{appendix_sec:proof_S1_power}), we know that 
$$
T_1\xrightarrow{w}\frac{\|\bm W^0(\rho)+\xi(\rho)\bm\Delta^{(1)}\|^2-\delta(\rho)}{n\rho^2(1-\rho)^2}
=\left(\frac{\xi\left(\rho\right)}{\rho(1-\rho)}\right)^2\left\|\mu_0-\mu_1\right\|^2.
$$
It is obvious that the maximum of $\left(\frac{\xi\left(\rho\right)}{\rho(1-\rho)}\right)^2$ is obtained at $\rho=\rho^*$. From the Argmax Theorem, we know that $ \frac{\hat \tau}{n}-\frac{\tau^*}{n}=o_p(1)$.
\end{proof}

\begin{corollary}[Power for $S_3$]\label{thm:S3_power}
	In AMOC setting, if (1) $d$ is a semi-metric of negative type,   (2) there exists some positive constant $M$ such that for all $i\in\{1,\cdots,n\}$, $\wt k(y_i,y_i)\leq M^2$, a.s., then
	\begin{equation*}
	\P_{H_A}\left(S_3>q^{(2)}_{\alpha}\right)\rightarrow 1,\quad n\rightarrow\infty,
	\end{equation*}
	if either $\sqrt{n}\|\mu_0-\mu_1\|^2\rightarrow \infty$ or $\sqrt{n}|v_0-v_1|\rightarrow \infty$.
\end{corollary}
\begin{proof}
	Corollary \ref{thm:localization_S3} is a direct consequence of Theorem  \ref{thm:S1_alternative} and Theorem \ref{thm:S2_alternative}.
\end{proof}

\subsection{Proof of Theorem \ref{thm:multiple_CP}}
\begin{proof}
For $S_1$, using exactly the same techniques as in Theorem \ref{thm:S1_null}, it is easy to show that under the alternative, for all $t\in[l',r']$ where $l'=l+\lceil(r-l)\rho_0\rceil$, $r'=l+\lceil(r-l)\rho_1\rceil$, we have
$$
 T_1^{l,r}(t)
    \xrightarrow{p}
    \lim_{l,r\rightarrow\infty}\left\|\frac{1}{t-l}\sum_{i=l}^t \E_{F^i}\phi(y)
    \frac{1}{u-t}\sum_{i=t+1}^r \E_{F^i}\phi(y)\right\|^2
    \quad\text{uniformly}.
$$
This implies that under the alternative, as $n\rightarrow\infty$,
$$
S_1^{l,r}=\max_{l'\leq t'\leq r'}\frac{(t-l)(r-t)}{r-l}T_1^{l,r}(t)
\xrightarrow{p}
\lim_{l,r\rightarrow\infty}\frac{(t-l)(r-t)}{r-l}\left\|\frac{1}{t-l}\sum_{i=l}^t \E_{F^i}\phi(y)
\frac{1}{u-t}\sum_{i=t+1}^r \E_{F^i}\phi(y)\right\|^2
=\O(n).
$$
Notice that
$$
\P_{H_0}(S\ge s)\le\alpha \quad\Leftrightarrow\quad	s\ge q_\alpha,
$$
where
$$
q_{\alpha}=\text{upper }\alpha\text{-th quantile of }\max_{\rho_0\leq \rho\leq \rho_1}\frac{\sum_{l=1}^\infty\lambda_l\left(W_l^0(\rho)^2-\rho(1-\rho)\right)}{\rho(1-\rho)}.
$$
Thus, for $S_1$, when $n\rightarrow\infty$, $\alpha\rightarrow0$, $n\alpha\rightarrow\infty$,
we have
$$
H_A:\quad s=\O_p(n),q_{\alpha}=\O(\alpha^{-1}),
\quad\Rightarrow\quad
\P_{H_A}\left(s\ge q_\alpha\right)
\rightarrow 1.
$$
$$
H_0:\quad s=\O_p(1),q_{\alpha}=\O(\alpha^{-1}),
\quad\Rightarrow\quad
\P_{H_0}\left(s\ge q_\alpha\right)\rightarrow0.
$$
This ensures that
$$
\lim_{n\rightarrow\infty}\P(|\hat{\mathcal{D}}|=|\mathcal{D}|)=1.
$$

From Lemma 3.2 of \citet{rice2019consistency}, we know that 
$$\arg\max_{\rho\in[\rho_0,\rho_1]}\lim_{l,r\rightarrow\infty}\left\|\frac{1}{\lceil(r-l)\rho\rceil}\sum_{i=l}^{l+\lceil(r-l)\rho\rceil} \E_{F^i}\phi(y)
-\frac{1}{\lceil(r-l)(1-\rho)\rceil}\sum_{i=l+\lceil(r-l)\rho\rceil+1}^r \E_{F^i}\phi(y)\right\|^2\in\{\rho^*_1,\cdots,\rho^*_K\}.$$ 
Then from the uniform convergence of $T_1^{l,u}(t)$ and the argmax Theorem, we know that localization consistency holds.
    
For $S_2$, notice that
$$
 T_2^{l,r}(t)
 \xrightarrow{p}
 \lim_{l,r\rightarrow\infty}\left|\frac{1}{t-l}\sum_{i=l}^t \E_{F^i}\|\phi(y_i)-\E_{F^i}\phi(y)\|^2
    -\frac{1}{r-t}\sum_{i=t+1}^r \E_{F^i}\|\phi(y_i)-\E_{F^i}\phi(y)\|^2\right|
    \quad\text{uniformly}.
$$
Then, similar as for $S_1$, the conclusion for $S_2$ follows directly.
\end{proof}

\end{document}


\def\spacingset#1{\renewcommand{\baselinestretch}%
		{#1}\small\normalsize} \spacingset{1}

\appendix
\section{Technical Proofs}
This section includes proofs to all theoretical results in the main text. First, let us introduce some additional notations.

\subsection{Additional Notations}
Denote $z_i=\phi(y_i)$ where $\phi(\cdot)$ is defined in \eqref{def:phi}. Denote $\epsilon_i=z_i-\mu_i$ where $\mu_i:= \E z_i$. Let
$s_n^2=n^{-1}\sum_{i=1}^nVar(\|\epsilon_i\|^2)$, i.e.,  $s_n^2=\rho^*Var_{F_0}(\|\epsilon\|^2)+(1-\rho^*)Var_{F_1}(\|\epsilon\|^2)$ if there is a change point and $s_n^2=Var_{F_0}(\|\epsilon\|^2)$ if there is no change point.  Denote $\bar\epsilon_n=\frac{1}{n}\sum_{i=1}^{n}\epsilon_i$, $\bar\epsilon_t=\frac{1}{t}\sum_{i=1}^{t}\epsilon_i$ and $\bar\epsilon_{n-t}=\frac{1}{n-t}\sum_{i=1}^{n-t}\epsilon_i$. Denote $\overline{\phi(y)}_t=\frac{1}{t}\sum_{i=1}^{t}\phi(y_i)$, $\overline {\phi(y)}_{n-t}=\frac{1}{t}\sum_{i=t+1}^{n}\phi(y_i)$. Define $\left(\bm{\mu}^*\right)_t=(\mu_1,\mu_2,\cdots,\mu_t)^T\in\mathcal{H}^t$,  $\left(\bm{\mu}^*\right)_{n-t}=(\mu_{t+1},\mu_{t+2},\cdots,\mu_n)^T\in\mathcal{H}^{n-t}$,  $\left(\bm{\mu}\right)_t=(\overline \mu_t,\overline \mu_t,\cdots,\overline \mu_t)^T\in \mathcal{H}^t$ where $\overline \mu_t=\frac{1}{t}\sum_{i=1}^{t}\mu_i$, and  $\left(\bm{\mu}\right)_{n-t}=(\overline \mu_{n-t},\overline \mu_{n-t},\cdots,\overline \mu_{n-t})^T\in\mathcal{H}^{n-t}$ where $\overline \mu_{n-t}=\frac{1}{n-t}\sum_{i=t+1}^{n}\mu_i$. The norm in spaces $\mathcal{H}^t$ and $\mathcal{H}^{n-t}$ are defined in the same way as that in $\mathcal{H}$. Write  $\left(\bm\epsilon\right)_i^t=(\epsilon_1,\cdots,\epsilon_t)^T$. Following \cite{arlot2012kernel}, define the operator $\Pi$ s.t.
$\Pi \left(\bm\epsilon\right)_i^t:= \arg\min_{f=(f_1,f_2,\cdots,f_t)\in \mathcal{H}^t, f_1=f_2=\cdots=f_t}\{\|f-\left(\bm\epsilon\right)_i^t\|^2\}=\left(\frac{1}{t}\sum_{i=1}^{t}\epsilon_{i}, \frac{1}{t}\sum_{i=1}^{t}\epsilon_{i}, \cdots,\frac{1}{t}\sum_{i=1}^{t}\epsilon_{i}\right)^T$ and the proof of the second equality can be found in Appendix A.1 of \cite{arlot2012kernel}. With some slight abuse of notation, we write $\mu_0=\E_{y\sim F_0}\phi(y)$ and $\mu_1=\E_{y\sim F_1}\phi(y)$.  Define $v_0=\E_{y\sim F_0}\|\epsilon\|^2$ and $v_1=\E_{y\sim F_1}\|\epsilon\|^2$. We use $\xrightarrow{P}$ to denote convergence in probability. In a single change point, we write $\tau^*$ also as $t^*$.

\subsection{Some Useful Results}
The following are some useful results which will be utilized in later proofs.
\begin{lemma}[Proposition 1 from \cite{arlot2012kernel}]\label{appendix_lemma:bound_noise_kernel_bd}
	If (1) $\|z_i\|^2\leq M^2,\,\forall i$ and (2) $z_i$'s are independent, then, for any $x>0$, 
	$$
	\P\left( \left|\,\,\frac{1}{n}\left\|\sum_{i=1}^{n}\epsilon_i\right\|^2 - \frac{1}{n}\E \left\|\sum_{i=1}^{n}\epsilon_i\right\|^2 \,\,\right|\leq \frac{14M^2}{3}\left(x+2\sqrt{2x}\right)\right)\geq 1-2e^{-x}	
	$$
\end{lemma}
From Equation (19) in \cite{arlot2012kernel}, $\frac{1}{n}\E \|\sum_{i=1}^{n}\epsilon_i\|^2\leq CM^2$ where $C=1$ if there is no change point and $C=2$ if there is a change point. Thus, we have
\begin{equation}\label{eq:deviation_bounded}
\P\left( \frac{1}{\sqrt{n}}\left\|\sum_{i=1}^{n}\epsilon_i\right\|
\leq
\sqrt{CM^2+\frac{14M^2}{3}\left(x+2\sqrt{2x}\right)}
\right)\geq 1-2e^{-x}	
\end{equation}

\begin{lemma}[Lemma 7.10 from \cite{garreau2018consistent}]\label{lemma:deviation_exp_bounded}
	If (1) there exists a positive constant $V<+\infty$ s.t. $\max_{1\leq i\leq n}\E\|\epsilon_i\|^2\leq V$,  (2) $z_i$'s are independent, then, for any $x>0$,
	$$
	\P\left(\left\|\sum_{i=1}^{n}\epsilon_i\right\|\leq e^{x/2}\sqrt{nV}\right)\geq 1-e^{-x}
	$$
\end{lemma}

Now we are ready to present proofs to theories in the main article.
\subsection{Proof of Theorem \ref{thm:S1_null}}
Theorem \ref{thm:S1_null} is a direct consequence of the following Lemma:
\begin{lemma}\label{lemma:S1_null_unscaled}
	Under the null, if distance $d$ satisfies 
	
	(1) is a semi-metric of negative type, and 
	
	(2) $\E_y\left(\E_{y'}d(y,y')-\frac{1}{2}\E_{y,y'}d(y,y')\right)^{2+\delta}<+\infty$ for some $\delta>0$,
	
	(3) $\E_{y,y'}\left[d(y,y')-\E_{\tilde y}d(y',\tilde y)-\E_{\tilde y}d(y,\tilde y)+\E_{\tilde y,\tilde\tilde y}d(\tilde y,\tilde\tilde y)\right]^2<+\infty$, \\
	then for any ${0< \rho< 1}$, as $n\rightarrow \infty$,
	\begin{equation*}
	n\rho^2(1-\rho)^2\left(\bar d_A(\lceil n\rho\rceil)-\frac{1}{2}\bar d_{B_1}(\lceil n\rho\rceil )-\frac{1}{2}\bar d_{B_2}(\lceil n\rho\rceil )\right)\xrightarrow{w} \sum_{l=1}^\infty\lambda_l\left(W_l^0(\rho)^2-\rho(1-\rho)\right)
	\end{equation*}
	where $W_l^0(\cdot)$'s are independent Brownian bridges, $\lambda_l$'s are eigenvalues of $\widetilde k$ which is defined in Equation \eqref{eq:mercer_expansion}.
\end{lemma}
\begin{proof}
	From assumption (1), we know that Lemma \ref{lemma:distance_induced_kernel} and Equation \ref{eq:mercer_expansion} holds. 
	
	First notice that assumption (2) is equivalent to $\E\|\phi(y)\|^{4+\delta'}<\infty$ for some $\delta'>0$. Since our data are i.i.d under the null,
	from Theorem 16 in \cite{tewes2017change}, by directly treating $\{\phi(y_i)\}$ as the observatons in Hilbert space $\mathcal{H}$, we have
	\begin{equation*}
	\frac{1}{\sqrt{n}}\sum_{i=1}^{\lceil n\rho \rceil}(\phi(y_i)-\mu)\xrightarrow{w}\bm W(\rho)
	\end{equation*}
	where $\bm W(\rho)$ is a Brownian motion in $\mathcal{H}$ and $\bm W(1)$ has the covariance operator $\Sigma$: $\mathcal{H}\rightarrow \mathcal{H}$, defined by
	$$
	\langle\Sigma \phi(y), \phi(y')\rangle
	=\E_{y''} \left[\langle \phi(y'')-\E\phi(y''), \phi(y')\rangle\langle \phi(y'')-\E\phi(y''), \phi(y')\rangle\right], \quad\forall y,y'\in \mathcal{H}
	$$
	From the definition of $\phi(\cdot)$ and $\widetilde k$ in Section \ref{sec:connections}, and the fact that $\E\phi(y)=0$, we have
	$$
	\langle\Sigma \phi(y), \phi(y')\rangle
	=\E_{y''}\sum_{l,m}\phi_l(y)\phi_m(y')\left[\phi_l(y'')\phi_m(y'')\right]
	=\sum_l\lambda_l\phi_l(y_1)\phi_l(y_2),
	$$
	as long as the last quantity is well-defined.
	Thus, we know $\bm W(\rho)=\left(\sqrt{\lambda_1}W_1(\rho),\sqrt{\lambda_2}W_2(\rho),\cdots\right)^T$ where $W_l(\rho)$ and $W_m(\rho)$ are independent Brownian motions if $l\neq m$. Thus, a direct consequence of Corollary 5.2.1 in \cite{tewes2017change} is that 
	\begin{equation}\label{appendix_eq:S1_null_eq1}
	    \frac{1}{n}\left\|\frac{n-t}{n}\sum_{i=1}^{t}\phi(y_i)-\frac{t}{n}\sum_{i=t+1}^{n}\phi(y_i)\right\|^2
	\xrightarrow{w} \left\|\bm W(\rho)-\rho \bm W(1)\right\|^2
	=\sum_l\lambda_l W^0_l(\rho)^2
	\end{equation}
	After some tedious calculations, we know
	\begin{equation}\label{appendix_eq:S1_null_eq2}
	\begin{split}
	&\frac{1}{n}\left\|\frac{n-t}{n}\sum_{i=1}^{t}\phi(y_i)-\frac{t}{n}\sum_{i=t+1}^{n}\phi(y_i)\right\|^2
	=\frac{(n-t)^2t^2}{n^3}\left[\frac{1}{t(n-t)} d_{A(t)}-\frac{1}{2t^2}d_{B_1(t)}-\frac{1}{2(n-t)^2}d_{B_2(t)}\right]
	\\=&\frac{(n-t)^2t^2}{n^3}\left[\frac{1}{t(n-t)} d_{A(t)}-\frac{1}{2t(t-1)}d_{B_1(t)}-\frac{1}{2(n-t)^2}d_{B_2(t)}\right]
	+\frac{(n-t)^2}{2n^3(t-1)}d_{B_1(t)}+\frac{t^2}{2n^3(n-t-1)}d_{B_1(t)}
	\\\stackrel{(a)}{=}&
	n\rho^2(1-\rho)^2\left[\bar d_{A(\lceil n\rho\rceil)}-\frac{1}{2}\bar d_{B_1(\lceil n\rho\rceil)}-\frac{1}{2}\bar d_{B_2(\lceil n\rho\rceil)}\right]
	+\frac{(n-t)^2t}{n^3(t-1)}\sum_{i=1}^{t}\|\hat \epsilon_i\|^2+\frac{t^2(n-t)}{n^3(n-t-1)}\sum_{i=t+1}^{n}\|\hat \epsilon_i\|^2
	\end{split}
	\end{equation}
	where (a) follows from the fact that $d_{B_1(t)}=2t\sum_{i=1}^{t}\|\hat \epsilon_i\|^2$ with $\hat\epsilon_i=\epsilon_i-\frac{1}{t} \sum_{l=1}^t\epsilon_l$ and $d_{B_2(t)}=2(n-t)\sum_{i=t+1}^{n}\|\hat \epsilon_i\|^2$ with $\hat\epsilon_i=\epsilon_i-\frac{1}{n-t} \sum_{l=t+1}^{n}\epsilon_l$. 
	
	If $t\rightarrow\infty,n-t\rightarrow\infty$ as $n\rightarrow \infty$, we know that
	\begin{equation}\label{appendix_eq:S1_null_eq3}
	    \begin{split}
	        \frac{1}{t}\sum_{i=1}^{t}\|\hat \epsilon_i\|^2 -\E\|\epsilon\|^2
	&= \frac{1}{t}\sum_{i=1}^{t}\left(\|\hat \epsilon_i\|^2 -\|\epsilon_i\|^2\right) + \frac{1}{t}\sum_{i=1}^{t}\|\epsilon_i\|^2-\E\|\epsilon\|^2
	\\&= \left\|\frac{1}{t}\sum_{i=1}^t\epsilon_i\right\|^2 + \frac{1}{t}\sum_{i=1}^{t}\left(\|\epsilon_i\|^2-\E\|\epsilon\|^2\right)
	\xrightarrow{P}
	0
	    \end{split}
	\end{equation}
	where the last convergence in probability follows from Lemma \ref{lemma:deviation_exp_bounded} (assumption 3 implies that condition (1) in Lemma \ref{lemma:deviation_exp_bounded} is satisfied) and law of large numbers (assumption 2 implies that $Var(\|\epsilon_i\|^2)<+\infty$). 
	Similarly we have
	\begin{equation}\label{appendix_eq:S1_null_eq4}
	    \frac{1}{n-t}\sum_{i=t+1}^{n}\|\hat \epsilon_i\|^2 -\E\|\epsilon\|^2
	= \frac{1}{n-t}\sum_{i=t+1}^{n}\left(\|\hat \epsilon_i\|^2 -\|\epsilon_i\|^2\right) + \frac{1}{n-t}\sum_{i=t+1}^{n}\|\epsilon_i\|^2-\E\|\epsilon\|^2
	\xrightarrow{P}
	0
	\end{equation}
	Since  $\E\|\epsilon\|^2=\sum_l\lambda_l$, combining Equation \eqref{appendix_eq:S1_null_eq1}, \eqref{appendix_eq:S1_null_eq2}, \eqref{appendix_eq:S1_null_eq3} and \eqref{appendix_eq:S1_null_eq4},
	we have 
	$$
	n\rho^2(1-\rho)^2\left[\bar d_{A(\lceil n\rho\rceil)}-\frac{1}{2}\bar d_{B_1(\lceil n\rho\rceil)}-\frac{1}{2}\bar d_{B_2(\lceil n\rho\rceil)}\right]
	\xrightarrow{w}
	\sum_l\lambda_l \left(W^0_l(\rho)^2 - \rho(1-\rho)\right)
	$$
	It is easy to check that this approximation also holds for $t=\O(1),n\rightarrow\infty$ or $n-t=\O(1),n\rightarrow\infty$, where we have $\frac{(n-t)^2t}{n^3(t-1)}\sum_{i=1}^{t}\|\hat \epsilon_i\|^2+\frac{t^2(n-t)}{n^3(n-t-1)}\sum_{i=t+1}^{n}\|\hat \epsilon_i\|^2\xrightarrow{P} 0=\rho(1-\rho)$.
	
	Now we want to make sure that $\sum_l\lambda_l \left(W^0_l(\rho)^2 - \rho(1-\rho)\right)$ is well defined. Notice that 
	$$
	\E\left[\sum_l\lambda_l \left(W^0_l(\rho)^2 - \rho(1-\rho)\right)\right]=0
	$$
	$$
	Var\left(\sum_l\lambda_l \left(W^0_l(\rho)^2 - \rho(1-\rho)\right)\right)=2\sum_l\lambda_l^2(1-\rho)^2\rho^2\stackrel{(b)}{<}+\infty
	$$
	where (b) follows from assumption (3) because assumption (3) is equivalent to $$\int_y\int_{y'}\tilde k(y,y')dF_0(y)dF_0(y')<+\infty$$
	which implies that $\sum_l\lambda_l^2<+\infty$.
\end{proof}

\subsection{Proof of Theorem \ref{thm:S2_null}}
Theorem \ref{thm:S2_null} is a direct consequence of Lemma \ref{lemma:S2_null_unscaled}.
\begin{lemma}\label{lemma:S2_null_unscaled}
	Under the null, if distance $d$ satisfies 
	
	(1) is a semi-metric of negative type,  
	
	(2) $\E d(y_i,y_j)\leq M^2,\,\forall i,j$,  and 
	
	(3) $\E_y|\E_{y'}d(y,y')-\E_{y,y'}d(y,y')|^{2+\delta}<+\infty$ for some $\delta>0$, \\
	then, for any ${0<\rho<1}$,
	\begin{equation}
	\frac{\sqrt{n}\rho(1-\rho)}{2\widehat s_n}(\bar d_{B_1(\lceil n\rho\rceil)}-\bar d_{B_2(\lceil n\rho\rceil)})\xrightarrow{w} W^0(\rho),\,n\rightarrow \infty
	\end{equation}
	where $W^0(\cdot)$ is a Brownian bridge.
\end{lemma}
\begin{proof}
	The proof follows from proof of Theorem 2.1 in \cite{lee2003cusum}.
	Write $t=\lceil n\rho\rceil$.
	\begin{equation*}\small
	\begin{split}
	&\frac{1}{\sqrt{n}\hsn}\frac{t(n-t)}{2n}\left(\db-\dbb\right)
	=\frac{1}{\sqrt{n}}\frac{s_n}{\hsn}\frac{1}{s_n}\left[ \frac{t}{t-1}\frac{n-t}{n}\sum_{i=1}^t\|\epsilon_i\|^2-\frac{n-t}{n-t-1}\frac{t}{n}\sum_{i=t+1}^n\|\epsilon_i\|^2\right]
	\\&+\frac{1}{\hsn}\frac{1}{\sqrt{n}}\frac{t}{t-1}\frac{n-t}{n}\sum_{i=1}^t\left(\|z_i-\bar z_t\|^2-\|\epsilon_i\|^2\right)-\frac{1}{\hsn}\frac{1}{\sqrt{n}}\frac{n-t}{n-t-1}\frac{t}{n}\sum_{i=t+1}^n\left(\|z_i-\bar z_{n-t}\|^2-\|\epsilon_i\|^2\right)
	\\&=U_1+U_2+U_3
	\end{split}
	\end{equation*}	
	
	Suppose that $t,n-t\rightarrow\infty$ when $n\rightarrow\infty$.
	
	First we show that $U_1\xrightarrow{w} W^0\left(\rho\right)$. Notice that assumption (3) in Lemma \ref{lemma:S2_null_unscaled} implies $\E|\|\epsilon\|^2-\E\|\epsilon\|^2|^{2+\delta}<+\infty$ for some $\delta>0$. Thus, by treating $\|\epsilon_i\|$ as a (univariate) variable, it is a direct consequence from Lemma 3.1 of \cite{doukhan2012mixing} that 
	$$
	\frac{1}{\sqrt{n}s_n}\left[\frac{t}{t-1}\frac{n-t}{n}\sum_{i=1}^t\|\epsilon_i\|^2-\frac{n-t}{n-t-1}\frac{t}{n}\sum_{i=t+1}^n\|\epsilon_i\|^2\right]
	\xrightarrow{w} W^0(\rho)
	$$
	Combined with Lemma \ref{prop:proof_of_S2_null}, we know that $U_1(t)\xrightarrow{w}W^0(\rho)$.
	
	Then we show that $U_2\xrightarrow{P} 0$.
	\begin{equation*}
	\begin{split}
	U_2=&\frac{1}{\sqrt{n}}\sum_{i=1}^{t}\left(\|z_{i}-\bar{z}_{t}\|^{2}-\|\epsilon_{i}\|^{2}\right)	=\frac{1}{\sqrt{n}}\sum_{i=1}^{t}\left(\|\epsilon_{i}-\bar{\epsilon}_{t}\|^{2}-\|\epsilon_{i}\|^{2}\right)
	=-\frac{1}{\sqrt{n}}\frac{1}{t}\|\sum_{i=1}^t\epsilon_i\|^2
	\xrightarrow{P}0
	\end{split}
	\end{equation*}
	where the last $\xrightarrow{P}$ follows from assumption (2) in Lemma \ref{lemma:S2_null_unscaled} and Lemma \ref{lemma:deviation_exp_bounded}.
	
	The fact that $T_3\xrightarrow{P}0$ can proved in a similar way.
	
	To sum, this means $\frac{1}{\sqrt{n}\hsn}\frac{t(n-t)}{2n}\left(\db-\dbb\right)\xrightarrow{w} W^0(\rho)$ if $t,n-t\rightarrow\infty$ when $n\rightarrow\infty$. And it is easy to show that this approximation also holds for $t=\O(1),n\rightarrow\infty$ or $n-t=\O(1),n\rightarrow\infty$. Thus, we conclude 
	$$\frac{1}{\sqrt{n}\hsn}\frac{t(n-t)}{2n}\left(\db-\dbb\right)\xrightarrow{w} W^0(\rho).$$
\end{proof}

\begin{lemma}\label{prop:proof_of_S2_null}
	If distance $d$ satisfies $Var(\|\phi(y_i)\|^2)\leq M^2$ for any $i$, we have that $\hsn\xrightarrow{P}c^*$ where $c^*$ is a bounded positive constant. Moreover, under the null,  we have $\frac{s_n}{\hsn}\xrightarrow{P} 1$. 
\end{lemma}
\begin{proof}
	The proof follows from proof of Lemma 3.3 in \cite{lee2003cusum}.
	Notice that
	$
	\hsn^2=\frac{1}{n}\sum_{i=1}^n\left(\|\hat\epsilon_i\|^2-\hat\mu_n\right)^2
	$
	where $\hat\mu_n=\frac{1}{n}\sum_{i=1}^n\|\hat\epsilon_i\|^2$ and $\hat\epsilon_i=z_i-\bar z_n$. Denote $\wt\epsilon_i=\mu_i-\frac{1}{n}\sum_{i=1}^{n}\mu_i+\epsilon_i$, $\widetilde\mu_n=\frac{1}{n}\sum_{i=1}^n\|\wt\epsilon_i\|^2$. Denote $\tsn^2 =\frac{1}{n}\sum_{i=1}^{n}(\|\wt\epsilon_{i}\|^{2}-\wt\mu_{n})^{2}$.  Notice that $\hat\epsilon_i=\wt\epsilon_i-\bar\epsilon_n$. Notice that $Var(\|\phi(y_i)\|^2)\leq M^2$ is equivalent to $\E\|\phi(y)\|^2+\E\|\phi(y)-\E\|\phi(y)\|\|^2<\infty$.
	
	Notice that
	\begin{equation*}
	\begin{split}
	{\hsn}^2
	=&\frac{1}{n}\sum_{i=1}^{n}(\|\hat{\epsilon_{i}}\|^{2}-\|\wt\epsilon_{i}\|^{2}+\|\wt\epsilon_{i}\|^{2}-\wt\mu_{n}+\wt\mu_{n}-\hat{\mu_{n}})^{2}
	\\=&\frac{1}{n}\sum_{i=1}^{n}(\|\wt\epsilon_{i}\|^{2}-\wt\mu_{n})^{2}+\left(\|\hat{\epsilon_{i}}\|^{2}-\|\wt\epsilon_{i}\|^{2}\right)^{2}
	+(\wt\mu_{n}-\hat{\mu_{n}})^{2}
	\\&+2(\|\hat{\epsilon_{i}}\|^{2}-\|\wt\epsilon_{i}\|^{2})(\|\wt\epsilon_{i}\|^{2}-\wt\mu_{n})+
	2(\|\hat{\epsilon_{i}}\|^{2}-\|\wt\epsilon_{i}\|^{2})(\wt\mu_{n}-\hat{\mu_{n}})+2(\|\wt\epsilon_{i}\|^{2}-\wt\mu_{n})(\wt\mu_{n}-\hat{\mu_{n}})
	\\:=&\tsn^2+R_{1}+R_{2}+2R_{3}+2R_{4}+2R_{5}
	\end{split}
	\end{equation*}
	where
	\begin{equation*}
	\begin{split}
	R_{1}
	&=\frac{1}{n}\sum_{i=1}^{n}\left(\|\hat{\epsilon_{i}}\|^{2}-\|\wt\epsilon_{i}\|^{2}\right)^{2}
	=\frac{1}{n}\sum_{i=1}^{n}\left(\|\wt\epsilon_{i}-\bar\epsilon_{n}\|^{2}-\|\wt\epsilon_{i}\|^{2}\right)^{2}
	\\&=\frac{1}{n}\sum_{i=1}^{n}\left(\|\bar\epsilon_n\|^{2}+2\langle\wt\epsilon_{i},\,\bar\epsilon_n\rangle\right)^{2}
	\leq\frac{1}{n}\sum_{i=1}^{n}\left[2\|\bar\epsilon_n\|^{4}+2\left(2\langle\wt\epsilon_{i},\,\bar\epsilon_n\rangle\right)^{2}\right]
	\\&\leq 2\|\bar\epsilon_{n}\|^{4}+\left(\frac{8}{n}\sum_{i=1}^{n}\|\wt\epsilon_{i}\|^{2}\right)\times\|\bar\epsilon_n\|^{2}
	\end{split}
	\end{equation*}
	Since $\frac{1}{n}\sum_{i=1}^n\|\wt\epsilon_i\|^2\xrightarrow{P}\E\|\wt\epsilon_i\|^2=\E\|\mu_i-\frac{1}{n}\sum_{i=1}^n\mu_i+\epsilon_i\|^2=\E\|\phi(y_i)-\frac{1}{n}\sum_{i=1}^n\mu_i\|^2\leq 2\E\|\phi(y_i)\|^2+2\E\|\frac{1}{n}\sum_{i=1}^n\mu_i\|^2\leq2\E\|\phi(y_i)\|^2+2\E\|\frac{1}{n}\sum_{i=1}^n(\mu_i+\epsilon_i)\|^2\leq C$, and $\|\bar\epsilon_n\|\xrightarrow{P}0$ (Lemma \ref{lemma:deviation_exp_bounded}), we have $R_1\xrightarrow{P}0$.
	\begin{equation*}
	\begin{split}
	R_{2}&=\left(\frac{1}{n}\sum_{i=1}^{n}\|\hat{\epsilon_{i}}\|^{2}-\frac{1}{n}\sum_{i=1}^{n}\|\wt\epsilon_{i}\|^{2}\right)^{2}=\left(\frac{1}{n}\sum_{i=1}^{n}\|\bar\epsilon_n \|^{2}+\langle\frac{2}{n}\sum_{i=1}^{n}\epsilon_{i},\,-\bar \epsilon_n \rangle\right)^{2}=\|\bar \epsilon_n \|^{4}
	\end{split}
	\end{equation*}
	Since $\|\bar \epsilon_n \|\xrightarrow{P}0$,  we have $R_2\xrightarrow{P} 0$.
	\begin{equation*}
	\begin{split}
	|R_{3}|
	&=|\frac{1}{n}\sum_{i=1}^{n}(\|\hat{\epsilon_{i}}\|^{2}-\|\wt\epsilon_{i}\|^{2})(\|\wt\epsilon_{i}\|^{2}-\wt\mu_{n})|
	\\&\leq\sqrt{\frac{1}{n}\sum_{i=1}^{n}(\|\hat{\epsilon_{i}}\|^{2}-\|\wt\epsilon_{i}\|^{2})^{2}\frac{1}{n}\sum_{i=1}^{n}(\|\wt\epsilon_{i}\|^{2}-\wt\mu_{n})^{2}}
	\\&\leq\sqrt{R_{1}\left(\frac{1}{n}\sum_{i=1}^{n}\|\wt\epsilon_{i}\|^{4}-\wt\mu_{n}^{2}\right)}
	\end{split}
	\end{equation*}
	From Law of Large Numbers, we have 
	\begin{equation}\label{appendix_eq:S2_null_eq1}
	    \frac{1}{n}\sum_{i=1}^{n}\|\wt\epsilon_{i}\|^{4}-\wt\mu_{n}^{2}
	\xrightarrow{P}
	\rho^*Var_{F_0}\left(\left\|\phi(y)-\frac{1}{n}\sum_{i=1}^n\mu_i\right\|^2\right)+(1-\rho^*)Var_{F_1}\left(\left\|\phi(y)-\frac{1}{n}\sum_{i=1}^n\mu_i\right\|^2\right)
	:=c^*
	\end{equation}
	where the last quantity is well-defined because
	\begin{equation*}
	    \begin{split}
	        &Var_{F_0}\left(\left\|\phi(y)-\frac{1}{n}\sum_{i=1}^n\mu_i\right\|^2\right)
	=\E\left(\|\phi(y)\|^2-\E\|\phi(y)\|^2-2\left\langle\frac{1}{n}\sum_{i=1}^n\mu_i,\epsilon\right\rangle\right)^2
	\\&\leq 2\E\left(\|\phi(y)\|^2-\E\|\phi(y)\|^2\right)^2+2\E\left(2\left\langle\frac{1}{n}\sum_{i=1}^n\mu_i,\epsilon\right\rangle\right)^2
	\\&=2Var(\|\phi(y)\|^2)+8\E\|\epsilon\|^2\left\|\frac{1}{n}\sum_{i=1}^n\mu_i\right\|^2
	\\&\leq 2Var(\|\phi(y)\|^2)+8\E\|\phi(y)\|^2\frac{1}{n^2}\left(\sum_{i=1}^n\left\|\mu_i\right\|\right)^2
	<+\infty
	    \end{split}
	\end{equation*}
	Similarly, $Var_{F_1}\left(\|\phi(y)-\frac{1}{n}\sum_{i=1}^n\mu_i\|^2\right)<\infty$. Thus,
	$(c^*)^2$ is a bounded positive constant.
	Combined with $R_1\xrightarrow{P}0$, we have $R_3\xrightarrow{P}0$.
	\begin{equation*}
	\begin{split}
	|R_{4}|
	=\left|\frac{2}{n}\sum_{i=1}^{n}(\|\hat{\epsilon_{i}}\|^{2}-\|\epsilon_{i}\|^{2})(\tilde\mu_{n}-\hat\mu_{n})\right|
	=2|\tilde\mu_{n}-\hat\mu_{n}|\times\left|\frac{1}{n}\sum_{i=1}^{n}\|\hat{\epsilon_{i}}\|^{2}-\E\|\epsilon_{i}\|^{2}\right|
	\xrightarrow{P}0
	\end{split}
	\end{equation*}
	where the last $\xrightarrow{P}$ follows from the fact that $R_2=\left(\tilde\mu_n-\hat\mu_n\right)^2\rightarrow 0$, and  $\frac{1}{n}\sum_{i=1}^{n}\|\hat{\epsilon_{i}}\|^{2}-\E\|\epsilon_{i}\|^{2}\xrightarrow{P} 0$ (law of large numbers).
	$$
	|R_{5}|=\left|\frac{1}{n}\sum_{i=1}^{n}\left(\|\epsilon_{i}\|^{2}-\tilde\mu_{n}\right)\left(\tilde\mu_{n}-\hat\mu_{n}\right)\right|=0
	$$
	where the last equality follows from the definition of $\tilde\mu_{n}$.
	
	Combining the above, we know that $\hsn-\tsn\xrightarrow{P}0$ and thus, $\frac{\tsn}{\hsn}\xrightarrow{P}1$. Equation \eqref{appendix_eq:S2_null_eq1} says $\tsn\xrightarrow{P}c^*$. Thus, we have $\hsn\xrightarrow{P}c^*$. This completes the first part of the desired conclusion.
	
	For the second conclusion, notice under the null, $c^*=s_n$
	and thus, $\frac{s_n}{\hsn}\xrightarrow{P}1$.
\end{proof}

\subsection{Proof of Theorem \ref{thm:consistency_S1}}
\begin{proof}
First we have the following equation: $\forall t$,
\begin{equation*}
\begin{split}
T_1(t)
=& \frac{1}{t}\sum_{i=1}^{t}k(y_i,y_i)+\frac{1}{n-t}\sum_{i=t+1}^{n}k(y_i,y_i)-2\langle \overline{\phi(y)}_{t},\overline{\phi(y)}_{n-t}\rangle
\\&-\frac{1}{t-1}\left(\left\|(\bm{\mu}^*)_t-(\bm{\mu})_t\right\|^2\,+\,2\left\langle(\bm{\mu}^*)_t-(\bm{\mu})_t,\,(\bm{\epsilon})_1^t\right\rangle\,-\,\|\Pi(\bm{\epsilon})_1^t\|^2\,+\,\|(\bm{\epsilon})_1^t\|^2\right)
\\&-\frac{1}{n-t-1}\left(\|(\bm{\mu}^*)_{n-t}-(\bm{\mu})_{n-t}\|^2\,+\,2\left\langle(\bm{\mu}^*)_{n-t}-(\bm{\mu})_{n-t},\,(\bm{\epsilon})_{t+1}^n\right\rangle\,-\,\|\Pi(\bm{\epsilon})_{t+1}^n\|^2\,+\,\|(\bm{\epsilon})_{t+1}^n\|^2\right)
\end{split}
\end{equation*}
Plugging this into the basic inequality $S_1(\hat t)\geq S_1(t^*)$. 

When $t^*<\hat t$, we have
\begin{align}
0\,\leq &
\,\,\frac{t^*-\hat t}{n^2}\sum_{i=1}^{t^*}k(y_i,y_i)+\frac{n-\hat t-t^*}{n^2}\sum_{i=t^*+1}^{\hat t}k(y_i,y_i)+\frac{\hat t-t^*}{n^2}\sum_{i=\hat t+1}^{n}k(y_i,y_i)
\label{appendix_eq:consistency_s1_line1}
\\&
+ \frac{t^*-\hat t}{n^2}\sum_{i=1}^{t^*}\|\epsilon_i\|^2+\frac{n-\hat t-t^*}{n^2}\sum_{i=t^*+1}^{\hat t}\|\epsilon_i\|^2+\frac{\hat t-t^*}{n^2}\sum_{i=\hat t+1}^{n}\|\epsilon_i\|^2
\label{appendix_eq:consistency_s1_line2}
\\&
-\frac{\hat t(n-\hat t)}{n^2(\hat t-1)}\left\|\firstmu-\firstempmu\right\|^2-\frac{\hat t(n-\hat t)}{n^2(n-\hat t-1)}\left\|\secondmu-\secondempmu\right\|^2
\label{appendix_eq:consistency_s1_line3}
\\&
-\frac{2\hat t(n-\hat t)}{n^2(\hat t-1)}\left\langle\firstmu-\firstempmu,\firstep\right\rangle-\frac{2\hat t(n-\hat t)}{n^2(n-\hat t-1)}\left\langle\secondmu-\secondempmu,\secondep\right\rangle
\label{appendix_eq:consistency_s1_line4}
\\&
+\frac{\hat t(n-\hat t)}{n^2(\hat t-1)}\left\|\Pi\firstempep\right\|^2+\frac{\hat t(n-\hat t)}{n^2(n-\hat t-1)}\left\|\Pi\secondempep\right\|^2
\label{appendix_eq:consistency_s1_line5}
\\&
-\frac{t^*(n-t^*)}{n^2(t^*-1)}\left\|\Pi\firstep\right\|^2-\frac{t^*(n-t^*)}{n^2(n-t^*-1)}\left\|\secondep\right\|^2
\label{appendix_eq:consistency_s1_line6}
\\&
-\frac{2\hat t(n-\hat t)}{n^2}\left\langle\overline{\phi(y)}_{\hat t},\,\overline{\phi(y)}_{n-\hat t}\right\rangle
+\frac{2t^*(n-t^*)}{n^2}\left\langle\overline{\phi(y)}_{t^*},\,\overline{\phi(y)}_{n-t^*}\right\rangle
\label{appendix_eq:consistency_s1_line7}
\end{align}
Now we will bound each line seperately. 

For line \eqref{appendix_eq:consistency_s1_line1} and line \eqref{appendix_eq:consistency_s1_line2}, we have for any $x>0$, 
\begin{equation}\label{appendix_eq:consistency_s1_term1}
\begin{split}
&\eqref{appendix_eq:consistency_s1_line1}+\eqref{appendix_eq:consistency_s1_line2}
\\=&
\frac{t^*-\hat t}{n^2}\sum_{i=1}^{t^*}\left(\|\mu_0\|^2+2\langle\mu_0,\epsilon_i\rangle\right)+\frac{n-\hat t-t^*}{n^2}\sum_{i=t^*+1}^{\hat t}\left(\|\mu_1\|^2+2\langle\mu_1,\epsilon_i\rangle\right)+\frac{\hat t-t^*}{n^2}\left(\|\mu_1\|^2+2\langle\mu_1,\epsilon_i\rangle\right)
\\=&
\frac{t^*-\hat t}{n^2}t^*\|\mu_0\|^2+\frac{(\hat t-t^*)(2n-2\hat t-t^*)}{n^2}\|\mu_1\|^2
\\&+2\frac{t^*-\hat t}{n^2}\sum_{i=1}^{t^*}\langle\mu_0,\epsilon_i\rangle+2\frac{n-\hat t-t^*}{n^2}\sum_{i=t^*+1}^{\hat t}\langle\mu_1,\epsilon_i\rangle+2\frac{\hat t-t^*}{n^2}\sum_{i=\hat t+1}^{n}\langle\mu_1,\epsilon_i\rangle
\\\stackrel{(a)}{\leq}&
U_1
\\&+2\|\mu_0\|\times \frac{|t^*-\hat t|}{n^2}\times\|\sum_{i=1}^{t^*}\epsilon_i\|+2\|\mu_1\|\times \frac{|n-t^*-\hat t|}{n^2}\times\|\sum_{i=t^*+1}^{\hat t}\epsilon_i\|+2\|\mu_1\|\times \frac{|\hat t-t^*|}{n^2}\times\|\sum_{i=\hat t+1}^{n}\epsilon_i\|
\\\stackrel{(b)}{\leq}&
U_1+\frac{2|t^*-\hat t|}{n^2}\sqrt{\frac{14t^*}{3}}(x+\sqrt{2})M^2+ \frac{2|n-t^*-\hat t|}{n^2}\sqrt{\frac{14(\hat t-t^*)}{3}}(x+\sqrt{2})M^2
\\&+ \frac{2|\hat t-t^*|}{n^2}\sqrt{\frac{14(n-\hat t)}{3}}(x+\sqrt{2})M^2
\quad\quad\quad
(w.p. \geq 1-3e^{-x})
\\\stackrel{(c)}{\leq}&
U_1+6M^2\frac{1}{\sqrt{n}}\sqrt{\frac{\hat t-t^*}{n}}
\end{split}
\end{equation}
where (a) follows from Cauchy Schwarz Inequality and we define $U_1=\frac{t^*-\hat t}{n^2}t^*\|\mu_0\|^2+\frac{(\hat t-t^*)(2n-2\hat t-t^*)}{n^2}\|\mu_1\|^2$. Inequality (b) follows from Lemma \ref{appendix_lemma:bound_noise_kernel_bd} and notice that $d(y_i,y_j)\leq M^2$ implies $k(y_i,y_i)\leq M^2$ for any $y_i$ and thus, $\|\mu_0\|^2\leq M^2$ and $\|\mu_1\|^2\leq M^2$. Inequality (c) follows from the fact that $\frac{\hat t-t^*}{n}\leq\sqrt{\frac{\hat t-t^*}{n}}$.

For line \eqref{appendix_eq:consistency_s1_line3} and line \eqref{appendix_eq:consistency_s1_line4}, from Proposition 3 of \cite{arlot2012kernel}, we have  for any $\theta>0$ and $x>0$, with probability at least $1-4e^{-x}$,
\begin{equation*}
\begin{split}
\eqref{appendix_eq:consistency_s1_line4}
\,\,\,{\leq}\,\,\,
&\frac{2\hat t(n-\hat t)}{n^2(\hat t-1)}\left[\theta\|\firstmu-\firstempmu\|^2+\left(\frac{3v_0}{2}+\frac{4M^2}{3}\right)x\right]
\\&+
\frac{2\hat t(n-\hat t)}{n^2(\hat t-1)}\left[\theta\|\secondmu-\secondempmu\|^2+\left(\frac{3v_1}{2}+\frac{4M^2}{3}\right)x\right]
\end{split}
\end{equation*}
Take $\theta=\frac{1}{3}$, we have with probability at least $1-4e^{-x}$,
\begin{equation}\label{appendix_eq:consistency_s1_term2}
\begin{split}
\eqref{appendix_eq:consistency_s1_line3}+\eqref{appendix_eq:consistency_s1_line4}
\stackrel{(a)}{\leq}
-\frac{1}{3}\frac{t(n-t)}{n^2(t-1)}(t-t^*)\|\mu_0-\mu_1\|^2+\frac{17}{2n}M^2x
\end{split}
\end{equation}
where (a) follows from the fact that $v_0\leq M^2$.

For line \eqref{appendix_eq:consistency_s1_line5} and line \eqref{appendix_eq:consistency_s1_line6}, from Lemma \ref{appendix_lemma:bound_noise_kernel_bd}, we know that for any $x>0$, with probability at least $1-2e^{-x}$, we have
\begin{equation}\label{appendix_eq:consistency_s1_term3}
\eqref{appendix_eq:consistency_s1_line5} + \eqref{appendix_eq:consistency_s1_line6}
\leq\left(\frac{\hat t(n-\hat t)}{n^2(\hat t-1)}+\frac{\hat t(n-\hat t)}{n^2(n-\hat t-1)}\right)\left(M^2+\frac{14}{3}M^2(x+2\sqrt{2x})\right)
\leq \frac{1}{n}\frac{28}{3}M^2(\sqrt{x}+\sqrt{2})^2
\end{equation}

For line \eqref{appendix_eq:consistency_s1_line7}, we have
\begin{equation*}
\begin{split}
&-\frac{2\hat t(n-\hat t)}{n^2}\langle\overline{\phi(y)}_{\hat t},\,\overline{\phi(y)}_{n-\hat t}\rangle
=
-\frac{2}{n^2}\langle t^*\mu_0+(\hat t-t^*)\mu_1+\sum_{i=1}^{\hat t}\epsilon_i,\,(n-\hat t)\mu_1+\sum_{i=\hat t+1}^{n}\epsilon_i\rangle
\\=
&
-\frac{2}{n^2}\left[t^*(n-\hat t)\langle\mu_0,\mu_1\rangle+
(\hat t-t^*)(n-\hat t)\|\mu_1\|^2+
\langle t^*\mu_0+(\hat t-t^*)\mu_1,\,\sum_{i=\hat t+1}^{n}\epsilon_i\rangle \right]
\\&
-\frac{2}{n^2}\left[\langle \sum_{i=1}^{\hat t}\epsilon_i,\,(n-\hat t)\mu_1\rangle+
\langle \sum_{i=1}^{\hat t}\epsilon_i,\,\sum_{i=\hat t+1}^{n}\epsilon_i\rangle
\right]
\end{split}
\end{equation*}
and 
\begin{equation*}
\begin{split}
&\frac{2t^*(n-t^*)}{n^2}\langle\overline{\phi(y)}_{t^*},\,\overline{\phi(y)}_{n-t^*}\rangle
\\=
&
\frac{2}{n^2}\left[t^*(n-t^*)\langle\mu_0,\mu_1\rangle+
\langle t^*\mu_0,\,\sum_{i=t^*+1}^{n}\epsilon_i\rangle 
+\langle \sum_{i=1}^{t^*}\epsilon_i,\,(n-t^*)\mu_1\rangle+
\langle \sum_{i=1}^{t^*}\epsilon_i,\,\sum_{i=t^*+1}^{n}\epsilon_i\rangle
\right]
\end{split}
\end{equation*}
Thus, for any $x>0$,
\begin{equation}\label{appendix_eq:consistency_s1_term4}
\begin{split}
\eqref{appendix_eq:consistency_s1_line7}
=
&
\,\,\,\frac{2}{n^2}\left[t^*(\hat t-t^*)\langle\mu_0,\mu_1\rangle-
(\hat t-t^*)(n-\hat t)\|\mu_1\|^2+
\langle \sum_{i=1}^{t^*}\epsilon_i,\,(\hat t-t^*)\mu_1\rangle
+
\langle \sum_{i=t^*+1}^{\hat t}\epsilon_i,\,t^*\mu_0-(n-\hat t)\mu_1\rangle
 \right]
\\&
+\frac{2}{n^2}\left[\langle \sum_{i=\hat t+1}^{n}\epsilon_i,\,-(\hat t-t^*)\mu_1\rangle+
\langle \sum_{i=1}^{t^*}\epsilon_i,\,\sum_{i=t^*+1}^{n}\epsilon_i\rangle
-\langle \sum_{i=1}^{\hat t}\epsilon_i,\,\sum_{i=\hat t+1}^{n}\epsilon_i\rangle
\right]
\\\stackrel{(a)}{\leq}&
\,\,\,\frac{2}{n^2}\left[t^*(\hat t-t^*)\langle\mu_0,\mu_1\rangle-
(\hat t-t^*)(n-\hat t)\|\mu_1\|^2+
(\hat t-t^*)\left(\sqrt{2xt^*}+\frac{x}{3}\right)M^2
\right]
\\&
+\frac{2}{n^2}\left[
(t^*+n-\hat t)\left(\sqrt{2x(\hat t-t^*)}+\frac{x}{3}\right)M^2
+(\hat t-t^*)\left(\sqrt{2x(n-\hat t)}+\frac{x}{3}\right)M^2
\right]
\\&
+\frac{2}{n^2}\left[
\frac{14}{3}\sqrt{t^*(n-t^*)}M^2(\sqrt{x}+\sqrt{2})^2
+\frac{14}{3}\sqrt{t(n-\hat t)}M^2(\sqrt{x}+\sqrt{2})^2
\right]
\quad\quad w.p. \geq 1-5e^{-x}
\\\stackrel{(b)}{\leq}&\,\,\,
\frac{2}{n^2}t^*(\hat t-t^*)\langle\mu_0,\mu_1\rangle-
\frac{2}{n^2}(\hat t-t^*)(n-\hat t)\|\mu_1\|^2
\\&+6\sqrt{\frac{14}{3}}M^2(\sqrt{x}+\sqrt{2})\frac{1}{\sqrt{n}}\sqrt{\frac{\hat t-t^*}{n}}+\frac{56}{3}\frac{M^2(\sqrt{x}+\sqrt{2})^2}{n}
\end{split}
\end{equation}
where (a) follows from Proposition 6 in \cite{arlot2012kernel}. Inequality (b) follows from the fact that $\frac{\hat t-t^*}{n}\leq \sqrt{\frac{\hat t-t^*}{n}}$.

Thus, combining \eqref{appendix_eq:consistency_s1_term1}, \eqref{appendix_eq:consistency_s1_term2}, \eqref{appendix_eq:consistency_s1_term3} and \eqref{appendix_eq:consistency_s1_term4}, we have: with probability at least $1-14e^{-x}$,
\begin{equation*}
\begin{split}
0\leq &
-\frac{(\hat t-t^*)t^*}{n^2}\|\mu_0-\mu_1\|^2-\frac{1}{2}\frac{(\hat t-t^*)(n-\hat t)}{n^2}\|\mu_0-\mu_1\|^2+
\frac{1}{n}\left(\frac{17}{2}x+28(\sqrt{x}+\sqrt{2})^2\right)M^2
\\&+12\sqrt{\frac{14}{3}}M^2(\sqrt{x}+\sqrt{2})\frac{1}{\sqrt{n}}\sqrt{\frac{\hat t-t^*}{n}}
\\\leq &
-\frac{1}{2}\frac{\hat t-t^*}{n}\frac{n-\hat t+t^*}{n}\|\mu_0-\mu_1\|^2+
\frac{1}{n}\left(\frac{17}{2}x+28(\sqrt{x}+\sqrt{2})^2\right)M^2
\\&+12\sqrt{\frac{14}{3}}M^2(\sqrt{x}+\sqrt{2})\frac{1}{\sqrt{n}}\sqrt{\frac{\hat t-t^*}{n}}
\\\leq &
-\frac{1}{2}\frac{\hat t-t^*}{n}(1-\rho_1+\rho_0)\|\mu_0-\mu_1\|^2+
\frac{1}{n}\left(\frac{17}{2}x+28(\sqrt{x}+\sqrt{2})^2\right)M^2
\\&+12\sqrt{\frac{14}{3}}M^2(\sqrt{x}+\sqrt{2})\frac{1}{\sqrt{n}}\sqrt{\frac{\hat t-t^*}{n}}
\end{split}
\end{equation*}
The solution to the above inequality is that
\begin{equation*}
\begin{split}
\frac{\hat t-t^*}{n}
\leq
\frac{1}{n}\left(\frac{12\sqrt{14/3}M^2(\sqrt{x}+\sqrt{2})}{(1-\rho_1+\rho_0)\|\mu_0-\mu_1\|^2}\right)^2
+
\frac{2}{n}\frac{17x/2+28(\sqrt{x}+\sqrt{2})^2}{(1-\rho_1+\rho_0)\|\mu_0-\mu_1\|^2}M^2
\end{split}
\end{equation*}
When $\hat t<t^*$, we can prove a similar inequality. Thus for any $\hat t$, the desired conclusion holds.
\end{proof}

\subsection{Proof of Theorem \ref{thm:consistency_S2}}
\begin{proof}
First notice that $\forall t$, under assumption (3) $\mu_0=\mu_1$, we have
\begin{equation}
\bar d_{B_1(t)}-\bar d_{B_2(t)}
=\frac{2}{t-1}\left(-\|\Pi(\bm\epsilon)_{1}^t\|^2+\|(\bm\epsilon)_{1}^t\|^2\right)
-\frac{2}{n-t-1}\left(-\|\Pi(\bm\epsilon)_{t+1}^n\|^2+\|(\bm\epsilon)_{t+1}^n\|^2\right)
\end{equation}
Plug this into the basic inequality, $S_2(\hat t)\geq S_1(t^*)$, we have
\begin{equation*}\footnotesize
\begin{split}
&LHS:=
\\&\sqrt{\frac{(n-t^*)t^*}{n^2}}
\times\left|
-\frac{1}{t^*(t^*-1)}\left\|\sum_{i=1}^{t^*}\epsilon_i\right\|^2+\frac{1}{t^*-1}\sum_{i=1}^{t^*}\|\epsilon_i\|^2
+\frac{1}{(n-t^*)(n-t^*-1)}\left\|\sum_{i=t^*+1}^{n}\epsilon_i\right\|^2-\frac{1}{n-t^*-1}\sum_{i=t^*+1}^{n}\|\epsilon_i\|^2
\right|
\\&\leq 
\sqrt{\frac{(n-\hat t)\hat t}{n^2}}
\times\left|
-\frac{1}{\hat t(\hat t-1)}\left\|\sum_{i=1}^{\hat t}\epsilon_i\right\|^2+\frac{1}{\hat t-1}\sum_{i=1}^{\hat t}\|\epsilon_i\|^2
+\frac{1}{(n-\hat t)(n-\hat t-1)}\left\|\sum_{i=\hat t+1}^{n}\epsilon_i\right\|^2-\frac{1}{n-\hat t-1}\sum_{i=\hat t+1}^{n}\|\epsilon_i\|^2
\right|
:= RHS
\end{split}
\end{equation*}
Now suppose that $v_0>v_1$. When $\hat t\leq t^*$, we have, for any $x>0$, when $n$ is sufficiently large,
\begin{equation}\label{appendix_eq:S2_consistency_RHS}
\begin{split}
RHS
=&
\sqrt{\frac{(n-t)t}{n^2}}
\times
\left\vert
-\frac{1}{\hat t(\hat t-1)}\left\|\sum_{i=1}^{\hat t}\epsilon_i\right\|^2+\frac{1}{(n-\hat t)(n-\hat t-1)}\left\|\sum_{i=\hat t+1}^{n}\epsilon_i\right\|^2\right.
\\&\left.+\frac{1}{\hat t-1}\sum_{i=1}^{\hat t}(\|\epsilon_i\|^2-v_i)-\frac{1}{n-\hat t-1}\sum_{i=\hat t+1}^{n}(\|\epsilon_i\|^2-v_i)
+\frac{\hat t}{\hat t-1}(v_0-v_1)+\frac{v_1}{n-\hat t-1}
\right\vert
\\\stackrel{(a)}{=}&
\sqrt{\frac{(n-\hat t)\hat t}{n^2}}
\times\left(
-\frac{1}{\hat t(\hat t-1)}\left\|\sum_{i=1}^{\hat t}\epsilon_i\right\|^2+\frac{1}{(n-\hat t)(n-\hat t-1)}\left\|\sum_{i=\hat t+1}^{n}\epsilon_i\right\|^2
+\frac{1}{\hat t-1}\sum_{i=1}^{\hat t}(\|\epsilon_i\|^2-v_i)\right) 
\\&+\sqrt{\frac{(n-\hat t)\hat t}{n^2}}
\times\left(-\frac{1}{n-\hat t-1}\sum_{i=\hat t+1}^{n}(\|\epsilon_i\|^2-v_i)
+\frac{t^*}{\hat t-1}(v_0-v_1)+\frac{v_1}{n-\hat t-1}\right)
\quad w.p. \geq 1-e^{-x}
\\\stackrel{(b)}{\leq}&
\sqrt{\frac{(n-\hat t)\hat t}{n^2}}
\times\left(
\frac{1}{(\hat t-1)}\frac{14}{3}(\sqrt{x}+\sqrt{2})^2M^2+\frac{1}{(n-\hat t-1)}\frac{14}{3}(\sqrt{x}+\sqrt{2})^2M^2
+\frac{1}{\hat t-1}\sum_{i=1}^{t}(\|\epsilon_i\|^2-v_i)\right) 
\\&+\sqrt{\frac{(n-\hat t)\hat t}{n^2}}
\times\left(-\frac{1}{n-\hat t-1}\sum_{i=\hat t+1}^{n}(\|\epsilon_i\|^2-v_i)
+\frac{t^*}{\hat t-1}(v_0-v_1)+\frac{v_1}{n-\hat t-1}\right)
\quad w.p. \geq 1-5e^{-x}
\end{split}
\end{equation}
where (a) follows from the fact that $-\frac{1}{t^*(t^*-1)}\|\sum_{i=1}^{t^*}\epsilon_i\|^2+\frac{1}{(n-t^*)(n-t^*-1)}\|\sum_{i=t^*+1}^{n}\epsilon_i\|^2
+\frac{1}{t^*-1}\sum_{i=1}^{t^*}(\|\epsilon_i\|^2-v_i)-\frac{1}{n-t^*-1}\sum_{i=t^*+1}^{n}(\|\epsilon_i\|^2-v_i)=\O_p(\frac{1}{\sqrt{n}})$ and assumption (4). Inequality (b) follows from Lemma \ref{appendix_lemma:bound_noise_kernel_bd} and union bound.

Similarly, we have
\begin{equation}\small\label{appendix_eq:S2_consistency_LHS}
\begin{split}
LHS
=&
\sqrt{\frac{(n-t^*)t^*}{n^2}}
\times
\left|
-\frac{1}{t^*(t^*-1)}\left\|\sum_{i=1}^{t^*}\epsilon_i\right\|^2+\frac{1}{(n-t^*)(n-t^*-1)}\left\|\sum_{i=t^*+1}^{n}\epsilon_i\right\|^2\right.
\\&\left.+\frac{1}{t^*-1}\sum_{i=1}^{t^*}(\|\epsilon_i\|^2-v_i)-\frac{1}{n-t^*-1}\sum_{i=t^*+1}^{n}(\|\epsilon_i\|^2-v_i)
+\frac{t^*}{t^*-1}v_0-\frac{(n-t^*)}{n-t^*-1}v_1
\right|
\\\geq&
\sqrt{\frac{(n-t^*)t^*}{n^2}}
\times\left(
-\frac{1}{(t^*-1)}\frac{14}{3}(\sqrt{x}+\sqrt{2})^2M^2-\frac{1}{(n-t^*-1)}\frac{14}{3}(\sqrt{x}+\sqrt{2})^2M^2
-\frac{1}{t^*-1}\sum_{i=1}^{t^*}(\|\epsilon_i\|^2-v_i)\right) 
\\&+\sqrt{\frac{(n-t^*)t^*}{n^2}}
\times\left(-\frac{1}{n-t^*-1}\sum_{i=t^*+1}^{n}(\|\epsilon_i\|^2-v_i)
+(v_0-v_1)-\frac{v_1}{n-t^*-1}-\frac{v_0}{t^*-1}\right)
\quad w.p. \geq 1-5e^{-x}
\end{split}
\end{equation}
Combing Equation \eqref{appendix_eq:S2_consistency_RHS} and Equation \eqref{appendix_eq:S2_consistency_LHS}, we have: 
\begin{equation*}
\begin{split}
&\left(\sqrt{\frac{(n-t^*)t^*}{n^2}}-\sqrt{\frac{(n-\hat t)\hat t}{n^2}}\frac{t^*}{\hat t}\right)|v_0-v_1|
\\\leq&\,\,\,
\sqrt{\frac{(n-\hat t)\hat t}{n^2}}\left[\frac{1}{n-\hat t-1}v_1+\frac{14}{3}\left(\frac{1}{n_0}+\frac{1}{n-n_1}\right)(\sqrt{x}+\sqrt{2})^2M^2\right]
\\&+
\sqrt{\frac{(n-t^*)t^*}{n^2}}\left[\frac{1}{n-t^*-1}v_1+\frac{1}{t^*-1}v_0+\frac{14}{3}\left(\frac{1}{n_0}+\frac{1}{n-n_1}\right)(\sqrt{x}+\sqrt{2})^2M^2\right]
\\&
+\sqrt{\frac{(n-\hat t)\hat t}{n^2}}\left[\frac{t^*}{\hat t-1}-\frac{t^*}{\hat t}\right]|v_0-v_1|
\\&
+\left(\sqrt{\frac{(n-\hat t)\hat t}{n^2}}\frac{1}{\hat t-1}-\sqrt{\frac{(n-t^*)t^*}{n^2}}\frac{1}{t^*-1}\right)\sum_{i=1}^{t^*}\left(\|\epsilon_i\|^2-v_i\right)
\\&
+\left(\sqrt{\frac{(n-\hat t)\hat t}{n^2}}\frac{1}{\hat t-1}+\sqrt{\frac{(n-t^*)t^*}{n^2}}\frac{1}{n-t^*-1}\right)\sum_{i=t^*+1}^{\hat t}\left(\|\epsilon_i\|^2-v_i\right)
\\&
-\left(\sqrt{\frac{(n-\hat t)\hat t}{n^2}}\frac{1}{n-\hat t-1}-\sqrt{\frac{(n-t^*)t^*}{n^2}}\frac{1}{n-t^*-1}\right)\sum_{i=\hat t+1}^{n}\left(\|\epsilon_i\|^2-v_i\right)
\quad w.p. \geq 1-10e^{-x}
\\\stackrel{(a)}{\leq}&\,\,\,
\frac{1}{2}\left[\frac{2}{n-n_1}M^2+\frac{2}{n_0}M^2+\frac{28}{3}\left(\frac{1}{n_0}+\frac{1}{n-n_1}\right)(\sqrt{x}+\sqrt{2})^2M^2\right]
\\&
+\frac{2}{n}\left(\sqrt{\frac{n-\hat t}{\hat t}}-\sqrt{\frac{n-t^*}{t^*}}\right)\left(\sqrt{2t^*\gamma^2x}+\frac{M^2}{3}x\right)
\\&
+\frac{2}{n}\left(\sqrt{\frac{n-\hat t}{\hat t}}+\sqrt{\frac{t^*}{n-t^*}}\right)\left(\sqrt{2(\hat t-t^*)\gamma^2x}+\frac{M^2}{3}x\right)
\\&
+\frac{2}{n}\left(\sqrt{\frac{\hat t}{n-\hat t}}-\sqrt{\frac{t^*}{n-t^*}}\right)\left(\sqrt{2(n-\hat t)\gamma^2x}+\frac{M^2}{3}x\right)
\quad\quad w.p. \geq 1-16e^{-x}
\\\leq&\,\,\,
\frac{1}{2}\left[\frac{2}{n-n_1}M^2+\frac{2}{n_0}M^2+\frac{28}{3}\left(\frac{1}{n_0}+\frac{1}{n-n_1}\right)(\sqrt{x}+\sqrt{2})^2M^2\right]
\\&
+\frac{2}{n}\left(2\sqrt{\frac{n-\hat t}{\hat t}}-\sqrt{\frac{n-t^*}{t^*}}+\sqrt{\frac{\hat t}{n-\hat t}}\right)\frac{M^2}{3}x
\\&
+\frac{2}{n}\left(\sqrt{\frac{n-t^*}{n-\hat t}}-\sqrt{\frac{t^*}{\hat t}}\right)\sqrt{n-\hat t}\left(\sqrt{\frac{\hat t}{n-t^*}}-1\right)\sqrt{2\gamma^2x}
\\&
+\frac{2}{n}\left(\sqrt{\frac{n-\hat t}{\hat t}}+\sqrt{\frac{t^*}{n-t^*}}\right)\sqrt{\frac{\hat t-t^*}{n}}\sqrt{2\gamma^2x}
\quad\quad w.p. \geq 1-16e^{-x}
\end{split}
\end{equation*}
where (a) follows from Proposition in \cite{arlot2012kernel} and the fact that $0<v_0,v_1\leq M^2$.

Let $c_1=\frac{1}{2}\left[\frac{2}{n-n_1}M^2+\frac{2}{n_0}M^2+\frac{28}{3}\left(\frac{1}{n_0}+\frac{1}{n-n_1}\right)(\sqrt{x}+\sqrt{2})^2M^2\right] + \frac{2}{n}\left(2\sqrt{\frac{n-\hat t}{\hat t}}-\sqrt{\frac{n-t^*}{t^*}}+\sqrt{\frac{\hat t}{n-\hat t}}\right)\frac{M^2}{3}x$, 
$c_2=\sqrt{\frac{t^*(n-\hat t)}{n^2}}(v_0-v_1)-\frac{2}{n}\sqrt{n-\hat t}\left(\sqrt{\frac{\hat t}{n-t^*}}-1\right)\sqrt{2\gamma^2 x}$, $w=\frac{t-t^*}{n}$, $c_3=\frac{2}{\sqrt{n}}\left(\sqrt{\frac{n-\hat t}{\hat t}}+\sqrt{\frac{t^*}{n-t^*}}\right)\sqrt{2\gamma^2x}$.
Then, we have: with probability at least $1-16e^{-x}$,
\begin{equation*}
c_2\left(\sqrt{1+\frac{w}{1-w-\rho^*}}-\sqrt{1-\frac{w}{w+\rho^*}}\right)
\leq c_1+c_3\sqrt{w}
\end{equation*}
Using the fact that 
\begin{equation*}
\begin{split}
c_2\left(\sqrt{1+\frac{w}{1-w-\rho^*}}-\sqrt{1-\frac{w}{w+\rho^*}}\right)
\geq 
c_2\left(\sqrt{1+w}-1\right)
\stackrel{(a)}{\geq}
c_2(\sqrt{2}-1)w
\end{split}
\end{equation*}
where (a) follows from the fact that $\sqrt{1+w}\geq 1+(\sqrt{2}-1)w$ for any $w>0$. We have
\begin{equation}\label{appendix_eq:s2_consistency_inequality}
c_2(\sqrt{2}-1)w\leq c_1+c_3\sqrt{w}
\end{equation}
Notice that
\begin{equation*}
\begin{split}
c_1\leq
&\frac{1}{n}\left[\frac{1}{1-\rho_1}M^2+\frac{1}{\rho_0}M^2+\frac{14}{3}\left(\frac{1}{\rho_0}+\frac{1}{1-\rho_1}\right)(\sqrt{x}+\sqrt{2})^2M^2\right] 
+ \frac{2}{n}\left(2\sqrt{\frac{1-\rho_0}{\rho_0}}+\sqrt{\frac{\rho_1}{1-\rho_1}}\right)\frac{M^2}{3}x
\end{split}
\end{equation*}
\begin{equation*}
\begin{split}
c_2\geq
\sqrt{\rho^*(1-\rho_1)}(v_0-v_1)-\frac{2}{\sqrt{n}}\sqrt{1-\rho_0}\left(\sqrt{\frac{\rho_1}{1-\rho^*}}-1\right)\sqrt{2\gamma^2 x}
\end{split}
\end{equation*}
\begin{equation*}
\begin{split}
c_3
\leq
\frac{2}{\sqrt{n}}\left(\sqrt{\frac{1-\rho_0}{\rho_0}}+\sqrt{\frac{\rho^*}{1-\rho^*}}\right)\sqrt{2\gamma^2x}
\end{split}
\end{equation*}
Thus, $c_2>c_3\sqrt{\rho^*(1-\rho_1)}(v_0-v_1)$ when $n$ is sufficiently large.

So Equation \eqref{appendix_eq:s2_consistency_inequality} yields (when $n$ is sufficiently large)
\begin{equation*}
\begin{split}
w
\leq
&\left(\frac{c_3}{2c_2(\sqrt{2}-1)}+\sqrt{\frac{c_1}{c_2(\sqrt{2}-1)}+\frac{c_3^2}{4c_2^2(\sqrt{2}-1)^2}}\right)^2
\stackrel{(b)}{\leq}
\frac{c_3^2}{c_2^2(\sqrt{2}-1)^2}+\frac{2c_1}{c_2(\sqrt{2}-1)}
\\\leq&
\widetilde C_0 \frac{1}{n}\left(\frac{\gamma^2x}{(v_0-v_1)^2}+\frac{M^2x}{|v_0-v_1|}\right)
\end{split}
\end{equation*}
where (b) utilizes the fact that $(x+y)^2\leq 2(x^2+y^2)$. 
\end{proof}

\subsection{Proof of Theorem \ref{thm:S1_power}}\label{appendix_sec:proof_S1_power}
Theorem \ref{thm:S1_power} is a  direct consequence of the following Theorem.
\begin{theorem}[Alternative distribution for $S_1$]\label{thm:S1_alternative}
	Under the alternative, if (1) $d$ is a semi-metric of negative type,  (2) $d(y_i,y_j)\leq M^2<\infty,\,\forall i,j$, (3) suppose there exists $\bm \Delta^{(1)}\in\mathcal{H}$ s.t. $\|\sqrt{n}(\mu_0-\mu_1)-\bm\Delta^{(1)}\|\rightarrow0$, then
	$$
	S_1
	\xrightarrow{w}
	\max_{\rho\in[\rho_0,\rho_1]}\left(\frac{\sum_l\left(\sqrt{\lambda_l}W^0(\rho)+\xi(\rho)\Delta_l^{(1)}\right)^2-\delta(\rho)}{\rho(1-\rho)}\right)
	$$
	where
	\begin{equation*}
	\delta(\rho)=
	\begin{cases}
	(1-\rho)\left((1-\rho)\rho^* v_0+(\rho-\rho^*+\rho\rho^*)v_1\right),\quad \text{if}\quad \rho^*\leq\rho
	\\
	\rho\left(\rho(1-\rho^*)v_1+(\rho\rho^*-2\rho+1)v_0\right),\quad \text{if}\quad \rho^*>\rho
	\end{cases}
	\end{equation*}
	and 
	$$\xi(\rho)=\begin{cases}
	\rho(1-\rho^*),\quad \text{if}\quad \rho\leq\rho^*\\
	(1-\rho)\rho^*,\quad\text{if}\quad \rho>\rho^*
	\end{cases}
	$$
\end{theorem}
\begin{proof}
First we have
\begin{equation*}
\begin{split}
&\frac{1}{n}\|\frac{n-t}{n}\sum_{i=1}^{t}\phi(y_i)-\frac{t}{n}\sum_{i=t+1}^{n}\phi(y_i)\|^2
=\frac{(n-t)^2t^2}{n^3}\left[\frac{1}{t(n-t)} d_A-\frac{1}{2t^2}d_{B_1}-\frac{1}{2(n-t)^2}d_{B_2}\right]
\\=&
n\rho^2(1-\rho)^2\left[\bar d_A(\lceil n\rho\rceil)-\frac{1}{2}\bar d_{B_1}(\lceil n\rho\rceil)-\frac{1}{2}\bar d_{B_2}(\lceil n\rho\rceil)\right]
+\frac{(n-t)^2t}{n^3(t-1)}\sum_{i=1}^{t}\|\hat \epsilon_i\|^2+\frac{t^2(n-t)}{n^3(n-t-1)}\sum_{i=t+1}^{n}\|\hat \epsilon_i\|^2
\end{split}
\end{equation*}
where $\hat\epsilon_i=\phi(y_i)-\frac{1}{t}\sum_{i=1}^{t}\phi(y_i)$ for $i=1,2,\cdots,t$ and $\hat\epsilon_i=\phi(y_i)-\frac{1}{n-t}\sum_{i=t+1}^{n}\phi(y_i)$ for $i=t+1,t+2,\cdots,n$. WLOG, assume $t^*\leq t$. For $i=1,2,\cdots,t$, write $\widetilde\epsilon_i=\mu_i-\bar \mu_t+\epsilon_i$ where $\bar \mu_t=\rho^*/\rho\mu_0+(1-\rho^*/\rho)\mu_1$. Write $\bar\epsilon_t=\frac{1}{t}\sum_{i=1}^{t}\epsilon_i$. Notice that 
\begin{equation*}
\begin{split}
&\frac{1}{t}\sum_{i=1}^{t}(\|\hat \epsilon_i\|^2 -\E\|\widetilde\epsilon\|^2)
= \frac{1}{t}\sum_{i=1}^{t}\left(\|\hat \epsilon_i\|^2 -\|\widetilde\epsilon_i\|^2\right) + \frac{1}{t}\sum_{i=1}^{t}(\|\widetilde \epsilon_i\|^2-\E\|\widetilde \epsilon_i\|^2)
\\&
= \frac{1}{t}\sum_{i=1}^{t}\left(\|\bar \epsilon_t\|^2 -2\langle\mu_i-\bar\mu_t+\epsilon_i,\bar\epsilon_t\rangle\right) 
\\&+ \frac{1}{t}\sum_{i=1}^{t}\left[\|\epsilon_i\|^2+\|\mu_i-\bar\mu_t\|^2+2\langle\mu_i-\bar\mu_t,\epsilon_i\rangle-\E\left(\|\epsilon_i\|^2+\|\mu_i-\bar\mu_t\|^2+2\langle\mu_i-\bar\mu_t,\epsilon_i\rangle\right)\right]
\\&
= -\|\bar \epsilon_t\|^2
+
\frac{1}{t}\sum_{i=1}^{t}\|\epsilon_i\|^2-\E\|\epsilon_i\|^2 +\frac{1}{t}\sum_{i=1}^{t}\left[2\langle\mu_i-\bar\mu_t,\epsilon_i\rangle\right]
\xrightarrow{P}0
\end{split}
\end{equation*}
where the last convergence follows from the fact that $\|\bar\epsilon_t\|=O_p(\frac{1}{\sqrt{t}})$. 
Similarly, we have
\begin{equation*}
\begin{split}
&\frac{1}{n-t}\sum_{i=t+1}^{n}(\|\hat \epsilon_i\|^2 -\E\|\widetilde\epsilon\|^2)
\xrightarrow{P}0
\end{split}
\end{equation*}
Notice that for $i=1,2,\cdots,t$, $\E_{F_0}\|\widetilde\epsilon\|^2=\E_{F_0}\left(\|\epsilon_i\|^2+\|\mu_i-\bar\mu_t\|^2+2\langle\mu_i-\bar\mu_t,\epsilon_i\rangle\right)=v_0+\frac{(1-\rho^*)^2}{\rho^2}\|\mu_0-\mu_1\|^2$, and
$\E_{F_1}\|\widetilde\epsilon\|^2=\E_{F_1}\left(\|\epsilon_i\|^2+\|\mu_i-\bar\mu_{n-t}\|^2+2\langle\mu_i-\bar\mu,\epsilon_i\rangle\right)=v_1+\frac{(\rho^*)^2}{\rho^2}\|\mu_0-\mu_1\|^2$. Thus, 
\begin{equation}\label{appendix_eq:proof_S1_power_deltarho}
\begin{split}
&\frac{(n-t)^2t}{n^3(t-1)}\sum_{i=1}^{t}\|\hat \epsilon_i\|^2+\frac{t^2(n-t)}{n^3(n-t-1)}\sum_{i=t+1}^{n}\|\hat \epsilon_i\|^2
\\&=
\bm 1(\rho^*\leq \rho)(1-\rho)\left(\rho^*(1-\rho)\frac{\rho-\rho^*}{\rho}\|\mu_0-\mu_1\|^2+(1-\rho)\rho^* v_0+(\rho-\rho^*+\rho\rho^*)v_1\right)
\\&+
\bm 1(\rho^*> \rho)\rho\left(\rho(\rho^*-\rho)\frac{1-\rho^*}{1-\rho}\|\mu_0-\mu_1\|^2+\rho(1-\rho^*)v_1+(\rho\rho^*-2\rho+1)v_0\right)
\doteq\delta(\rho)
\end{split}
\end{equation}

Under the alternative, from corollary 5.2.2 of \cite{tewes2017change}, if $\|\sqrt{n}(\mu_0-\mu_1)-\Delta\|\rightarrow0$, then
we have
$$
\frac{1}{n}\|\frac{n-t}{n}\sum_{i=1}^{t}\phi(y_i)-\frac{t}{n}\sum_{i=t+1}^{n}\phi(y_i)\|^2\xrightarrow{w}\|W^0(\rho)+\xi(\rho)\Delta\|^2
$$
where 
\begin{equation}\label{appendix_eq:proof_S1_power_xirho}
    \xi(\rho)=\rho(1-\rho^*)\, \text{if}\, \rho\leq\rho^* \quad\text{and}\quad \xi(\rho)=(1-\rho)\rho^*\, \text{if}\, \rho>\rho^*
\end{equation}

Thus, under the alternative,
$$
S_1
\xrightarrow{w}
\max_{\rho\in[\rho_0,\rho_1]}\left(\frac{\|W^0(\rho)+\xi(\rho)\Delta\|^2-\delta(\rho)}{\rho(1-\rho)}\right)
$$
Notice that 
$\|W^0(\rho)+\xi(\rho)\Delta\|^2=O_p(n\|\mu_0-\mu_1\|^2)$, $\delta(\rho)=o(n\|\mu_0-\mu_1\|^2)$. Thus, when
$\sqrt{n}\|\mu_0-\mu_1\|\rightarrow\infty$,
the desired conclusion holds.
\end{proof}

\subsection{Proof of Theorem \ref{thm:S2_power}}
Theorem \ref{thm:S2_power} is a direct consequence of the following theorem.
\begin{theorem}[Alternative distribution for $S_2$]\label{thm:S2_alternative}
	Under the alternative, if (1) $d$ is a semi-metric of negative type, (2) $d$ is a semi-metric of negative type, (3) $\sqrt{n}(v_0-v_1)\rightarrow \Delta_v^{(2)}$ and $\sqrt{n}\|\mu_0-\mu_1\|^2\rightarrow\Delta_{\mu}^{(2)}$, then 
	$$S_2\xrightarrow{w}\max_{\rho\in[\rho_0,\rho_1]}\left(\frac{|G+\Delta^{(2)}|}{\sqrt{\rho(1-\rho)}}\right)$$
	where $G$ is some Gaussian process and 
	$$
	\Delta^{(2)}=
	\begin{cases}
	\frac{1}{s_n}\rho^*(1-\rho)\left(\Delta_v^{(2)}+\frac{\rho-\rho^*}{\rho}\Delta_{\mu}^{(2)}\right),\rho\geq\rho^*\\
	\frac{1}{s_n}(1-\rho^*)\rho\left(\Delta_v^{(2)}-\frac{\rho^*-\rho}{1-\rho}\Delta_{\mu}^{(2)}\right), \rho<\rho^*
	\end{cases}
	$$
	where $s_n^2=\rho^*Var_{F_0}\left(\|(1-\rho^*)(\mu_0-\mu_1)+\epsilon\|^2\right) +(1-\rho^*)Var_{F_0}\left(\|\rho^*(\mu_1-\mu_0)+\epsilon\|^2\right)$. 
\end{theorem}
\begin{proof}
First, from Lemma \ref{prop:proof_of_S2_null}, we know that
$$
\hsn^2\xrightarrow{P}s_n^2=\rho^*Var_{F_0}(\|\widetilde\epsilon_i\|^2)+(1-\rho^*)Var_{F_0}(\|\widetilde\epsilon_i\|^2)
$$
where $\widetilde\epsilon_i=\mu_i+\epsilon_i-\frac{1}{n}\sum_{i=1}^{n}\mu_i$. Thus $s_n^2$ is some strictly positive constant.

Define $\epsilon_i^*=\mu_i+\epsilon_i-\frac{1}{t}\sum_{i=1}^t\mu_i$ for all $i=1,2,\cdots, t$ and $\epsilon_i^*=\mu_i+\epsilon_i-\frac{1}{n-t}\sum_{i=t+1}^n\mu_i$ for all $i=1,2,\cdots, t$. Notice that
\begin{equation*}\small
\begin{split}
&\frac{1}{\sqrt{n}\hsn}\frac{t(n-t)}{2n}\left(\db-\dbb\right)
=\frac{1}{\sqrt{n}}\frac{s_n}{\hsn}\frac{1}{s_n}\left[ \frac{t}{t-1}\frac{n-t}{n}\sum_{i=1}^t\|\epsilon_i^*\|^2-\frac{n-t}{n-t-1}\frac{t}{n}\sum_{i=t+1}^n\|\epsilon_i^*\|^2\right]
\\&+\frac{1}{\hsn}\frac{1}{\sqrt{n}}\frac{t}{t-1}\frac{n-t}{n}\sum_{i=1}^t\left(\|z_i-\bar z_t\|^2-\|\epsilon_i^*\|^2\right)-\frac{1}{\hsn}\frac{1}{\sqrt{n}}\frac{n-t}{n-t-1}\frac{t}{n}\sum_{i=t+1}^n\left(\|z_i-\bar z_{n-t}\|^2-\|\epsilon_i^*\|^2\right)
\\&=U_1+U_2+U_3
\end{split}
\end{equation*}	
For $U_1$,
we have
\begin{equation}\label{appendix_eq:s2_power_u1}
\begin{split}
U_1=&
\frac{1}{\sqrt{n}}\frac{s_n}{\hsn}\frac{1}{s_n}\left[ \frac{t}{t-1}\frac{n-t}{n}\sum_{i=1}^t\left(\|\epsilon_i^*\|^2-\E\|\epsilon_i^*\|^2\right)-\frac{n-t}{n-t-1}\frac{t}{n}\sum_{i=t+1}^n\left(\|\epsilon_i^*\|^2-\E\|\epsilon_i^*\|^2\right)\right] 
\\&+ \frac{1}{\sqrt{n}\hsn}t^*(1-\rho)\left(v_0-v_1+\frac{\rho-\rho^*}{\rho}\|\mu_0-\mu_1\|^2\right)\bm 1(\rho\geq\rho^*)
\\&+ \frac{1}{\sqrt{n}\hsn}(n-t^*)\rho\left(v_0-v_1-\frac{\rho^*-\rho}{1-\rho}\|\mu_0-\mu_1\|^2\right)\bm 1(\rho<\rho^*)
\end{split}
\end{equation}
Since $d$ is bounded, and $\hsn$ converges in probability to some strictly positive constant, we know 
$$
\frac{1}{\sqrt{n}}\frac{s_n}{\hsn}\frac{1}{s_n}\left[ \frac{t}{t-1}\frac{n-t}{n}\sum_{i=1}^t\left(\|\epsilon_i^*\|^2-\E\|\epsilon_i^*\|^2\right)-\frac{n-t}{n-t-1}\frac{t}{n}\sum_{i=t+1}^n\left(\|\epsilon_i^*\|^2-\E\|\epsilon_i^*\|^2\right)\right] 
\xrightarrow{w}
G
$$
where $G$ is some Gaussian process.  

For $U_2$, we have
\begin{equation*}
\begin{split}
U_2&=\frac{1}{\hsn}\frac{1}{\sqrt{n}}\frac{t}{t-1}\frac{n-t}{n}\sum_{i=1}^t\left(\|z_i-\bar z_t\|^2-\|\epsilon_i^*\|^2\right)
=
\frac{1}{\hsn}\frac{1}{\sqrt{n}}\frac{t}{t-1}\frac{n-t}{n}\sum_{i=1}^t\left(\|z_i-\bar \mu_t-\bar\epsilon_t\|^2-\|z_i-\bar\mu_t\|^2\right)
\\&=\frac{1}{\hsn}\frac{1}{\sqrt{n}}\frac{t^2}{t-1}\frac{n-t}{n}(-\|\bar\epsilon_t\|^2)
\xrightarrow{P}0
\end{split}
\end{equation*}
because $\|\bar \epsilon_{t}\|=O_p(1/\sqrt{t})$. Similarly we have $U_3\xrightarrow{P}0$.

So under the alternative,
$$S_2\xrightarrow{w}\max_{\rho\in[\rho_0,\rho_1]}\left(\frac{|G+\Delta|}{\sqrt{\rho(1-\rho)}}\right)$$
where $G$ is some Gaussian process and (suppose $\sqrt{n}(v_0-v_1)\rightarrow \Delta_v$ and $\sqrt{n}\|\mu_0-\mu_1\|^2\rightarrow\Delta_{\mu}$)
$$
\Delta=
\begin{cases}
\frac{1}{s_n}\rho^*(1-\rho)\left(\Delta_v+\frac{\rho-\rho^*}{\rho}\Delta_{\mu}\right),\rho\geq\rho^*\\
\frac{1}{s_n}(1-\rho^*)\rho\left(\Delta_v-\frac{\rho^*-\rho}{1-\rho}\Delta_{\mu}\right), \rho<\rho^*
\end{cases}
$$
Thus, when $\mu_0=\mu_1$, $\sqrt{n}|v_0-v_1|\rightarrow\infty$, we have $\Delta_v\rightarrow\infty$ and $\Delta_{\mu}=0$.
\begin{equation*}\small
\begin{split}
&\lim_{n\rightarrow\infty}\P_{H_A}\left(S_2>q^{(2)}_{\alpha}\right)
\\&=\lim_{n\rightarrow\infty}\P_{H_A}\left(\frac{1}{\sqrt{\rho(1-\rho)}}|G+\frac{1}{\sqrt{n}s_n}t^*(1-\rho)(v_0-v_1)\bm 1(\rho\geq\rho^*)
+ \frac{1}{\sqrt{n}s_n}(n-t^*)\rho(v_0-v_1)\bm 1(\rho<\rho^*)|>q^{(2)}_{\alpha}\right)
\\&=1
\end{split}
\end{equation*}
Thus, we can get the desired conclusion.
\end{proof}

\subsection{Derivation for Higher Order Correction}\label{appendix_sec:S2_correction}
First we derive Equation \eqref{eq:S2_tilde}:
From Proof of Lemma \ref{lemma:S2_null_unscaled}, we know that 
\begin{equation*}
    \begin{split}
        \frac{\sqrt{n}\rho(1-\rho)}{2\hsn}T_2=U_1+U_2+U_3
    \end{split}
\end{equation*}
where
$$U_1=\frac{1}{\sqrt{n}s_n}\left[\frac{t}{t-1}\frac{n-t}{n}\sum_{i=1}^t\|\epsilon_i\|^2-\frac{n-t}{n-t-1}\frac{t}{n}\sum_{i=t+1}^n\|\epsilon_i\|^2\right]\xrightarrow{w}W^0(\rho)$$ $$U_2=-\frac{1}{\sqrt{n}\hsn}\frac{t}{t-1}\left(1-\frac{t}{n}\right)\frac{1}{t}\sum_{i=1}^t\|\epsilon_i\|^2
\approx -\frac{1}{\sqrt{n}\hsn}\frac{t}{t-1}\left(1-\frac{t}{n}\right)\E\|\epsilon\|^2$$
$$U_3=-\frac{1}{\sqrt{n}\hsn}\frac{n-t}{n-t-1}\frac{t}{n}\frac{1}{n-t}\sum_{i=t+1}^n\|\epsilon_i\|^2
\approx -\frac{1}{\sqrt{n}\hsn}\frac{n-t}{n-t-1}\frac{t}{n}\E\|\epsilon\|^2$$
Thus, $U_2$ and $U_3$ are of order $O_p(\frac{1}{\sqrt{n}})$. By replacing the true mean with the estimated sample version, we can get Equation \eqref{eq:S2_tilde}, which cancels the $O_p(\frac{1}{\sqrt{n}})$ term from $U_2$ and $U_3$.

Equation \eqref{eq:S2_higher_order_correctness} corrects for the $O_p(\frac{1}{\sqrt{n}})$ coming from $U_1$:
Write $Z\left(\frac{t}{n}\right)=\frac{U_1}{\sqrt{\rho(1-\rho)}}$, following \cite{chen2015graph}, we have
\begin{equation*}
    \begin{split}
        &\P\left(\max_{n_0\leq t\leq n_1}Z\left(\frac{t}{n}\right)>b\right)
        \\&=\frac{1}{b}\sum_{n_0\leq t\leq n_1}\int_{x=0}^{\infty}p\left(Z\left(\frac{t}{n}\right)=b+\frac{1}{b}x\right)\P\left(\max_{n_0\leq s\leq n_1}Z\left(\frac{s}{n}\right)<b\C Z\left(\frac{t}{n}\right)=b+\frac{1}{b}x\right)dx
    \end{split}
\end{equation*}
We approximate $p\left(Z\left(\frac{t}{n}\right)=b+\frac{1}{b}x\right)$ using 3$^{rd}$ order Edgeworth Expansion and approximate $\P\left(\max_{n_0\leq s\leq n_1}Z\left(\frac{s}{n}\right)<b\C Z\left(\frac{t}{n}\right)\right)$ using a random walk.

Notice that $Z\left(\frac{t}{n}\right)$ is a sum of independent, non-identical distributed random variables, so we can apply Edgeworth Expansion and get when $b\rightarrow\infty$, $x^2/(2b^2)$ is negligible to $x$ and $x/b$ is neglibible to $b$, so let $V$ be the skewness of $Z\left(\frac{t}{n}\right)$, then
\begin{equation*}
    \begin{split}
        &p\left(Z\left(\frac{t}{n}\right)=b+\frac{1}{b}x\right)
        \approx \phi\left(b+\frac{1}{b}x\right)+\frac{1}{\sqrt{n}}\frac{1}{6}V\left[\left(b+\frac{x}{b}\right)^3-3\left(b+\frac{x}{b}\right)\right]\phi\left(b+\frac{x}{b}\right)
        \\= &\phi(b)e^{-x^2/(2b^2)-x}\left[1+\frac{1}{6\sqrt{n}}V\left(b+\frac{x}{b}\right)\left(\left(b+\frac{x}{b}\right)^2-3\right)\right]
        \\\approx & \phi(b)e^{-x}\left[1+\frac{1}{6\sqrt{n}}Vb\left(b^2-3\right)\right]
    \end{split}
\end{equation*}

To approximate $\P\left(\max_{n_0\leq s\leq n_1}Z\left(\frac{s}{n}\right)<b\C Z\left(\frac{t}{n}\right)\right)$, notice that 
\begin{equation*}
    \begin{split}
        b\left(Z\left(\frac{s}{n}\right)-Z\left(\frac{t}{n}\right)\right)\C Z\left(\frac{t}{n}\right)=b+\frac{x}{b}\sim N\left(-f_{t/n,-}'(0)|\frac{s-t}{n}|b^2,2f_{t/n,-}'(0)|\frac{s-t}{n}|b^2\right)
    \end{split}
\end{equation*}
where 
\begin{equation*}
    \begin{split}
        f_{x,-}'(0)=\frac{\partial}{\partial \delta}corr(Z\left(0\right), Z\left(\delta\right))\C_{\delta=0}
        =\frac{1}{2x(1-x)}
    \end{split}
\end{equation*}

So let $W_m^{(t)}$ be a random walk with $W_1^{(t)}\sim N(\mu^{(t)},(\sigma^2)^{(t)})$ where $\mu^{(t)}=\frac{1}{n}f_{t/n,-}(0)b^2$ and $(\sigma^2)^{(t)}=2\mu^{(t)}$. We have
$$
\P\left(\max_{n_0\leq s\leq n_1}Z\left(\frac{s}{n}\right)<b\C Z\left(\frac{t}{n}\right)\right)
\approx \P\left(\max_{n_0\leq s<t}-W^{(t)}_{t-s}<-x\right)
\approx \P\left(\min_{m\geq 1}W^{(t)}_m>x\right)
$$
Combining the above, we have
\begin{equation*}
    \begin{split}
        &\P\left(\max_{n_0\leq t\leq n_1}Z\left(\frac{t}{n}\right)>b\right)
        \approx \frac{\phi(b)}{b}\left[1+\frac{1}{6\sqrt{n}}Vb\left(b^2-3\right)\right]\sum_{n_0\leq t\leq n_1}\int_{x=0}^{\infty}e^{-x}\P\left(\min_{m\geq 1}W^{(t)}_m>x\right)dx
        \\&=\frac{\phi(b)}{b}\left[1+\frac{1}{6\sqrt{n}}Vb\left(b^2-3\right)\right]\sum_{n_0\leq t\leq n_1}(-f_{t/n,-}(0))b^2\nu\left(b\sqrt{-\frac{2}{n}f_{t/n,-}(0)}\right)
        \\&=b\phi(b)\int_{\rho_)}^{\rho_1}\left[1+\frac{1}{6\sqrt{n}}Vb\left(b^2-3\right)\right](-f_{x,-}(0))\nu\left(b\sqrt{-2f_{x,-}(0)}\right)dx
    \end{split}
\end{equation*}
And calculation shows (recall that $v=\E\|\epsilon\|^2$)
$$V=\frac{n-2t}{\sqrt{t(n-t)}}\times\text{Skewness of }\|\epsilon\|^2
=\frac{n-2t}{\sqrt{t(n-t)}}\times\frac{\E\left[\|\epsilon\|^6-3\|\epsilon\|^4v+3\|\epsilon\|^2v^2-v^3\right]}{\left[\E\left(\|\epsilon\|^4-2\|\epsilon\|^2v+v^2\right)\right]^{1.5}}$$
From symmetry of $Z\left(\frac{t}{n}\right)$, we know $\P\left(\max_{n_0\leq t\leq n_1}|Z\left(\frac{t}{n}\right)|>b\right)=2\P\left(\max_{n_0\leq t\leq n_1}Z\left(\frac{t}{n}\right)>b\right)$. By replacing the corresponding true moments by the sample version in $K$, we get Equation \eqref{eq:S2_higher_order_correctness}.

\subsection{Proof of Proposition \ref{prop:S1_S2_ind}}
First notice that 
$$\lim_{n\rightarrow\infty}nT_1(t)=\lim_{n\rightarrow\infty}n\|\bar\phi(y)_t-\bar\phi(y)_{n-t}\|^2+C_t$$ 
where $C_t$ is some constant depending on $t$. And
\begin{equation*}
    \begin{split}
        &\lim_{n\rightarrow\infty}\sqrt{n}T_2(t)
        \\&=\lim_{n\rightarrow\infty}\sqrt{n}\left[\frac{1}{t}\sum_{i=1}^t\left(\|\phi(y_i)-\bar\phi(y_i)_t\|^2-\E\|\phi(y_i)-\bar\phi(y_i)_t\|^2\right)\right]
        \\&-\sqrt{n}\left[\frac{1}{n-t}\sum_{i=t+1}^n\left(\|\phi(y_i)-\bar\phi(y_i)_t\|^2-\E\|\phi(y_i)-\bar\phi(y_i)_t\|^2\right)\right]+\widetilde C_t
    \end{split}
\end{equation*}
where $\widetilde C_t$ is a constant depending on $t$. 
Thus $\sqrt{n}\left(\bar\phi_l(y)_t-\E\bar\phi_l(y)_t\right)$, $\sqrt{n}\left(\bar\phi_l(y)_{n-t}-\E\bar\phi_l(y)_{n-t}\right)$, $\frac{1}{\sqrt{n}}\sum_{i=1}^t\left[\|\phi(y_i)-\bar\phi(y_i)_t\|^2-\E\|\phi(y_i)-\bar\phi(y_i)_t\|^2\right]$ and $\frac{1}{\sqrt{n}}\sum_{i=t+1}^n\left[\|\phi(y_i)-\bar\phi(y_i)_t\|^2-\E\|\phi(y_i)-\bar\phi(y_i)_t\|^2\right]$ are all asymptotically Gaussian with mean 0, we only need to check that their covariance converges to 0. Since our data are i.i.d, we only need to check that the pair: $\sqrt{n}\left(\bar\phi_l(y)_t-\E\bar\phi_l(y)_t\right)$ and $\frac{1}{\sqrt{n}}\sum_{i=1}^t\|\phi(y_i)-\bar\phi(y_i)_t\|^2-\E\|\phi(y_i)-\bar\phi(y_i)_t\|^2$,  $\sqrt{n}\left(\bar\phi_l(y)_{n-t}-\E\bar\phi_l(y)_{n-t}\right)$ and $\frac{1}{\sqrt{n}}\sum_{i=t+1}^n\|\phi(y_i)-\bar\phi(y_i)_t\|^2-\E\|\phi(y_i)-\bar\phi(y_i)_t\|^2$ are asymptotically uncorrelated for any $l$.
Since under the null, $\E\phi_l(y_i)=0$,
\begin{equation}
    \begin{split}
        &Cov\left(\sqrt{n}\left(\bar\phi_l(y)_t-\E\bar\phi_l(y)_t\right),\frac{1}{\sqrt{n}}\sum_{i=1}^t\|\phi(y_i)-\bar\phi(y_i)_t\|^2-\E\|\phi(y_i)-\bar\phi(y_i)_t\|^2\right)
        \\&=\sum_{i=1}^t\E\|\phi(y_i)-\bar\phi(y)_t\|^2\left(\bar\phi_l(y)_t-\E\bar\phi_l(y)_t\right)
        =\left(1-\frac{1}{t}\right)\E\|\phi(y_i)\|^2\phi_l(y_i)
    \end{split}
\end{equation}
Thus, we get the desired conclusion.

\subsection{Theoretical Guarantees for $S_3$}\label{sec_appendix:S3_guarantees}
\begin{corollary}[Asymptotic null distribution for $S_3$]\label{thm:S3_null}
	Under $H_0$, if distance $d$ satisfies 
	
	(1) $d$ is a semi-metric of negative type, 
	
	(2) $\E d(y_i,y_j)\leq M^2,\,\forall i,j$, 
	
	(3) $\E_y\left(\E_{y'}d(y,y')-\frac{1}{2}\E_{y,y'}d(y,y')\right)^{2+\delta}<+\infty$ for some $\delta>0$,
	
	(4) $\E_y|\E_{y'}d(y,y')-\E_{y,y'}d(y,y')|^{2+\delta'}<+\infty$ for some $\delta'>0$,
	\\then as $n\rightarrow\infty$,
	\begin{equation}\label{eq:S3_asymp_null}
	S_3\xrightarrow{w} 
	\max_{\rho_0\leq \rho\leq \rho_1}\frac{1}{{\rho(1-\rho)}}\left(W^0(\rho)\right)^2
	\end{equation}
	where $W^0(\cdot)$ is a Brownian bridge.
\end{corollary}
\begin{proof}
	Notice that Assumption (2) in Theorem \ref{thm:S2_null} is equivalent to
	$$
	2\sum_l\lambda_l\leq M^2
	$$
	Assumption (3) in Theorem \ref{thm:S1_null} is equivalent to
	$$
	\sum_l\lambda_l^2\leq+\infty
	$$
	Thus, assumption (2) in Theorem \ref{thm:S2_null} implies assumption (3) in Theorem \ref{thm:S1_null}. Then, Thoerem \ref{thm:S3_null} is a direct consequence of Theorem \ref{thm:S1_null} and Theorem \ref{thm:S2_null}.
\end{proof}

\begin{corollary}[Localization Consistency for $S_3$]\label{thm:localization_S3}
	Under $H_A$, suppose $d(y_i,y_j)\leq M^2<\infty,\,\forall i,j$, and $d$ is a semi-metric of negative type, then
	$$
	\left|\frac{\widehat \tau-\tau^*}{n}\right|=o_p(1)
	$$
	where $\widehat \tau$ is the estimated change point using statistics $S_3$.
\end{corollary}
\begin{proof}
From the proof of Theorem \ref{thm:S1_power} (Section \ref{appendix_sec:proof_S1_power}), we know that 
$$
T_1\xrightarrow{w}\frac{\|\bm W^0(\rho)+\xi(\rho)\Delta\|^2-\delta(\rho)}{n\rho^2(1-\rho)^2}
=\left(\frac{\xi\left(\rho\right)}{\rho(1-\rho)}\right)^2\left\|\mu_0-\mu_1\right\|^2
$$
It is obvious that the maximum of $\left(\frac{\xi\left(\rho\right)}{\rho(1-\rho)}\right)^2$ is obtained at $\rho=\rho^*$. From the Argmax Theorem, we know that $ \frac{\hat \tau}{n}-\frac{\tau^*}{n}=o_p(1)$.
\end{proof}

\begin{corollary}[Power for $S_3$]\label{thm:S3_power}
	If (1) $d(y_i,y_j)\leq M^2<\infty,\,\forall i,j$, (2) $d$ is a semi-metric of negative type,  then
	\begin{equation*}
	\P_{H_A}\left(S_3>q^{(2)}_{\alpha}\right)\rightarrow 1,\quad n\rightarrow\infty
	\end{equation*}
	if either $\sqrt{n}\|\mu_0-\mu_1\|^2\rightarrow \infty$ or $\sqrt{n}|v_0-v_1|\rightarrow \infty$.
\end{corollary}
\begin{proof}
	Corollary \ref{thm:localization_S3} is a direct consequence of Theorem  \ref{thm:S1_alternative} and Theorem \ref{thm:S2_alternative}.
\end{proof}

\subsection{Proof of Theorem \ref{thm:multiple_CP}}
\begin{proof}
For $S_1$, using exactly the same techniques as in Theorem \ref{thm:S1_null}, it is easy to show that 
$$
 T_1^{l,u}(t)\rightarrow
 \left\|\frac{1}{t-l}\sum_{i=l}^t \E\phi(y_i)
    -\frac{1}{u-t}\sum_{i=t+1}^u \E\phi(y_i)\right\|^2\quad\text{uniformly}
$$
From Lemma 3.2 of \cite{rice2019consistency}, we know that $\arg\max_{l\leq k\leq u}\|\frac{1}{t-l}\sum_{i=l}^t \E\phi(y_i)
    -\frac{1}{u-t}\sum_{i=t+1}^u \E\phi(y_i)\|^2\in\mathcal{D}$. Then from the uniform convergence of $T_1^{l,u}(t)$ and the argmax Theorem, we know that the desired conclusion holds.
    
For $S_2$, notice that
$$
 T_2^{l,u}(t)\rightarrow
 \left|\frac{1}{t-l}\sum_{i=l}^t \|\phi(y_i)-\E\phi(y_i)\|^2
    -\frac{1}{u-t}\sum_{i=t+1}^u \|\phi(y_i)-\E\phi(y_i)\|^2\right|\quad\text{uniformly}
$$
Then, similar as the case for $S_1$, the conclusion for $S_2$ follows directly.
\end{proof}

\bibliographystyle{apa}
\bibliography{changepoint}